\documentclass[a4paper,twocolumn,11pt,accepted=2022-09-30]{quantumarticle}

\pdfoutput=1
\usepackage[utf8]{inputenc}
\usepackage[english]{babel}
\usepackage[T1]{fontenc}
\usepackage{amsmath}
\usepackage{hyperref}

\usepackage{tikz}
\usepackage{lipsum}

\usepackage{amsmath}
\usepackage{amssymb}
\usepackage{graphicx}
\usepackage{esint}
\usepackage{color}
\usepackage{hyperref}
\usepackage{mathtools}
\usepackage{stackengine}
\usepackage{caption}
\usepackage[numbers,sort&compress]{natbib}

\makeatletter
\newcommand{\removelatexerror}{\let\@latex@error\@gobble}
\makeatother

\usepackage[ruled]{algorithm2e}

\newtheorem{theorem}{Theorem}
\newtheorem{definition}{Definition}
\newtheorem{lemma}{Lemma}

\newtheorem{proposition}[theorem]{Proposition}

\newenvironment{proof}{\noindent \textbf{{Proof~} }}{\hfill $\blacksquare$}

\begin{document}
	
	\title{Quantum variational learning for quantum error-correcting codes}

	\author{Chenfeng Cao}
	\affiliation{Department of Physics, The Hong Kong University of Science and Technology, Clear Water Bay, Kowloon, Hong Kong, China}
	\orcid{0000-0001-5589-7503}
	
	\author{Chao Zhang}
	\affiliation{Department of Physics, The Hong Kong University of Science and Technology, Clear Water Bay, Kowloon, Hong Kong, China}
	\orcid{0000-0002-2093-7496}
	
	\author{Zipeng Wu}
	\affiliation{Department of Physics, The Hong Kong University of Science and Technology, Clear Water Bay, Kowloon, Hong Kong, China}
	\orcid{0000-0002-9349-1325}
	
	\author{Markus Grassl}
	\affiliation{International Centre for Theory of Quantum Technologies, University of Gdansk, 80-309 Gdansk, Poland}
	\orcid{0000-0002-3720-5195}
	
	\author{Bei Zeng}
	\email{zengb@ust.hk}
	\affiliation{Department of Physics, The Hong Kong University of Science and Technology, Clear Water Bay, Kowloon, Hong Kong, China}
	\orcid{0000-0003-3989-4948}

	\begin{abstract}	 
	 Quantum error correction is believed to be a necessity for large-scale fault-tolerant quantum computation. In the past two decades, various constructions of quantum error-correcting codes (QECCs) have been developed, leading to many good code families. However, the majority of these codes are not suitable for near-term quantum devices. Here we present VarQEC, a noise-resilient variational quantum algorithm to search for quantum codes with a hardware-efficient encoding circuit. The cost functions are inspired by the most general and fundamental requirements of a QECC, the Knill-Laflamme conditions. Given the target noise channel (or the target code parameters) and the hardware connectivity graph, we optimize a shallow variational quantum circuit to prepare the basis states of an eligible code. In principle, VarQEC can find quantum codes for any error model, whether additive or non-additive, degenerate or non-degenerate, pure or impure. We have verified its effectiveness by (re)discovering some symmetric and asymmetric codes, e.g., $((n,2^{n-6},3))_2$ for $n$ from 7 to 14. We also found new $((6,2,3))_2$ and $((7,2,3))_2$ codes that are not equivalent to any stabilizer code, and extensive numerical evidence with VarQEC suggests that a $((7,3,3))_2$ code does not exist. Furthermore, we found many new channel-adaptive codes for error models involving nearest-neighbor correlated errors. Our work sheds new light on the understanding of QECC in general, which may also help to enhance near-term device performance with channel-adaptive error-correcting codes.
	\end{abstract}

	\section{Introduction}\label{Introduction}
	
	Fault-tolerant quantum computers promise to solve some computational problems much faster than classical machines, such as quantum chemistry simulation~\cite{jones2012faster}, prime factorization~\cite{shor1999polynomial}, solving linear systems of equations~\cite{harrow2009quantum}. However, quantum information carried by current noisy intermediate-scale quantum (NISQ) systems is highly fragile and can be easily altered by the environment. The aforementioned tasks are so far out of reach.
	
	The most promising technique to maintain coherence and protect the quantum information from noise is \textit{quantum error-correcting codes}~\cite{shor1995scheme,gottesman1997stabilizer, lidar2013quantum, zeng2019quantum, girvin2021introduction}. The main idea of quantum error correction is to encode the low-dimensional quantum state in a larger system such that errors occurring during the computation can be corrected due to the physical redundancy. As long as the noise rate $p$ is below a specific threshold, QECCs can correct the error and reduce the error probability from $\mathcal{O}(p)$ to higher orders. In recent years, the intrinsic connections between QECCs and other areas of physics, such as quantum gravity~\cite{pastawski2015holographic}, have also been noticed. 
	
	Knill and Laflamme devised sufficient and necessary conditions (known as the \textit{Knill-Laflamme conditions}) for quantum error correction~\cite{knill1997theory}. In principle, we can find any QECC as long as we find solutions to the Knill-Laflamme conditions. However, solving these systems of equations is extremely difficult in the general case. Therefore many open problems in this field remain unsolved, e.g., do all degenerate QECCs obey the Hamming bound? Which QECC has the highest error threshold? In practice, researchers usually analyze QECCs under the Pauli framework and have developed various QECC families, such as surface codes~\cite{kitaev1997quantum, fowler2012surface}, Calderbank-Shor-Steane (CSS) codes~\cite{calderbank1996good, steane1996multiple}, stabilizer codes~\cite{gottesman1997stabilizer}, codeword stabilized (CWS) codes~\cite{cross2008codeword,chuang2009codeword}, quantum low-density parity-check codes~\cite{breuckmann2021quantum, panteleev2021asymptotically}.

	Up till now, no logical qubit/operation with useful fidelity was realized in experiments since current gate noise rates are still much larger than the requirements. Very recently, Egan $\textit{et al.}$~\cite{egan2021fault} prepared a Bacon–Shor logical qubit with 13 trapped-ion qubits and demonstrated a logical single-qubit Clifford gate. Further, Postler $\textit{et al.}$~\cite{postler2021demonstration} demonstrated a logical $T$-gate based on the 7-qubit color code. However, the fidelities of these state-of-the-art logical qubits are even lower than those of the physical qubits. From a theoretical perspective, the codes used in those experiments are not device-tailored and may not be optimal for the system. The noise channels on different physical platforms differ significantly~\cite{dawson2006noise, wilen2021correlated, guo2021testing}. Symmetric QECC constructions under the Pauli framework can not be directly adapted to non-Hermitian/non-unitary noise channels. It is highly desirable to design asymmetric or channel-adaptive QECCs with a hardware-efficient encoder. Such device-tailored codes can protect logical information more efficiently.

	Besides analytical constructions, researchers have been trying to find QECCs with computational methods for a long time. Refs.~\cite{yu2007graphical, hu2008graphical, chuang2009codeword, PhysRevA.101.042307} designed classical algorithms for finding quantum codes associated with graphs. Ref.~\cite{li2017fault} used numerical greedy search for finding stabilizer codes. These algorithms, however, cannot find arbitrary codes and are extremely time-consuming. With the popularity of artificial intelligence, researchers also started to design and optimize quantum codes with neural networks~\cite{fosel2018reinforcement, baireuther2018machine, andreasson2019quantum, nautrup2019optimizing}. These classical black-box models perform pretty well for certain problems. In this work, we add a new general method to this toolbox. We devise a hybrid quantum-classical algorithm called VarQEC for finding quantum error-correcting codes. The cost functions therein are based on the Knill-Laflamme conditions. We iteratively update the parameters in a variational quantum circuit (VQC) with stochastic gradient descent. If the final cost functions are sufficiently small, we obtain an approximate quantum code whose inaccuracy is bounded. Compared with the classical iterative algorithm introduced in Ref.~\cite{reimpell2005iterative}, our method yields the encoding circuit, not merely the encoding isometry. After finding a QECC and its encoder, the decoding operation can be found via various methods like semidefinite programming~\cite{fletcher2007optimum}, convex optimization~\cite{fletcher2007channel}, or classical/quantum machine learning~\cite{sweke2020reinforcement, liu2019neural, locher2022quantum}.

	VarQEC allows for non-Hermitian or non-unitary errors and is surprisingly effective. We numerically verify its effectiveness up to 14 qubits. For symmetric Pauli errors, we successfully rediscover many good quantum codes, e.g., $((5,2,3))_2$, $((5,6,2))_2$, $((6,2,3))_2$, $((7,2,3))_2$, $((8,8,3))_2$, $((9,8,3))_2$, $((10,2^4,3))_2$, $((11,2^5,3))_2$, $((12,2^6,3))_2$, $((13,2^7,3))_2$, $((14,2^8,3))_2$, $((10,4,4))_2$. Some of the $((6,2,3))_2$, $((7,2,3))_2$ codes we find are not locally equivalent to any CWS code. It is an open question of whether there is a quantum code with parameters $((7,3,3))_2$, our numerical evidence suggests that it is non-existent. Then we apply VarQEC to search for asymmetric codes (which detect more Pauli-$X$/$Y$ errors than Pauli-$Z$ errors or vice versa) and make new discoveries. Furthermore, we search for channel-adaptive codes for nearest-neighbor collective amplitude damping and nearest-neighbor collective phase-flips, and find eligible new codes with a hardware-efficient encoding circuit for various connectivity graphs. Since VarQEC is capable to find a QECC with the shallowest possible encoding circuit, it is promising to design codes with sufficient fidelity that can be tested and implemented on near-term devices. Although only relatively small systems were investigated in this paper, hierarchical concatenation can construct good quantum codes with large code lengths and distances~\cite{gottesman1997stabilizer, knill1996concatenated, grassl2009generalized}.

	The paper is organized as follows. In Sec.~\ref{Sec:QEC}, we introduce some background of quantum error correction. In Sec.~\ref{Sec:theory}, we introduce our cost functions and present propositions to support our definitions. In Sec.~\ref{Sec:algorithm}, we explain the VarQEC algorithm in detail. In Sec.~\ref{Sec:results}, we show quantum codes (re)discovered thereby, including symmetric, asymmetric, and channel-adaptive codes for nearest-neighbor collective amplitude damping and nearest-neighbor collective phase-flips. Sec.~\ref{Sec:noise-resilience} discusses the noise resilience feature of VarQEC. Sec.~\ref{Sec:BP} discusses the barren plateaus and the noise-induced barren plateaus in VarQEC optimization. In Sec.~\ref{Sec:Experiment}, we verify our algorithm by an experiment on an IBM quantum device. The conclusions and future directions are summarized and discussed in Sec.~\ref{Conclusion}. The appendices give some proof details, a discussion on overparameterization, an alternative variational ansatz, some non-CWS quantum codes, and a list of the quantum weight enumerators of quantum codes discovered by VarQEC.

	\section{Preliminaries}\label{Sec:QEC}
	In classical computation and communication, redundancy is added when encoding a message such that the errors can be detected and corrected. Although each bit may flip with some probability, the encoded message can be recovered with high probability. The philosophy behind quantum error correction is the same. We use several low-fidelity physical qudits (e.g., qubits) to encode the logical quantum information redundantly and nonlocally. Then quantum errors can be detected through syndrome measurements and corrected through a unitary operation. A $q$-ary QECC $\mathcal{C}$ is a $K$-dimensional subspace of the $q^n$-dimensional Hilbert space $(\mathbb{C}^{q})^{\otimes n}$, where $n$ is the number of physical qudits (referred to as the \textit{code length}). For qubit systems, $q=2$, $\mathcal{C} \subset (\mathbb{C}^{2})^{\otimes n}$. When $K=1$, the code is a fixed quantum state without computational use. Throughout this paper, we only discuss $K \geq 2$.

	Knill and Laflamme developed a general theory of quantum error correction. They obtained the sufficient and necessary conditions for an exact QECC~\cite{knill1997theory}: a quantum code with orthonormal basis states $\{|\psi_j\rangle\}$ corrects the error set $\mathcal{E}=\{E_{\alpha}\}$ if and only if
	\begin{equation}\label{eq: KL}
		P_c E_{\alpha}^{\dagger} E_{\beta} P_c=\lambda_{\alpha \beta} P_c,
	\end{equation}
	holds for all $E_{\alpha}, E_{\beta} \in \mathcal{E}$. Here, $P_c = \sum_j |\psi_j\rangle\langle\psi_j|$ is the orthogonal projector onto the code space, and each $\lambda_{\alpha \beta}$ is a complex number. Moreover, we say the quantum code is \textit{non-degenerate} if the matrix $\lambda_{\alpha \beta}$ has full rank~\cite{gottesman2002introduction}.
	
	We can understand these conditions intuitively. When $i \neq j$, $\langle\psi_i|E_{\alpha}^{\dagger} E_{\beta}|\psi_j\rangle = 0$ for any error product $E_{\alpha}^{\dagger} E_{\beta}$.  This means orthogonal logical states remain orthogonal after the noise channel, the logical information is not corrupted. When $i = j$, $\langle\psi_j|E_{\alpha}^{\dagger} E_{\beta}|\psi_j\rangle = \lambda_{\alpha \beta}$ with $\lambda_{\alpha \beta}$ being a constant only determined by the error product. This indicates that the projections between subspaces induced by different errors are information-preserving, the errors have an orthogonal decomposition. Therefore, we can correct the error without knowing or destroying the quantum superposition state.

	The \textit{quantum error detection conditions} have a similar form: a quantum code with code space projector $P_c$ can detect the error set $\mathcal{E} = \{E_{\mu}\}$ if and only if
	\begin{equation}\label{eq: qec_detection}
		P_c E_{\mu} P_c=\lambda_{\mu} P_c
	\end{equation}
	holds for all $E_{\mu} \in \mathcal{E}$. 
	
	In experiments, most quantum errors are uncorrelated single-qudit errors. A natural measure of the capability of a QECC is the number of single-qudit errors that it can detect. This motivated the concept of ``code distance'': the distance of a QECC is the largest possible integer $d$ such that the code can detect any error non-trivially acting on at most $d-1$ qudits. Researchers usually denote the code parameters of a $q$-ary QECC with code length $n$, code dimension $K$, and code distance $d$ as $((n,K,d))_q$.
	
	Comparing the Knill-Laflamme conditions and quantum error detection conditions, we know that a distance-$d$ QECC can correct any error set $\mathcal{E}=\{E_{\alpha}\}$ with each $E_{\alpha}$ non-trivially acting on at most $\lfloor (d-1)/2 \rfloor$ qudits.
	
	For convenience, $2$-ary quantum codes are usually constructed and analyzed in the Pauli framework. Consider an $n$-fold Pauli tensor product
	\begin{equation}
		O_{\alpha} \in \{X,Y,Z,I\}^{\otimes n}.
	\end{equation}
	Denote the number of $X$ factors, $Y$ factors and $Z$ factors in $O_{\alpha}$ as $\operatorname{wt}_{\mathrm{X}}(O_{\alpha})$, $\operatorname{wt}_{\mathrm{Y}}(O_{\alpha})$, and $\operatorname{wt}_{\mathrm{Z}}(O_{\alpha})$. The weight of $O_{\alpha}$ is 
	\begin{equation}
		\operatorname{wt}(O_{\alpha}) = \operatorname{wt}_{\mathrm{X}}(O_{\alpha}) + \operatorname{wt}_{\mathrm{Y}}(O_{\alpha}) + \operatorname{wt}_{\mathrm{Z}}(O_{\alpha}).
	\end{equation}
	An equivalent definition of the code distance of a QECC with projector $P_c$ is the largest possible integer $d$ such that 
	\begin{equation}
		P_c O_{\alpha} P_c=\lambda_{\alpha} P_c
	\end{equation}
	holds for all Pauli tensor product $O_{\alpha}$ with $\operatorname{wt}(O_{\alpha})<d$.
	
	In practical scenarios, Pauli-$Z$ errors are usually more prevalent than Pauli-$X$ and Pauli-$Y$~\cite{aliferis2009fault}. Accordingly, we use a parameter $c_{\text{Z}}$ to characterize this noise bias and define the following $c_{\text{Z}}$-effective weight and $c_{\text{Z}}$-effective distance. 
	
	\begin{definition}
		The $c_{\text{Z}}$-effective weight of a Pauli tensor product $O_{\alpha}$ is 
		\begin{equation}
			\operatorname{wt}_e(O_{\alpha}, c_{\text{Z}}) = \operatorname{wt}_{\mathrm{X}}(O_{\alpha}) + \operatorname{wt}_{\mathrm{Y}}(O_{\alpha}) + c_{\text{Z}} \operatorname{wt}_{\mathrm{Z}}(O_{\alpha}),
		\end{equation}
		where $c_{\text{Z}} > 0$. The $c_{\text{Z}}$-effective distance of a quantum code with projector $P_c$ is the largest possible integer $d_e(c_{\text{Z}})$ such that 
		\begin{equation}
			P_c O_{\alpha} P_c=\lambda_{\alpha} P_c
		\end{equation}
		holds for all Pauli tensor product $O_{\alpha}$ with $\operatorname{wt}_e(O_{\alpha}, c_{\text{Z}})<d_e(c_{\text{Z}})$.
	\end{definition}
	This definition is a generalization of the concept of ``effective distance'' introduced in Ref.~\cite{jackson2016concatenated}. An asymmetric code with code parameters $((n,K,d_e(c_{\text{Z}})))_2$ can correct arbitrary Pauli error with $c_{\text{Z}}$-effective weight smaller than $d_e(c_{\text{Z}})/2$, and detect arbitrary Pauli error with $c_{\text{Z}}$-effective weight smaller than $d_e(c_{\text{Z}})$. When Pauli-$Z$ errors occur more frequently than Pauli-$X$/$Y$ errors, $0<c_{\text{Z}}<1$; when the relaxation times ($T_1$) are much smaller than the dephasing times ($T_2$), Pauli-$X$/$Y$ errors occur more frequently, $c_{\text{Z}}>1$.

	Quantum codes with relatively small distances can be concatenated to construct a code with large code length and distance, as illustrated in Fig.~\ref{fig: schematic_canc}. Suppose we have an outermost code with parameters $((n_1,K,d_1))_q$, other outer codes with parameters $((n_2,q,d_2))_q$, $((n_3,q,d_3))_q$, $\dots$, $((n_{l-1},q,d_{l-1}))_q$, and an inner code with parameters $((n_{l},q,d_{l}))_q$. We can construct a large code through several levels of concatenation: the logical data is first encoded using the outermost code, each physical qudit therein is further encoded using the $((n_2,q,d_2))_q$ code, and so forth. The hierarchically concatenated quantum code has parameters
	\begin{equation}
		((\prod_{j}n_j,K,\prod_{j}d_j))_q.
	\end{equation} 
	
	Likewise, we can concatenate asymmetric codes. A distance lower bound is given as follows. 
	\begin{theorem}\label{theorem1}
		Consider asymmetric outer codes with parameters $((n_1,K,d_e(c_{\text{Z}})=\delta_1))_2$, $((n_2,2,d_e(c_{\text{Z}})=\delta_2))_2$, $((n_3,2,d_e(c_{\text{Z}})=\delta_3))_2$, $\dots$, $((n_{l-1},2,\delta_{l-1}))_2$, and an inner code with parameters $((n_{l},2,\delta_{l}))_2$. Concatenating these codes yields a new code with parameters
		\begin{equation}
			((\prod_{j}n_j,K,d_e(c_{\text{Z}})=\delta))_2
		\end{equation}
		where
		\begin{equation}
			\delta \geq \min\{1,c_{\text{Z}}\}\prod_{j} \left \lceil \frac{\delta_j}{\max\{1, c_{\text{Z}}\}} \right \rceil.
		\end{equation}
	\end{theorem}
	\begin{proof}
		Assume the concatenated code cannot detect a Pauli tensor product $O_{\alpha}$. For the outer code, errors occur on at least  $\left\lceil\delta_1/\max\{1, c_{\text{Z}}\} \right\rceil$ qubits. Each of these qubits is connected to a block of the first inner code $((n_2,2,d_e(c_{\text{Z}})=\delta_2))_2$ and for every such block, errors occur on at least $\left\lceil\delta_2\max\{1, c_{\text{Z}}\} \right\rceil$ qubits. From similar arguments, errors occur on at least $\left\lceil\delta_{j}\max\{1, c_{\text{Z}}\} \right\rceil$ qubits in the $j$-th block. The weight of $O_{\alpha}$ is bounded by
		\begin{equation}
			\operatorname{wt}(O_{\alpha}) \geq \prod_{j} \left \lceil \delta_j/\max\{1, c_{\text{Z}}\} \right \rceil.
		\end{equation}
		Hence, the $c_{\text{Z}}$-effective weight of $O_{\alpha}$ is at least
		\begin{equation}
			\min\{1,c_{\text{Z}}\}\prod_{j} \left \lceil \delta_j/\max\{1, c_{\text{Z}}\} \right \rceil.
		\end{equation}
		The concatenated code can detect any Pauli tensor product with $c_{\text{Z}}$-effective weight smaller than this value, we conclude
		\begin{equation}
			\delta \geq \min\{1,c_{\text{Z}}\}\prod_{j} \left \lceil \frac{\delta_j}{\max\{1, c_{\text{Z}}\}} \right \rceil.
		\end{equation}
	\end{proof}
	
	\begin{figure}[tb]
		\centering
		\includegraphics[width=6cm]{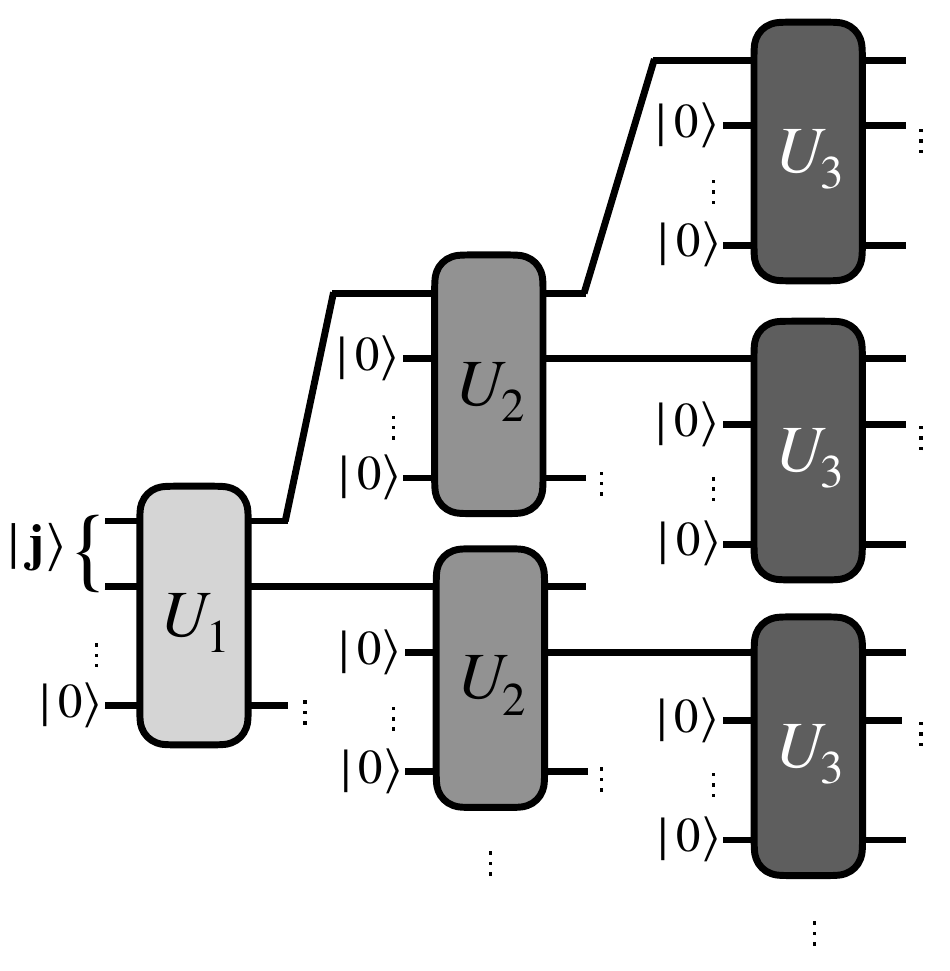}
		\caption{Schematic illustration of quantum code concatenation. After finding quantum codes with encoders $U_1, U_2, U_3,\dots$, we hierarchically concatenate these encoders to obtain a large-distance quantum code.}
		\label{fig: schematic_canc}
	\end{figure}
	
	\section{Theoretical Basis}\label{Sec:theory}
	
	A lot of methods for constructing QECCs are using the stabilizer formalism, but there are not that many outside the Pauli framework. This work aims to search for quantum codes based on the most fundamental principle, i.e., the Knill-Laflamme conditions and the quantum error detection conditions. A crucial tool in our scheme is the variational quantum circuit which consists of multiple layers of parameterized quantum gates. 
	
	The primary ingredient of a variational algorithm is the cost function(s). We define the cost functions of VarQEC as follows.
	
	\begin{definition}[Cost functions]\label{def_cost}
		Consider an error set $\mathcal{E} = \{E_{\mu}\}$ and a length-$n$ quantum code with parameterized orthogonal basis states 
		\begin{equation}
			\{|\psi_1(\boldsymbol{\theta})\rangle,|\psi_2(\boldsymbol{\theta})\rangle,\dots,|\psi_K(\boldsymbol{\theta})\rangle\}.
		\end{equation}
		We define the $\ell_{1}$-norm cost function
		\begin{equation}\label{l1norm}
			\begin{aligned}
				C^{\ell_{1}}_{n, K,\mathcal{E}}(\boldsymbol{\theta})\equiv&\sum_{E_{\mu} \in \mathcal{E}} \Big(\sum_{1 \leq i<j \leq K} \big|\langle\psi_i|E_{\mu} |\psi_j\rangle\big| \\
				&+\sum_{j=1}^K \big|\langle\psi_j|E_{\mu}|\psi_j\rangle-\overline{\langle E_{\mu} \rangle}\big|/2 \Big)
			\end{aligned}
		\end{equation}
		and the $\ell_{2}$-norm cost function
		\begin{equation}\label{l2norm}
			\begin{aligned}
				C^{\ell_{2}}_{n, K,\mathcal{E}}(\boldsymbol{\theta})\equiv&\sum_{E_{\mu} \in \mathcal{E}}  \Big(\sum_{1 \leq i<j \leq K} \big|\langle\psi_i|E_{\mu} |\psi_j\rangle\big|^2 \\
				&+\sum_{j=1}^K \big|\langle\psi_j|E_{\mu}|\psi_j\rangle-\overline{\langle E_{\mu} \rangle}\big|^2/4 \Big)
			\end{aligned}
		\end{equation}
		where $\overline{\langle E_{\mu} \rangle} = \sum_{j=1}^{K}\langle\psi_j|E_{\mu}|\psi_j\rangle/K$.
	\end{definition}
	Clearly, $C^{\ell_{1}}_{n, K,\mathcal{E}}$ and $C^{\ell_{2}}_{n, K,\mathcal{E}}$ are always non-negative and have the same zero-points. When $C^{\ell_{1}}_{n, K,\mathcal{E}} \leq 1$, $C^{\ell_{2}}_{n, K,\mathcal{E}} \leq (C^{\ell_{1}}_{n, K,\mathcal{E}})^2$. When $C^{\ell_{1}}_{n, K,\mathcal{E}}=0$, the quantum code can perfectly detect the error set $\mathcal{E}$.
	
	To find symmetric codes with code parameters $((n,K,d))_2$, we use the Pauli error model and choose
	\begin{equation}
		\mathcal{E}= \{O_{\alpha}|\operatorname{wt}(O_{\alpha})<d \},
	\end{equation}
	where $O_{\alpha}$s are Pauli tensor products. Likewise, when searching for asymmetric codes with code parameters $((n,K,d_e(c_{\text{Z}})))_2$, we choose
	\begin{equation}
		\mathcal{E}= \{O_{\alpha}|\operatorname{wt_e}(O_{\alpha}, c_{\text{Z}})<d_e(c_{\text{Z}}) \}.
	\end{equation}
	To find channel-adaptive codes for a general noise channel $\mathcal{N}(\rho)=\sum_{\alpha} E_{\alpha} \rho E_{\alpha}^{\dagger}$, we choose
	\begin{equation}
		\mathcal{E}= \{E^{\dagger}_{\alpha}E_{\beta}| E_{\alpha}, E_{\beta} \text{ are Kraus operators of } \mathcal{N}\}.
	\end{equation}
	Note that the error set $\mathcal{E}$ in principle can include non-unitary and non-Hermitian errors. For such errors, we can either twirl them to Pauli errors or simulate them directly by adding ancilla qubits and performing \textit{positive-operator valued measures} (POVMs). 
	
	In practice, due to the inexact realization of an encoding isometry, quantum error correction/detection conditions are not exactly satisfied, and QECCs cannot protect the information from errors perfectly. Nevertheless, QECCs can still detect and correct most errors. Such approximate quantum error correction schemes hold great promise~\cite{leung1997approximate,schumacher2002approximate}. A parameter $\varepsilon$ characterizes the \textit{inaccuracy} of an approximate code. If a QECC is $\varepsilon$-correctable for a noise channel $\mathcal{N}$, its worst-case entanglement fidelity is greater than $1-\varepsilon$ with appropriate recovery~\cite{brandao2019quantum}. \text {Bény} $\textit{et al.}$ proposed an approximate version of the Knill-Laflamme conditions for such approximate codes.
	
	\begin{lemma}[Corollary 2, Ref.~\cite{beny2010general}]\label{lemma1}
		A code defined by the projector $P_c$ is $\varepsilon$-correctable under a noise channel $\mathcal{N}(\rho)=\sum_{\alpha} E_{\alpha} \rho E_{\alpha}^{\dagger}$, if and only if
		\begin{equation}\label{eq: PBP}
			P_c E_{\alpha}^{\dagger} E_{\beta} P_c=\lambda_{\alpha \beta} P_c+P_c B_{\alpha \beta} P_c,
		\end{equation}
		where $\lambda_{\alpha \beta}$ are the components of a non-negative Hermitian operator with trace one, $B_{\alpha \beta}$ is a Hermitian operator, and the Bures distance~\cite{bures1969extension} between two channels 
		$\Lambda(\rho)=\sum_{\alpha \beta} \lambda_{\alpha \beta} \operatorname{Tr}(\rho)|\alpha\rangle\langle \beta|$ and $(\Lambda+\mathcal{B})(\rho)=\Lambda(\rho)+\sum_{\alpha \beta} \operatorname{Tr}\left(\rho B_{\alpha \beta}\right)|\alpha\rangle\langle \beta|$ satisfies
		\begin{equation}
			d(\Lambda+\mathcal{B}, \Lambda) \leq \varepsilon.
		\end{equation}
	\end{lemma}
	
	Based on this lemma, we modify Corollary 5 of Ref.~\cite{brandao2019quantum} and give a proposition to support our definitions. 
	
	\begin{proposition}\label{prop2}
		Consider an $n$-qubit noise channel $\mathcal{N}(\rho)=\sum_{\alpha} E_{\alpha} \rho E_{\alpha}^{\dagger}$, and a quantum error-correcting code
		\begin{equation}
			\mathcal{C} = \text{span}\{|\psi_1\rangle,|\psi_2\rangle,\dots,|\psi_K\rangle\}.
		\end{equation}
		We choose the error product set\\ $\mathcal{E}= \{E^{\dagger}_{\alpha}E_{\beta}| E_{\alpha}, E_{\beta} \text{ are Kraus operators of } \mathcal{N}\}$.
		Denote the cost function Eq.~\eqref{l1norm} of the basis states as $C^{\ell_{1}}_{n, K,\mathcal{E}}$. Then the code $\mathcal{C}$ is $\varepsilon$-correctable under $\mathcal{N}$ with $\varepsilon$ bounded by
		\begin{equation}
			\varepsilon \leq  K\sqrt{2C^{\ell_{1}}_{n, K,\mathcal{E}}}.
		\end{equation}
	\end{proposition}
	\begin{proof}
		Let $\lambda_{\alpha \beta} = \sum_{j}\langle\psi_j|E^{\dagger}_{\alpha}E_{\beta}|\psi_j\rangle/K$. 
		
		To satisfy Eq.~\eqref{eq: PBP}, we set
		\begin{equation}
			\begin{aligned}
				B_{\alpha \beta}=&\sum_{i \neq j}\langle\psi_{i}|E^{\dagger}_{\alpha} E_{\beta}| \psi_{j}\rangle|\psi_{i}\rangle\langle\psi_{j}|\\&+\sum_{j}\big(\langle\psi_{j}|E^{\dagger}_{\alpha}E_{\beta}| \psi_{j}\rangle-\lambda_{\alpha \beta}\big)| \psi_{j}\rangle\langle\psi_{j}|.
			\end{aligned}
		\end{equation}
		
		Then
		\begin{equation}
			\begin{aligned}
				d(\Lambda+\mathcal{B}, \Lambda) &\leq K\|\mathcal{B}\|_{1}^{1 / 2} \\&\leq K\Big(\sum_{\alpha,\beta}\left\|B_{\alpha \beta}\right\|_{1}\Big)^{1 / 2} \\
				&\leq K\Big(\sum_{\alpha,\beta}\Big(\sum_{i \neq j} |\langle\psi_i|E^{\dagger}_{\alpha}E_{\beta} |\psi_j\rangle| +\\ &\quad\quad\sum_{j} |\langle\psi_j|E^{\dagger}_{\alpha}E_{\beta}|\psi_j\rangle-\lambda_{\alpha \beta}|\Big)\Big)^{1 / 2}\\
				&= K\sqrt{2C^{\ell_{1}}_{n, K,\mathcal{E}}}.
			\end{aligned}
		\end{equation}
		According to Lemma \ref{lemma1}, the inaccuracy $\varepsilon$ of code $\mathcal{C}$ is upper bounded by $K\sqrt{2C^{\ell_{1}}_{n,K,\mathcal{E}}}$. 
	\end{proof}
	
	In short, given a noise channel $\mathcal{N}$, as long as we minimize the channel-adaptive cost function to a sufficiently small value, we rigorously find an approximate channel-adaptive code with small inaccuracy. Similar bounds for symmetric or asymmetric quantum codes are given as follows.
	
	\begin{proposition}\label{prop3}
		Consider an $n$-qubit noise channel $\mathcal{N}(\rho)=\sum_{\alpha} E_{\alpha} \rho E_{\alpha}^{\dagger}$ where each $E_{\alpha}$ non-trivially acts on no more than $\lfloor (d-1)/2 \rfloor$ qubits, and a quantum error-correcting code
		\begin{equation}
			\mathcal{C} = \text{span}\{|\psi_1\rangle,|\psi_2\rangle,\dots,|\psi_K\rangle\}.
		\end{equation}
		We choose $\mathcal{E}= \{O_{\alpha}|\operatorname{wt}(O_{\alpha})<d \}$, where $O_{\alpha}$s are Pauli tensor products. Denote the cost function Eq.~\eqref{l1norm} of the basis states as $C^{\ell_{1}}_{n, K, \mathcal{E}}$, the number of Kraus operators of $\mathcal{N}$ as $m$. Then the code $\mathcal{C}$ is $\varepsilon$-correctable under $\mathcal{N}$ with $\varepsilon$ bounded by
		\begin{equation}
			\varepsilon \leq  2^{n/4 + d/2}K\sqrt{mC^{\ell_{1}}_{n, K,\mathcal{E}}}.
		\end{equation}
	\end{proposition}
	
	\begin{proof}
		The proof is given in Appendix ~\ref{ap:proof7}.
	\end{proof}

	\begin{proposition}\label{prop4}
		Consider an $n$-qubit noise channel $\mathcal{N}(\rho)=\sum_{\alpha} E_{\alpha} \rho E_{\alpha}^{\dagger}$ with each $E_{\alpha}$ proportional to a Pauli error with $c_{\text{Z}}$-effective weight smaller than $d_e(c_{\text{Z}})/2$, and a quantum error-correcting code
		\begin{equation}
			\mathcal{C} = \text{span}\{|\psi_1\rangle,|\psi_2\rangle,\dots,|\psi_K\rangle\}.
		\end{equation}
		We choose $\mathcal{E}= \{O_{\alpha}|\operatorname{wt_e}(O_{\alpha}, c_{\text{Z}})<d_e(c_{\text{Z}}) \}$, where $O_{\alpha}$s are Pauli tensor products.
		Denote the cost function Eq.~\eqref{l1norm} of the basis states as $C^{\ell_{1}}_{n, K, \mathcal{E}}$, the number of Kraus operators of $\mathcal{N}$ as $m$. Then the code $\mathcal{C}$ is $\varepsilon$-correctable under $\mathcal{N}$ with $\varepsilon$ bounded by
		\begin{equation}
			\varepsilon \leq  K\sqrt{2mC^{\ell_{1}}_{n, K,\mathcal{E}}}.
		\end{equation}
	\end{proposition}
	\begin{proof}
		The proof is given in Appendix ~\ref{ap:proof8}.
	\end{proof}
	
	Note that Propositions~\ref{prop3} and~\ref{prop4} give pretty loose bounds. The true code inaccuracy, which depends on the particular noise channel, is usually significantly smaller.

	\section{Algorithm}\label{Sec:algorithm}
	
	Variational quantum circuits (VQCs) have been widely used in near-term quantum algorithms for various tasks~\cite{cerezo2021variational, bharti2022noisy}, such as ground state preparation~\cite{peruzzo2014variational, kandala2017hardware}, eigenenergy estimation~\cite{nam2020ground, cao2021energy}, quantum data compression~\cite{romero2017quantum, cao2021noise}, quantum circuit compiling~\cite{sharma2020noise, xu2021variational}. Give a pure product state as input, one iteratively updates the circuit parameters based on measurement results, and finally outputs the desired state. In VarQEC, the output states serve as the basis states of a quantum code, and its encoder is given by the quantum circuit. The structure of our algorithm is illustrated in Fig.~\ref{fig: schematic_varqec}.

	\begin{figure}[tb]
		\centering
		\includegraphics[width=6.5cm]{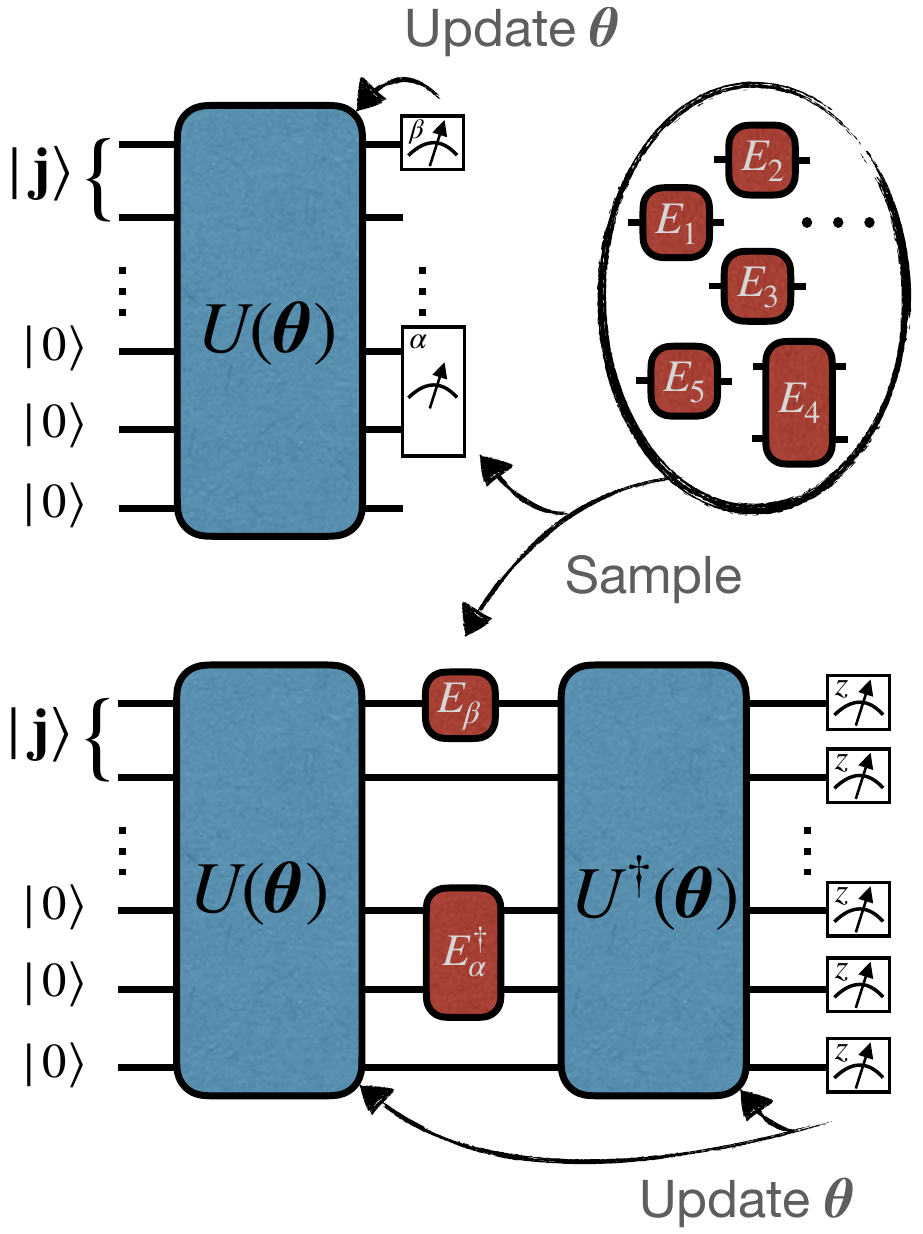}
		\caption{Schematic illustration of VarQEC. The encoder is trained via mini-batch learning: we iteratively sample errors from an error set, run the variational quantum circuit $U(\boldsymbol{\theta})$ and do measurements, then update $\boldsymbol{\theta}$.}
		\label{fig: schematic_varqec}
	\end{figure}
	
	Suppose we have a NISQ device with a \textit{hardware connectivity graph} $G$. The vertices denote qubits, and the edges denote adjacent qubit pairs. One can apply single-qubit rotations to each qubit and two-qubit gates to adjacent qubits. We aim to find a $K$-dimensional QECC that can detect an error set $\mathcal{E} = \{E_{\mu}\}$, and the encoding circuit should be as shallow as possible. 
	
	Before running the algorithm, we design a multilayered VQC which is hardware-efficient for the connectivity graph. Denote the number of VQC layers as $L$, the maximum acceptable number of layers as $L_{\mathrm{max}}$, the evolution of the VQC as $U(\boldsymbol{\theta})$ where $\boldsymbol{\theta}$ are the circuit parameters. We start from $L=1$ and sample the initial $\boldsymbol{\theta}$ randomly. Also, we delicately select $k = \lceil\log(K)\rceil$ physical qubits to prepare the logical data, where the logarithm is with respect to base 2. These $k$ qubits should be scattered instead of concentrated since we hope the remaining qubits are connected to them by very few edges.
	
	First, we initialize the selected qubits to one of the $K$ binary strings $|\mathbf{0}\rangle$, $|\mathbf{1}\rangle$, $\dots$, $|\mathbf{K-1}\rangle$, and initialize the remaining qubits to $|0\rangle^{\otimes(n-k)}$. These product states span the input code space
	\begin{equation}
		\mathcal{C}_{\mathrm{in}} = \text{span}\{|\mathbf{0}\rangle|0\rangle^{\otimes(n-k)},\dots,|\mathbf{K-1}\rangle|0\rangle^{\otimes(n-k)}\}.
	\end{equation}
	$\mathcal{C}_{\mathrm{in}}$ is a QECC with distance $d = 1$. 
	
	The cost functions can be estimated by running specific circuits and doing measurements. To estimate $\langle\psi_j|E_{\mu}|\psi_j\rangle$, we prepare the initial state $|\mathbf{j-1}\rangle|0\rangle^{\otimes(n-k)}$, evolve the system with the VQC $U(\boldsymbol{\theta})$, then measure the local observable $E_{\mu}$. If errors $E_{\mu1}$, $E_{\mu2}$, $\dots$ commute, they can be measured simultaneously in a single shot. To estimate $\big|\langle\psi_i|E_{\mu} |\psi_j\rangle\big|$, we start from $|\mathbf{j-1}\rangle|0\rangle^{\otimes(n-k)}$, then sequentially evolve the system with  $U(\boldsymbol{\theta})$, $E_{\mu}$, and $U^{\dagger}(\boldsymbol{\theta})$, then measure the final state in the computational basis. The measurements are assisted by \textit{post-selection}: we first measure the $n-k$ auxiliary qubits, and if the result is $|0\rangle^{n-k}$,  we measure the remaining $k$ qubits. Denote the probability of obtaining the binary string $|\mathbf{i-1}\rangle|0\rangle^{\otimes(n-k)}$ as $p_{ij}$, $\big|\langle\psi_i|E_{\mu} |\psi_j\rangle\big|=\sqrt{p_{ij}}$. Theoretically, this step will also yield $\langle\psi_j|E_{\mu}|\psi_j\rangle$. However, since VarQEC is a NISQ algorithm, we prefer to use a shallower circuit to estimate cost function terms whenever possible.
	
	In the above description, we assume the error set $\mathcal{E}$ only consists of Pauli errors. It does not matter if $\mathcal{E}$ includes non-unitary or non-Hermitian terms. Adding ancilla qubits or Pauli twirling can handle it. See Sec. \ref{Sec:collective-ad} for a detailed example.
	
	The optimization of $\boldsymbol{\theta}$ consists of two stages. The first and the main stage is \textit{mini-batch learning}. After sampling the initial $\boldsymbol{\theta}$, we minimize $C^{\ell_{2}}_{n, K,\mathcal{E}}$ with mini-batch gradient descent. The schematic is shown in Fig.~\ref{fig: schematic_varqec}. Within each iteration, we sample a subset $\mathcal{E}_{S} \subset  \mathcal{E}$, estimate the corresponding partial $\ell_{2}$-norm cost function
	\begin{equation}
		\begin{aligned}
			C^{\ell_{2}}_{n, K,\mathcal{E}_{S}}(\boldsymbol{\theta})\equiv&\sum_{E_{\mu} \in \mathcal{E}_{S}} \Big(\sum_{1 \leq i<j \leq K} \big|\langle\psi_i|E_{\mu} |\psi_j\rangle\big|^2 \\
			&+\sum_{j=1}^K \big|\langle\psi_j|E_{\mu}|\psi_j\rangle-\overline{\langle E_{\mu} \rangle}\big|^2/4 \Big)
		\end{aligned}
	\end{equation}
	and its gradient 
	\begin{equation}
		\nabla \boldsymbol{\theta} =\frac{\partial C^{\ell_{2}}_{n, K,\mathcal{E}_{S}}(\boldsymbol{\theta})}{\partial \boldsymbol{\theta}}
	\end{equation}
	via measurements, then perform a single gradient descent step with a learning rate $\eta$:
	\begin{equation}
		\boldsymbol{\theta} \leftarrow \boldsymbol{\theta} - \eta\nabla\boldsymbol{\theta}.
	\end{equation}
	The required number of measurements to estimate $C^{\ell_{2}}_{n, K,\mathcal{E}_{S}}(\boldsymbol{\theta})$ up to additive error $\epsilon$ is of order $\mathcal{O}(K^2|\mathcal{E}_{S}|^2/\epsilon^2)$. The gradient can be estimated through finite-differencing or by combining the chain rule and the parameter shift rule~\cite{mitarai2018quantum}. Mini-batch gradient descent allows for a more robust convergence and avoids being trapped in a local minimum. We repeat sampling and gradient descent until convergence. The reason that we minimize $C^{\ell_{2}}_{n, K,\mathcal{E}}$ first is because it converges much faster than $C^{\ell_{1}}_{n, K,\mathcal{E}}$. In addition, $C^{\ell_{2}}_{n, K,\mathcal{E}}$ is differentiable but $C^{\ell_{1}}_{n, K,\mathcal{E}}$ is not. 
	
	If the error set $\mathcal{E}$ consists of too many terms, a promising alternative method is to construct ``classical shadows''~\cite{huang2020predicting} for each basis state $|\psi_j\rangle$, then use the shadows to estimate the cost functions classically. The shadow tomography technique can help us implement large-batch optimization with a smaller measurement overhead.

	After adequate mini-batch learning, if $C^{\ell_{2}}_{n, K, \mathcal{E}}$ is relatively small (e.g., $C^{\ell_{2}}_{n, K, \mathcal{E}} < 0.01$), we estimate $C^{\ell_{1}}_{n, K,\mathcal{E}}$ and \textit{fine-tune} the parameters $\boldsymbol{\theta}$ with respect to it since $C^{\ell_{1}}_{n, K,\mathcal{E}}$ is directly related to the inaccuracy of the code (see Propositions~\ref{prop2},~\ref{prop3},~\ref{prop4}). In this work, we use Powell's method~\cite{powell1964efficient}, a gradient-free optimizer, for fine-tuning. If $C^{\ell_{1}}_{n, K,\mathcal{E}}$ is smaller than an acceptable cost tolerance $C^{\ell_{1}}_{\mathrm{tol}}$, we stop the optimization and output the final parameters $\boldsymbol{\theta}_{\mathrm{opt}}$. Throughout this paper, we set the tolerance as 
	\begin{equation}
		C^{\ell_{1}}_{\mathrm{tol}} \equiv 1 \times 10^{-6}.
	\end{equation}
	In the ideal case, we obtain the optimal parameters
	\begin{equation}
		\boldsymbol{\theta}_{\mathrm{opt}} =  \arg \min _{\boldsymbol{\theta}} C^{\ell_{1}}_{n, K,\mathcal{E}}(\boldsymbol{\theta}).
	\end{equation}
	The output QECC
	\begin{equation}
		\begin{aligned}
			\mathcal{C}_{\mathrm{out}}(\boldsymbol{\theta}_{\mathrm{opt}})=  & \text{span}\{ \\ &|\psi_1\rangle = U(\boldsymbol{\theta}_{\mathrm{opt}})|\mathbf{0}\rangle|0\rangle^{\otimes(n-k)},
			\\&  |\psi_2\rangle = U(\boldsymbol{\theta}_{\mathrm{opt}})|\mathbf{1}\rangle|0\rangle^{\otimes(n-k)},
			\\& \dots,
			\\&|\psi_K\rangle = U(\boldsymbol{\theta}_{\mathrm{opt}})|\mathbf{K-1}\rangle|0\rangle^{\otimes(n-k)}\\ &\qquad \quad \}.
		\end{aligned}
	\end{equation}
	is the target approximate quantum code with small inaccuracy. The variational quantum circuit $U(\boldsymbol{\theta}_{\mathrm{opt}})$ serves as the encoding circuit. Further, we can remove redundant gates from the VQC.
	
	If $C^{\ell_{1}}_{n, K,\mathcal{E}}$ is greater than $C^{\ell_{1}}_{\mathrm{tol}}$, we increase the circuit depth $L$ and repeat the optimization steps. If $C^{\ell_{1}}_{n, K,\mathcal{E}}$ is always greater than the tolerance even when $L = L_{\mathrm{max}}$, we fail to find an eligible code. The detailed procedure is illustrated in Algorithm~\ref{alg:qnn}.

	A natural question arises: can a fixed-depth VQC find any $((n, K))_2$ quantum code? \text {Haug} $\textit{et al.}$ used the quantum Fisher information matrix to assess the expressive power of a VQC with a fixed input state $|0\rangle^{\otimes n}$~\cite{QFIM_1}. We generalize this notion to multiple inputs to assess the expressive power of a VQC in VarQEC. If a VQC is capable of finding any $((n, K))_2$ quantum code, we say it is \textit{overparameterized} with respect to code parameters $((n, K))_2$. See Appendix~\ref{ap:effect quantum dimension} for more details.
	
	When the VQC $U(\boldsymbol{\theta})$ is underparameterized for $((n,K))_2$, the set of reachable output codes forms a low-dimensional submanifold of the complex Grassmannian $\mathbf{G r}(K, 2^n)$,
	\begin{equation}
		\{\mathcal{C}_{\mathrm{out}}(\boldsymbol{\theta})|\boldsymbol{\theta}\} \subseteq \mathbf{G r}(K, 2^n).
	\end{equation}
	\removelatexerror
	\begin{algorithm*}[H]
		\caption{VarQEC}
		\label{alg:qnn}
		\SetKwInOut{Return}{Return}
		\KwIn{Error set $\mathcal{E}$, hardware-efficient VQC $U(\boldsymbol{\theta})$ with $L$ layers, acceptable number of layers $L_{\mathrm{max}}$, acceptable cost tolerance $C^{\ell_{1}}_{\mathrm{tol}}$.}
		\KwOut{An approximate quantum code with a hardware-efficient encoder that detects $\mathcal{E}$.}
		$L \leftarrow 1$.\\
		\While{$L \leq L_{\mathrm{max}}$ and $C^{\ell_{1}}_{n, K,\mathcal{E}}(\boldsymbol{\theta}) > C^{\ell_{1}}_{\mathrm{tol}}$}{
			\While{$C^{\ell_{2}}_{n, K,\mathcal{E}}(\boldsymbol{\theta})$ has not converged}{
				Sample a subset $\mathcal{E}_S \subset \mathcal{E}$.\\
				Prepare the $K$ input strings.\\
				Run $U(\boldsymbol{\theta})$, output $\{|\psi_j\rangle\}$.\\
				Measure observables $E_{\mu} \in \mathcal{E}_S$.\\
				Prepare the $K$ input strings.\\
				Run $U^{\dagger}(\boldsymbol{\theta})E_{\mu} U(\boldsymbol{\theta})$ for $E_{\mu} \in \mathcal{E}_S$.\\
				Do projective measurements. \\
				Estimate $C^{\ell_{2}}_{n, K,\mathcal{E}_S}(\boldsymbol{\theta})$.\\
				Vary $\boldsymbol{\theta}$, repeat the above steps to estimate $\partial{C^{\ell_{2}}_{n, K,\mathcal{E}}(\boldsymbol{\theta})}/\partial{\boldsymbol{\theta}}$.\\
				Perform a gradient descent step, update $\boldsymbol{\theta}$.\\
			}
			\If{$C^{\ell_{2}}_{n, K, \mathcal{E}}(\boldsymbol{\theta}) < 0.01$}{\While{$C^{\ell_{1}}_{n, K, \mathcal{E}}(\boldsymbol{\theta})$ has not converged}{
					Prepare the $K$ input strings.\\
					Run $U(\boldsymbol{\theta})$, output $\{|\psi_j\rangle\}$.\\
					Measure observables $E_{\mu} \in \mathcal{E}$.\\
					Prepare the $K$ input strings.\\
					Run $U^{\dagger}(\boldsymbol{\theta})E_{\mu}U(\boldsymbol{\theta})$ for $E_{\mu} \in \mathcal{E}$.\\
					Do projective measurements. \\
					Estimate $C^{\ell_{1}}_{n, K,\mathcal{E}_S}(\boldsymbol{\theta})$.\\
				Vary $\boldsymbol{\theta}$, repeat the above steps to estimate $\partial{C^{\ell_{1}}_{n, K,\mathcal{E}}(\boldsymbol{\theta})}/\partial{\boldsymbol{\theta}}$.\\
					Minimize $C^{\ell_{1}}_{n, K,\mathcal{E}}(\boldsymbol{\theta})$, update $\boldsymbol{\theta}$.\\

			}}
			$L \leftarrow L+1$.\\
			
		}
		$\boldsymbol{\theta}_{\mathrm{opt}} \leftarrow \boldsymbol{\theta}$.\\
		\Return{$\boldsymbol{\theta}_{\mathrm{opt}}$, $C^{\ell_{1}}_{n, K,\mathcal{E}}(\boldsymbol{\theta}_{\mathrm{opt}})$.}
	\end{algorithm*}
	\noindent The VarQEC algorithm searches this submanifold for an eligible code. When $U(\boldsymbol{\theta})$ is overparameterized for $((n,K))_2$, it can explore all relevant directions and the set of reachable output codes is equivalent to $\mathbf{G r}(K, 2^n)$, i.e., VarQEC is capable to find arbitrary $((n,K))_2$ quantum code. The required number of periodic bounded real parameters to overparameterize a VQC is at least $2K(2^n-K)$ since the complex dimension of $\mathbf{G r}(K, 2^n)$ is $K(2^n-K)$.
	
	In Ref.~\cite{johnson2017qvector}, Johnson $\textit{et al.}$ proposed a related algorithm named QVECTOR, which samples a random 2-design unitary and optimizes parameterized encoding and decoding circuits simultaneously to improve the quantum average fidelity. Compared with QVECTOR, VarQEC can find not only channel-adaptive codes but also quantum codes with specific code parameters. VarQEC does not need a deep random circuit, which is a daunting challenge on NISQ devices, to sample a bunch of input states.  We train the encoder without considering the decoder. The optimization is less likely to be trapped in a local minimum. The cost functions are estimated by measuring some local observables. We can rigorously obtain an $\varepsilon$-correctable approximate QECC with arbitrarily small $\varepsilon$. The noise models in our methods are also flexible and can be artificially assigned.

	\section{Results}\label{Sec:results}
	
	\subsection{Symmetric codes}
	We verify the validity of our algorithm by rediscovering some symmetric codes with well-known code parameters. The cost functions are defined in Eqs.~\eqref{l1norm},~\eqref{l2norm}. For code parameters $((n,K,d))_2$, the total number of Pauli errors $O_\alpha$ to consider is
	\begin{equation}\label{Pauli_number}
		|\{O_{\alpha}\}| = \sum_{j=0}^{d-1}\left(\begin{array}{c}
			n \\
			j
		\end{array}\right) 3^{j}.
	\end{equation}

	Without loss of generality, we use the complete bipartite connectivity graph: denote the qubits selected for the input as $\{Q_0, Q_1, \dots, Q_{k-1}\}$ ($k = \lceil\log(K)\rceil$), the unselected qubits as $\{Q_{k}, Q_{k+1}, \dots, Q_{n-1}\}$, the graph consists of $k(n-k)$ edges that connect every selected qubit and every unselected qubit, qubits in the same set are not directly connected. For such graphs, the initial logical data can spread to each qubit rapidly since the graph diameter is only $2$. The variational quantum circuit has alternating layers of single-qubit rotations $R_x$-$R_z$ acting on all qubits and Ising-type interactions $R_{zz}$ acting on adjacent qubits. Denote the number of layers as $L$. The VQC evolution is of the form
	\begin{equation}
		U(\boldsymbol{\theta}) = U_{E}(\boldsymbol{\theta}_E)\prod_{l=1}^L U_l(\boldsymbol{\theta}_l),
	\end{equation}
	where $\boldsymbol{\theta}_l$ and $\boldsymbol{\theta}_E$ are elements in $\boldsymbol{\theta}$, $U_l(\boldsymbol{\theta}_l)$ denotes the $l$-th layer evolution, $U_{E}(\boldsymbol{\theta}_E)$ denotes the rightmost $R_x$-$R_z$ rotations which are used to search the manifold of locally equivalent quantum codes. Since $R_z$ and $R_{zz}$ gates in the last layer commute, and $R_z$-$R_x$-$R_z$ rotations can realize arbitrary single-qubit unitary, locally equivalent QECCs can be found by the same VQC. In principle, any $n$-qubit unitary evolution can be realized by this ansatz with a sufficiently large number of layers since $\{R_x, R_z, R_{zz}\}$ is a universal quantum gate set. The connectivity graph and the periodic-structured VQC ansatz for $n=5, k =2$ are shown in Fig.~\ref{fig:qnn_str}(a,b). In general, with the increase of $L$, the achievable quantum codes form a higher dimensional submanifold of $\mathbf{G r}(K, 2^n)$, as shown in Fig.~\ref{fig:qnn_str}(c). When the VQC is overparameterized ($L$ is no less than a critical number $L_{\mathrm{crit}}$) for code parameters $((n,K))_2$, VarQEC can explore the whole $\mathbf{G r}(K, 2^n)$ manifold. 
	
	An alternative variational circuit for finding additive quantum codes is discussed in Appendix \ref{ap:QFT-ansatz}.

	\begin{figure}[t]
		\centering
		\includegraphics[width=7.5cm]{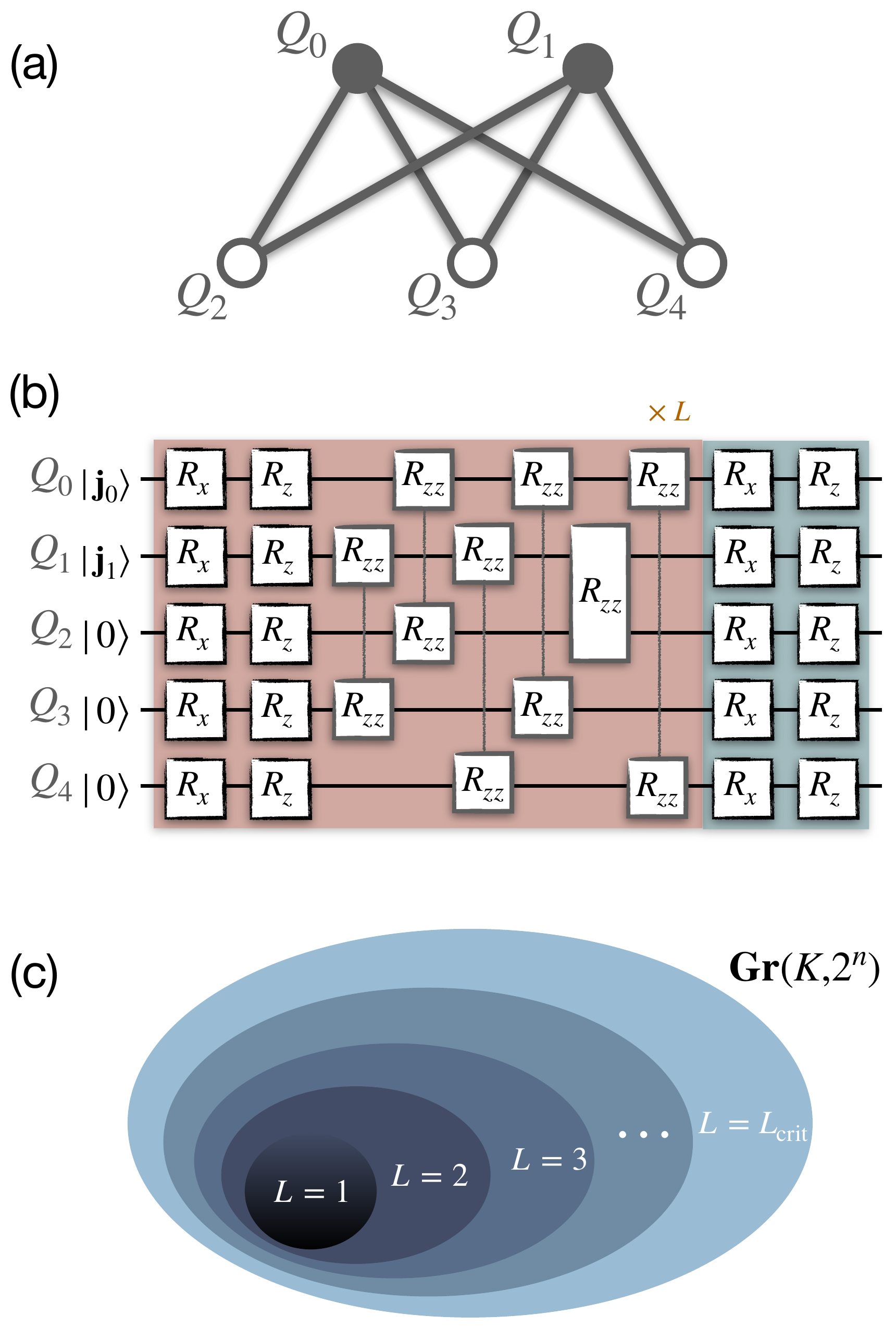}
		\caption{Schematic of a connectivity graph, the periodic-structured variational ansatz, and the achievable quantum codes. (a) The bipartite connectivity graph with five physical qubits. Gray lines connect adjacent qubits. $\{Q_0,Q_1\}$ are the selected qubits to prepare the logical data. (b) The corresponding variational quantum circuit (VQC) with $L$ layers. Within each layer, we apply $R_x$-$R_z$ rotations to each qubit and apply $R_{zz}$ gates to adjacent qubits. (c) With the increase of $L$, VarQEC is capable to find quantum codes in a higher-dimensional manifold until overparameterization ($L =L_{\mathrm{crit}}$).}
		\label{fig:qnn_str}
	\end{figure}

	We apply our algorithm to search for $((n,K,d))_2$ codes where the code length $n$ ranges from $3$ to $12$, the code dimension $K$ ranges from $2$ to $8$, the code distance $d$ ranges from $2$ to $4$. Fig.~\ref{fig:5Kd} shows $C^{\ell_{1}}_{n, K, \mathcal{E}}(\boldsymbol{\theta}_{\mathrm{opt}})$ as a function of code parameters $n$, $K$, $d$. We rediscover quantum codes with parameters
	\begin{equation}
		\begin{aligned}
			&((4,4,2))_2, ((5,6,2))_2,\\
			&((5,2,3))_2, ((6,2,3))_2,\\
			&((7,2,3))_2, ((8,8,3))_2,\\
			&((9,8,3))_2, ((10,4,4))_2,\\
			&((11,4,4))_2.
		\end{aligned}
	\end{equation}
	The $((5,2,3))_2$ and  $((8,8,3))_2$ codes are non-degenerate, the $((6,2,3))_2$ codes are degenerate, and both cases are possible for $((7,2,3))_2$.  There is an obvious phase transition between achievable and (probably) non-achievable code parameters. Further, we fix $d=3$ and rediscover larger codes with parameters
	\begin{equation}
		\begin{aligned}
			&((10,2^4,3))_2, ((11,2^5,3))_2, \\
			&((12,2^6,3))_2, ((13,2^7,3))_2,\\
			&((14,2^8,3))_2.
		\end{aligned}
	\end{equation}

	All these codes can be found by a shallow VQC, e.g., a $5$-layer VQC can encode $((5,2,3))_2$, $((12,2^6,3))_2$, $((14,2^8,3))_2$; a $4$-layer VQC can encode $((6,2,3))_2$, $((8,8,3))_2$; a $3$-layer VQC suffices to encode $((7,2,3))_2$.  In our experiments, either $C^{\ell_{1}}_{n,K,\mathcal{E}}(\boldsymbol{\theta}_{\mathrm{opt}}) < 1 \times 10^{-6}$ or $C^{\ell_{1}}_{n,K,\mathcal{E}}(\boldsymbol{\theta}_{\mathrm{opt}}) \geq 1$ holds. Fig.~\ref{fig:learning_curves} shows some cost curves at the mini-batch learning stage. Within each iteration, we sample $20\%$ of $\{O_{\alpha}\}$ as the batch and perform a stochastic gradient descent with learning rate $\eta =1 \times 10^{-2}$.

	\begin{figure}
		\centering
		\includegraphics[width=\columnwidth]{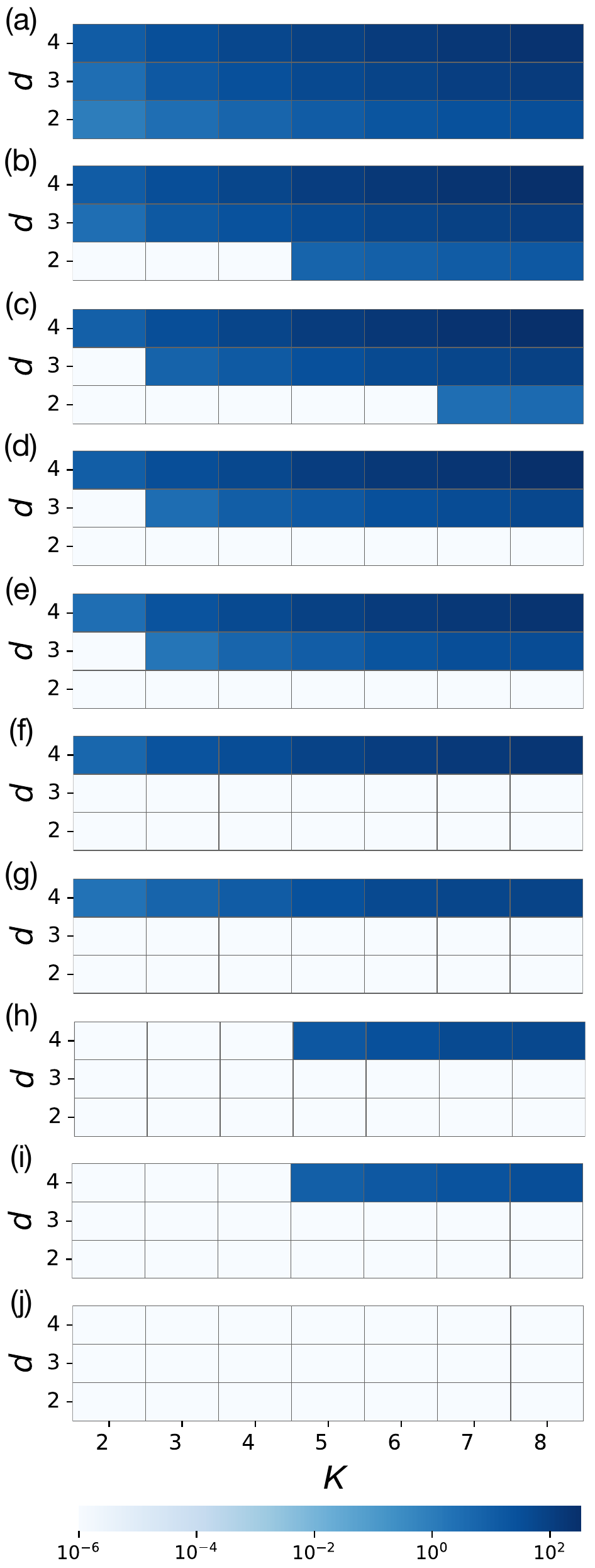}
		\caption{Minimum cost function $C^{\ell_{1}}_{n,K,\mathcal{E}}(\boldsymbol{\theta}_{\mathrm{opt}})$ for different code length $n$, code dimension $K$, and code distance $d$. (a) $n=3$. (b)$n=4$. (c)$n=5$. (d)$n=6$. (e)$n=7$. (f)$n=8$. (g)$n=9$. (h)$n=10$. (i)$n=11$. (j)$n=12$. }\label{fig:5Kd}
	\end{figure}
	\begin{figure*}[t]
		\centerline{\includegraphics[width=1\textwidth]{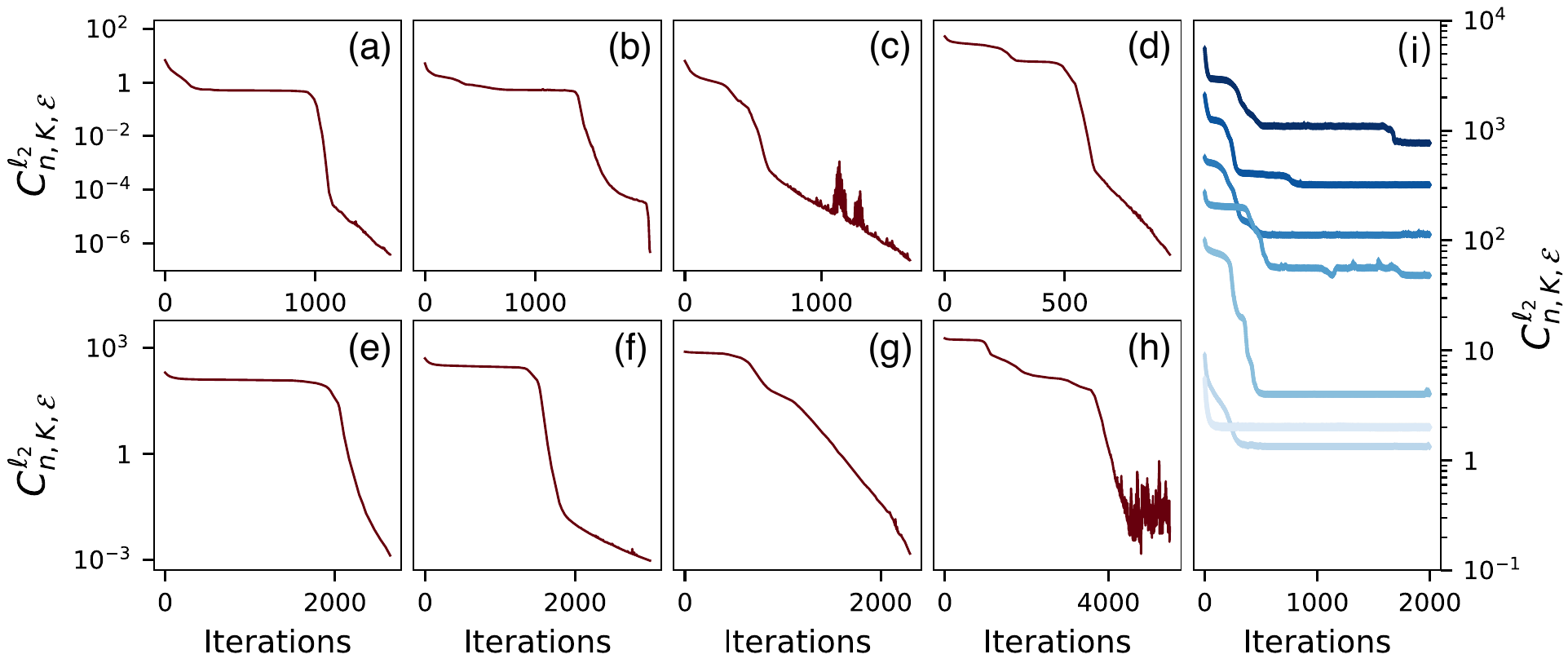}}
		\caption{Learning curves of VarQEC for finding symmetric codes with achievable code parameters (a-h)$((5,2,3))_2$, $((6,2,3))_2$, $((7,2,3))_2$, $((8,8,3))_2$, $((11,2^5,3))_2$, $((12,2^6,3))_2$, $((13, 2^7,3))_2$, $((14,2^8,3))_2$, and (probably) non-achievable code parameters (i) from bottom to top: $((7,3,3))_2$, $((4,2,3))_2$,$((9,2^4,3))_2$,$((10,2^5,3))_2$,$((11,2^6,3))_2$,$((12,2^7,3))_2$,$((13,2^8,3))_2$.}
		\label{fig:learning_curves}
	\end{figure*}

	In the following, we verify the \textit{local equivalence} (LE) between two quantum codes with $K$-dimensional projectors $P_{c}$ and $P'_{c}$ by sampling permutations of qubits $\Pi_q$ and numerically minimizing the LE-cost function
	\begin{equation}
		C_{LE} = |\operatorname{Tr}(U(\boldsymbol{\phi})\Pi_qP_{c}\Pi^{-1}_qU^{\dagger}(\boldsymbol{\phi})P'_{c}) - K|^2,
	\end{equation}
	where 
	\begin{equation}
		\begin{aligned}
			U(\boldsymbol{\phi}) = &\exp({-i\sum_{j=1}^n \phi_{j,1}Z_j})\exp({-i\sum_{j=1}^n \phi_{j,2}X_j})\\& \cdot \exp({-i\sum_{j=1}^n \phi_{j,3}Z_j})
		\end{aligned}
	\end{equation}
	is a product of single-qubit unitaries with $3n$ parameters $\boldsymbol{\phi}$. If there exist $\Pi_q$ and $\boldsymbol{\phi}$ such that $C_{LE} < 1 \times 10^{-10}$, we say $P_{c}$ and $P'_{c}$ are (locally) equivalent.

	The minimum code length that protects a logical qubit against arbitrary one-qubit errors is $n=5$. The $((5,2,3))_2$ code we rediscover is equivalent to the perfect code devised in Ref.~\cite{laflamme1996perfect}. This code is unique and translational invariant. The $((5,6,2))_2$ code we rediscover is equivalent to the original non-additive CWS code devised in Ref.~\cite{rains1997nonadditive}. For parameters $((6,2,3))_2$, we sample different initial VQC parameters $\boldsymbol{\theta}$ and find a mass of non-additive codes that are not mutually equivalent. This is consistent with our observation that an infinite family of non-equivalent $((6,2,3))_2$ codes exist. For parameters $((7,2,3))_2$, we find non-equivalent quantum codes and some of them are not equivalent to CWS codes. See Appendix~\ref{ap:623} for more discussions on $((6,2,3))_2$ and $((7,2,3))_2$. The $((8,8,3))_2$ code we rediscover is equivalent to the additive code stabilized by
	\begin{equation}
		\begin{array}{llllllllll}
			g_{1} & = & X & X & X & X & X & X & X & X\\
			g_{2} & = & Z & Z & Z & Z & Z & Z & Z & Z\\
			g_{3} & = & I & X & Y & Z & Z & Y & X & I\\
			g_{4} & = & Z & Y & Z & Y & X & I & X & I \\
			g_{5} & = & X & Y & Y & X & I & Z & Z & I
		\end{array}
	\end{equation}
	up to permutation of qubits.

	It is an open question whether a $((7,3,3))_2$ QECC exists. We have not yet found such a code with VarQEC, even if using an overparameterized VQC ($L=31$) that is capable of finding any $((7,3))_2$ quantum code and sampling 20000 optimization starting points. This strongly indicates that a quantum code with parameters $((7,3,3))_2$ is nonexistent.

	\subsection{Asymmetric codes}
	In quantum experiments, the decoherence time of a physical qubit is mainly influenced by two factors: the relaxation time $T_1$ and the dephasing time $T_2$. Relaxation leads to all Pauli errors, whereas dephasing only leads to phase-flips (Pauli-$Z$ errors). Denote the probabilities of $X$, $Y$, and $Z$ errors as $p_x$, $p_y$, $p_z$ respectively. Usually, $p_x = p_y \neq p_z$. The asymmetry between $X/Y$ and $Z$ errors motivates people to construct \textit{asymmetric} QECCs that handle them differently~\cite{steane1996simple, ioffe2007asymmetric}. Asymmetric codes are more resource-efficient since they can detect/correct more Pauli-$X$/$Y$ errors than Pauli-$Z$ errors or vice versa in response to demand. Researchers have extended several constructions from symmetric codes to asymmetric codes~\cite{ioffe2007asymmetric, sarvepalli2008asymmetric, sarvepalli2009asymmetric, ezerman2011additive, ezerman2013css, jackson2016codeword, bonilla2021xzzx}. Note that the classification of symmetric and asymmetric codes depends on the error detecting/correcting capability instead of the code construction method.

	For a system with $X/Y$-error probabilities $p_x=p_y$ and $Z$-error probability $p_z$. We set the bias parameter $c_{\text{Z}}$ as
	\begin{equation}
		c_{\text{Z}} = \frac{\log p_z}{\log p_x}.
	\end{equation}
	such that $p_z = p_x^{c_{\text{Z}}}$. 
	
	In most scenarios, dephasing is dominating and phase-flip errors are more prevalent than $X/Y$ errors. Accordingly, $0<c_{\text{Z}}<1$. First, we fix $c_{\text{Z}} = 1/2$ (i.e., $p_z \approx p_x^{1/2}$) and apply VarQEC to find $((n,K,d_e(1/2)))_2$ codes that encodes one logical qubit ($K=2$) or one logical qutrit ($K=3$). We discover asymmetric codes
	\begin{equation}
		\begin{aligned}
			&((6,2,d_e(\frac{1}{2})=2))_2, \\&((7,3,d_e(\frac{1}{2})=2))_2.
		\end{aligned}
	\end{equation}
	They can detect more $Z$ errors than $X/Y$ errors, specifically, detect the error set
	\begin{equation}
		\begin{aligned}
			\mathcal{E}_{1/2}^{\{2\}}=\{I\} \cup\{&X_{j}, Y_{j}, Z_{j}, X_{i} Z_{j}, Y_{i} Z_{j}, Z_iZ_jZ_k \}
		\end{aligned}
	\end{equation}
	with indices $i,j,k \in [1,n]$.

	We now consider the opposite situation where $X/Y$ errors are more prevalent than $Z$. In the extreme case, $T_2 \to +\infty$, the only source of decoherence is qubit relaxation. This process at finite temperature is modeled by the generalized amplitude damping channel. Its Kraus representation has operators
        \allowdisplaybreaks
        \begin{alignat}{5}
			A_{0}=&\sqrt{p}\left(\begin{array}{cc}
				1 & 0 \\
				0 & \sqrt{1-\gamma}
			\end{array}\right)\nonumber\\=&\sqrt{p}I-\frac{\sqrt{p}\gamma}{4}(I-Z)+\mathcal{O}\left(\sqrt{p}\gamma^{2}\right),\nonumber\\
			A_{1}=&\sqrt{p}\left(\begin{array}{cc}
				0 & \sqrt{\gamma} \\
				0 & 0
			\end{array}\right)=\frac{\sqrt{p\gamma}}{2}(X+i Y), \nonumber\\
			\quad A_{2}=&\sqrt{1-p}\left(\begin{array}{cc}
				\sqrt{1-\gamma} & 0 \\
				0 & 1
			\end{array}\right)\nonumber\\=&\sqrt{1-p}I-\frac{\sqrt{1-p}\gamma}{4}(I+Z)\nonumber\\&+\mathcal{O}\left(\sqrt{1-p}\gamma^{2}\right),\nonumber\\
			A_{3}=&\sqrt{1-p}\left(\begin{array}{cc}
				0 & 0 \\
				\sqrt{\gamma} & 0
			\end{array}\right)=\frac{\sqrt{\gamma-p\gamma}}{2}(X-i Y),
	 \end{alignat}
	where $\gamma$ is the damping rate, $p$ is a constant determined by the temperature. $A_{0}$ and $A_{2}$ introduce Pauli-$Z$ errors of order $\mathcal{O}(\gamma)$, $A_{1}$ and $A_{3}$ introduce Pauli-$X$ and -$Y$ errors of order $\mathcal{O}(\sqrt{\gamma})$. When $\gamma$ is small, $c_{\text{Z}} = \log p_z/\log p_x \approx 2$. 
	
	Now we fix $c_{\text{Z}} = 2$ and apply VarQEC to find asymmetric codes with $2$-effective distance $3$. We rediscover codes with parameters 
	\begin{equation}
		\begin{aligned}
			&((5,2,d_e(2)=3))_2, \\&((6,4,d_e(2)=3))_2, \\&((7,8,d_e(2)=3))_2.
		\end{aligned}
	\end{equation}
	These codes were introduced in Ref.~\cite{jackson2016concatenated}. They can detect more $X/Y$ errors than $Z$ errors, i.e., the error set
	\begin{equation}
		\mathcal{E}_2^{\{3\}}=\{I\} \cup\left\{X_{j}, Y_{j}, Z_{j}, X_{i} X_{j}, X_{i} Y_{j}, Y_{i} Y_{j}\right\}
	\end{equation}
	with indices $i,j \in [1,n]$.
	
	Furthermore, we find new codes with $2$-effective distance $4$ with $K=2$ or $K=3$, i.e.,
	\begin{equation}
		\begin{aligned}
			&((6,2,d_e(2)=4))_2, \\&((8,3,d_e(2)=4))_2.
		\end{aligned}
	\end{equation}
	Some $((6,2,d_e(2)=4))_2$ codes are equivalent to the additive $((6,2,3))_2$ code stabilized by
	\begin{equation}
		\begin{array}{llllllllllllll}
			g_{1} & = &X&I&X&Y&Z&X\\
			g_{2} & = &Z&I&I&I&I&Z\\
			g_{3} & = &I&X&X&X&X&I\\
			g_{4} & = &I&Z&I&Y&X&Z\\
			g_{5} & = &I&I&Z&X&Y&Z.
		\end{array}
	\end{equation}
	They can detect the error set
	\begin{equation}
		\begin{aligned}
			\mathcal{E}_2^{\{4\}}=\mathcal{E}_2^{\{3\}} \cup\{&X_{i} Z_{j}, Y_{i} Z_{j}, X_iX_jX_k, X_iX_jY_k, \\&X_iY_jY_k, Y_iY_jY_k \}.
		\end{aligned}
	\end{equation}
	with indices $i,j,k \in [1,n]$. For the generalized amplitude damping channel, $((6,2,d_e(2)=4))_2$ and $((8,3,d_e(2)=4))_2$ can detect up to three $A_{1}$/$A_{3}$ errors or one $A_0$/$A_2$ error, and correct one $A_{1}$/$A_{3}$ error. Assisted by post-selection~\cite{prabhu2021distancefour}, these codes hold the promise to achieve lower logical error rate than codes with $d_e(2)=3$.

	\subsection{Channel-adaptive codes}
	In the previous sections, we have only discussed uncorrelated errors, symmetric or asymmetric. This section considers quantum channels with correlated noise. We apply VarQEC to find the corresponding channel-adaptive codes.
	
	Correlated errors are ubiquitous in quantum computing experiments. When two adjacent qubits are not sufficiently separated, the errors occurring on them can be highly correlated~\cite{wilen2021correlated}. These spatially correlated errors invalidate many well-known quantum codes and dim the hope of fault-tolerant quantum computing. Suppose we ignore the exact connectivity graph and the noise type. In that case, we need at least $11$ physical qubits to protect one qubit of information from general correlated errors (i.e., the double error-correcting $((11,2,5))_2$ code)~\cite{calderbank1998quantum}. Even so, the encoding isometry may not be hardware-efficient. In the following, we investigate two correlated noise channels in detail and introduce channel-adaptive codes discovered by VarQEC.

	\subsubsection{Nearest-neighbor collective amplitude damping}\label{Sec:collective-ad}
	The first testbed is the \textit{nearest-neighbor collective amplitude damping channel}. Suppose we have $n$ qubits in a ring, as shown in Fig.~\ref{fig:ring}. Every two neighboring qubits collectively interact with a single environment and exhibit collective dynamics of amplitude damping~\cite{hama2020quantum, grassl2018quantum}. The corresponding Kraus operators are 
	\begin{equation} 
		\begin{aligned}
			K_{0}=&\sqrt{\frac{\gamma_{01}}{2}} |00\rangle (\langle 01| + \langle10|)\\&+\sqrt{\frac{\gamma_{12}}{2}} (|01\rangle + |10\rangle) \langle 11|,\\
			K_{1}=&\sqrt{\gamma_{02}} |00\rangle \langle 11|, \\
			K_2 =  & \sqrt{\frac{1-\gamma_{01}}{4}} (|01\rangle + |10\rangle)(\langle 01| + \langle10|)\\
			&+ \sqrt{1 - \gamma_{02} - \gamma_{12}} |11\rangle\langle11| \\& + \frac{1}{2} (|01\rangle - |10\rangle)(\langle 01| - \langle10|)\\ & + |00\rangle\langle00|,
		\end{aligned}
	\end{equation}
	where $\gamma_{01}$, $\gamma_{02}$, $\gamma_{12}$ are damping rates. For a short decay time $\tau$, $\gamma_{01}$ and $\gamma_{12}$ are of order $\mathcal{O}(\tau)$, $\gamma_{02}$ is of order $\mathcal{O}(\tau^2)$. Each error acts on two neighbouring qubits $Q_j$-$Q_{j+1}$.
	To find quantum codes that approximately correct one nearest-neighbor collective amplitude damping error, we expand the above Kraus operators with respect to $\tau$, abandon trivial/high-order terms and obtain
	\begin{equation} 
		\begin{aligned}
			K'_{0}&=  \frac{1}{\sqrt{2}}|00\rangle (\langle 01| + \langle10|) + \frac{1}{\sqrt{2}}(|01\rangle + |10\rangle) \langle 11|,\\
			K'_{1}&=  |00\rangle\langle 11|,\\
			K'_{2}&= \frac{1}{2}(|01\rangle + |10\rangle)(\langle 01| + \langle10|) +  |11\rangle\langle 11|.
		\end{aligned}
	\end{equation}
	Each $K'_{0}$ contributes one factor of $\sqrt{\tau}$, each $K'_{1}$/$K'_{2}$  contributes one factor of $\tau$. Suppose $E_{\alpha}$ and $E_{\beta}$ are products of the identity, $K'_{0}$, $K'_{1}$, and $K'_{2}$. The target error set (in VarQEC)
	\begin{equation}
		\mathcal{E} = \left\{E^{\dagger}_{\alpha} E_{\beta}\right\}
	\end{equation}
	consists of terms with total order less than $\tau^{3/2}$.
	\begin{figure}[t]
		\centering
		\includegraphics[width=8.05cm]{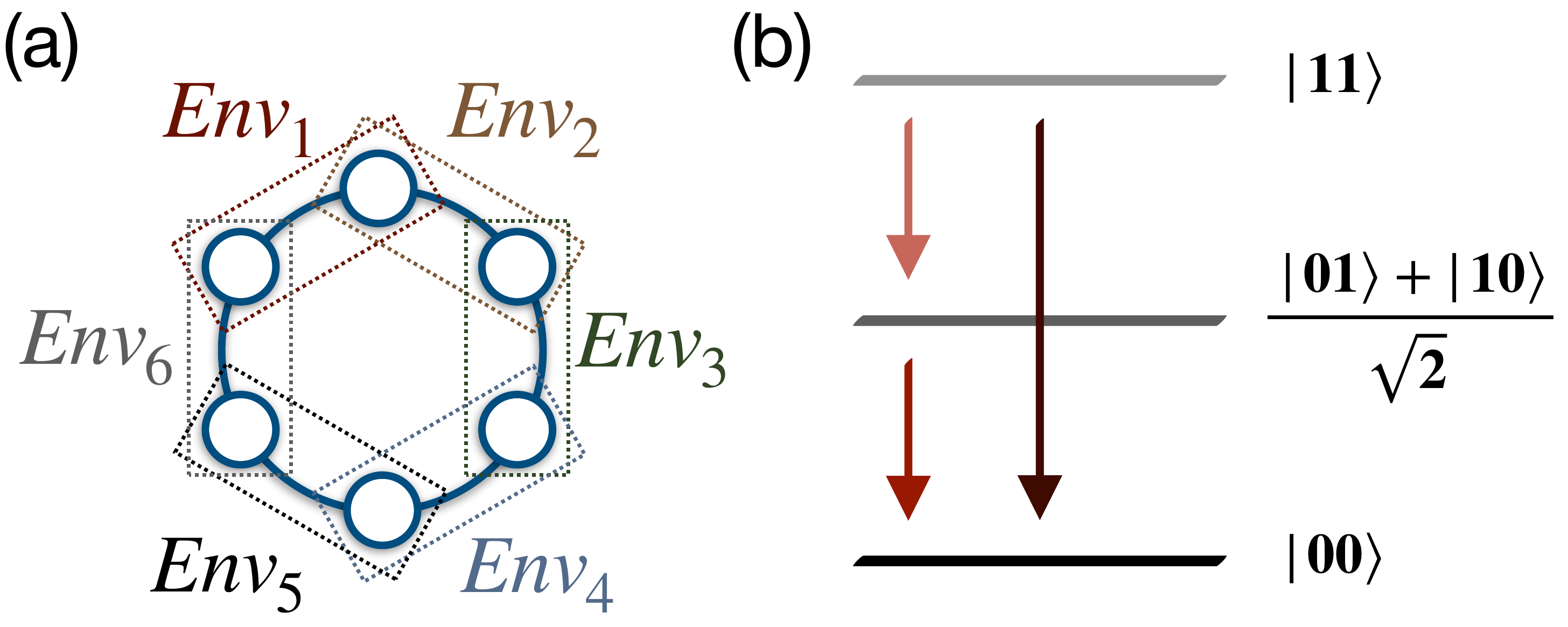}
		\caption{Schematic of nearest-neighbor collective amplitude damping. (a) Qubits in a ring, neighboring two qubits collectively interact with a single environment. (b) Decay processes.}
		\label{fig:ring}
	\end{figure}
	Note that here, some error products $E^{\dagger}_{\alpha} E_{\beta}$ are non-unitary and non-Hermitian. To compute cost functions $C^{\ell_{1}}_{n, K,\mathcal{E}}$ and $C^{\ell_{2}}_{n, K,\mathcal{E}}$ (Eqs.~\eqref{l1norm},~\eqref{l2norm}), we need to estimate
	\begin{equation}
		\langle\psi_j|E^{\dagger}_{\alpha} E_{\beta}\ |\psi_j\rangle
	\end{equation}
	and
	\begin{equation}
		\big|\langle\psi_i|E^{\dagger}_{\alpha} E_{\beta}\ |\psi_j\rangle\big|
	\end{equation}
	for various $i,j,\alpha,\beta$ ($i\neq j$). 
	
	$\langle\psi_j|E^{\dagger}_{\alpha} E_{\beta}\ |\psi_j\rangle$ is a complex number that can be obtained as follows: we prepare the state $|\psi_j\rangle$ and measure two Hermitian observables $(E^{\dagger}_{\alpha} E_{\beta}+E^{\dagger}_{\beta} E_{\alpha})/2$ and $(E^{\dagger}_{\alpha} E_{\beta}-E^{\dagger}_{\beta} E_{\alpha})/2i$. The first expectation value gives the real part of $\langle\psi_j|E^{\dagger}_{\alpha} E_{\beta}\ |\psi_j\rangle$ and the second expectation value gives its imaginary part.
	
	We estimate $\big|\langle\psi_i|E^{\dagger}_{\alpha} E_{\beta}\ |\psi_j\rangle\big|$ with $i \neq j$ using POVMs. Specifically, we prepare the state $|\psi_j\rangle$, add an ancilla qubit and implement the operation
	\begin{equation}
		\begin{aligned}
			\Lambda_{\alpha\beta}(|\psi_j\rangle) = &E^{\dagger}_{\alpha} E_{\beta}|\psi_j\rangle\langle\psi_j| E^{\dagger}_{\beta} E_{\alpha}\otimes |0\rangle\langle0|_{\text{anc}}\\& + E_{\text{aux}}|\psi_j\rangle\langle\psi_j| E^{\dagger}_{\text{aux}}\otimes |1\rangle\langle1|_{\text{anc}},
		\end{aligned}
	\end{equation}
	where $E_{\text{aux}}$ is an auxiliary Kraus operator such that $E^{\dagger}_{\beta} E_{\alpha}E^{\dagger}_{\alpha} E_{\beta} + E^{\dagger}_{\text{aux}}E_{\text{aux}} = I$. Then we measure the ancilla qubit in the computational basis and post-select the cases of $|0\rangle$. Measuring the ancilla qubit in the $|0\rangle$ state indicates that the error $E^{\dagger}_{\alpha} E_{\beta}$ has occurred.  The corresponding probability is 
	\begin{equation}
		p_0 = \operatorname{Tr}(E^{\dagger}_{\alpha} E_{\beta}|\psi_j\rangle\langle\psi_j| E^{\dagger}_{\beta} E_{\alpha})
	\end{equation}
	and the corresponding state is
	\begin{equation}
		\frac{E^{\dagger}_{\alpha} E_{\beta}|\psi_j\rangle\langle\psi_j| E^{\dagger}_{\beta} E_{\alpha}}{p_0}.
	\end{equation}
	For these postselected states, we apply the inverse of the VQC and do projective measurements. The conditional probability of obtaining the binary string $|\mathbf{i-1}\rangle|0\rangle^{\otimes(n-k)}$ is
	\begin{equation}
		\begin{aligned}
			p_{ij} = &\frac{\big|\langle\psi_i|E^{\dagger}_{\alpha} E_{\beta}|\psi_j\rangle\big|^2}{p_0}.
		\end{aligned}
	\end{equation}
	Therefore, we can estimate $\big|\langle\psi_i|E^{\dagger}_{\alpha} E_{\beta}|\psi_j\rangle\big|$ by $\sqrt{p_{ij}p_0}$.
	
	Suppose 0 or 1 error occurs during a short decay time $\tau$, we apply VarQEC and find quantum codes with length and dimension
	\begin{equation}
		\begin{aligned}
			&((4,3))_2, ((5,2))_2\\
			&((6,5))_2, ((7,8))_2\\
			&((8,9))_2, ((9,16))_2.
		\end{aligned}
	\end{equation}
	These codes can reduce the error from $\mathcal{O}(\sqrt{\tau})$ to $\mathcal{O}(\tau)$.

	\subsubsection{Nearest-neighbor collective phase-flips}
	The second noise channel we consider is a combined channel of \textit{nearest-neighbor collective phase-flips} and \textit{single-qubit errors}. The channel consists of two stages. In the first stage, a local depolarizing error with noise rate $p$
	\begin{equation}\label{N_DP}
		\mathcal{N}_{\mathrm{DP}_j}(\rho) = (1-\frac{3p}{4})\rho + \frac{p}{4}(X_j\rho X_j + Y_j\rho Y_j + Z_j\rho Z_j)
	\end{equation} 
	occurs on each qubit. In other words, Pauli errors $X,Y, Z$ occur on each qubit with probability $p/4$. Different local errors act independently. We denote the corresponding global noise channel as 
	\begin{equation}\label{N_1}
		\mathcal{N}_1 = \prod_{j=1}^n \mathcal{N}_{\mathrm{DP}_j}
	\end{equation}
	In the second stage, nearest-neighbor collective phase-flip errors $ZZ$ with noise rate $p_{zz}$
	\begin{equation}\label{N_ZZ}
		\mathcal{N}_{\mathrm{Z}_i\mathrm{Z}_j}(\rho) = (1-p)\rho + p_{\mathrm{zz}}Z_iZ_j \rho Z_iZ_j
	\end{equation}
	occur on adjacent qubit pairs $Q_i$-$Q_j$. We denote the corresponding global noise channel as 
	\begin{equation}\label{N_2}
		\mathcal{N}_2 = \prod_{\langle i,j \rangle}\mathcal{N}_{\mathrm{Z}_i\mathrm{Z}_j}.
	\end{equation}
	The overall process is
	\begin{equation}\label{overall_N}
		\mathcal{N} = \mathcal{N}_2 \circ \mathcal{N}_1.
	\end{equation}
	
	Directly applying VarQEC to $\mathcal{N}$ is not resource efficient since the Kraus representation of $\mathcal{N}$ consists of $\mathcal{O}(\exp(n))$ operators. For practical purposes, we apply our algorithm to the following channel instead,
	\begin{equation}
		\begin{aligned}
			\mathcal{N}'(\rho) = &(1-\sum_j \frac{3p}{4}-\sum_{\langle i,j \rangle}p_{\mathrm{zz}})\rho\\&+\sum_j \frac{p}{4}(X_j\rho X_j +  Y_j\rho Y_j +  Z_j\rho Z_j)\\&+\sum_{\langle i,j \rangle}p_{\mathrm{zz}}Z_iZ_j \rho Z_iZ_j.
		\end{aligned}
	\end{equation}
	The second term takes summation over all qubits. The last term takes summation over all adjacent qubit pairs $\langle i,j \rangle$. Its Kraus operators are
	\begin{equation}
		\begin{aligned}\label{eq:E_p}
			\mathcal{E}' = \Big\{&\sqrt{(1-\sum_j \frac{3p}{4}-\sum_{\langle i,j \rangle}p_{\mathrm{zz}})}I, \\&\sqrt{\frac{p}{4}}X_j, \sqrt{\frac{p}{4}}Y_j, \sqrt{\frac{p}{4}}Z_j, \sqrt{p_{\mathrm{zz}}}Z_iZ_j\Big\}.
		\end{aligned}
	\end{equation}
	where qubit-$i$ and qubit-$j$ are adjacent. $\mathcal{N}'$ is the first-order approximation of $\mathcal{N}$ with respect to the error parameters $p$ and $p_{zz}$. The Kraus representation of $\mathcal{N}'$ consists of only poly($n$) operators. $\mathcal{N}$ and $\mathcal{N}'$ are equivalent in the zero-noise limit,
	\begin{equation}
		\lim_{p,p_{\mathrm{zz}}\rightarrow 0} \mathcal{N}' = \mathcal{N}.
	\end{equation}

	\begin{figure*}[t]
		\centerline{\includegraphics[width=0.7\textwidth]{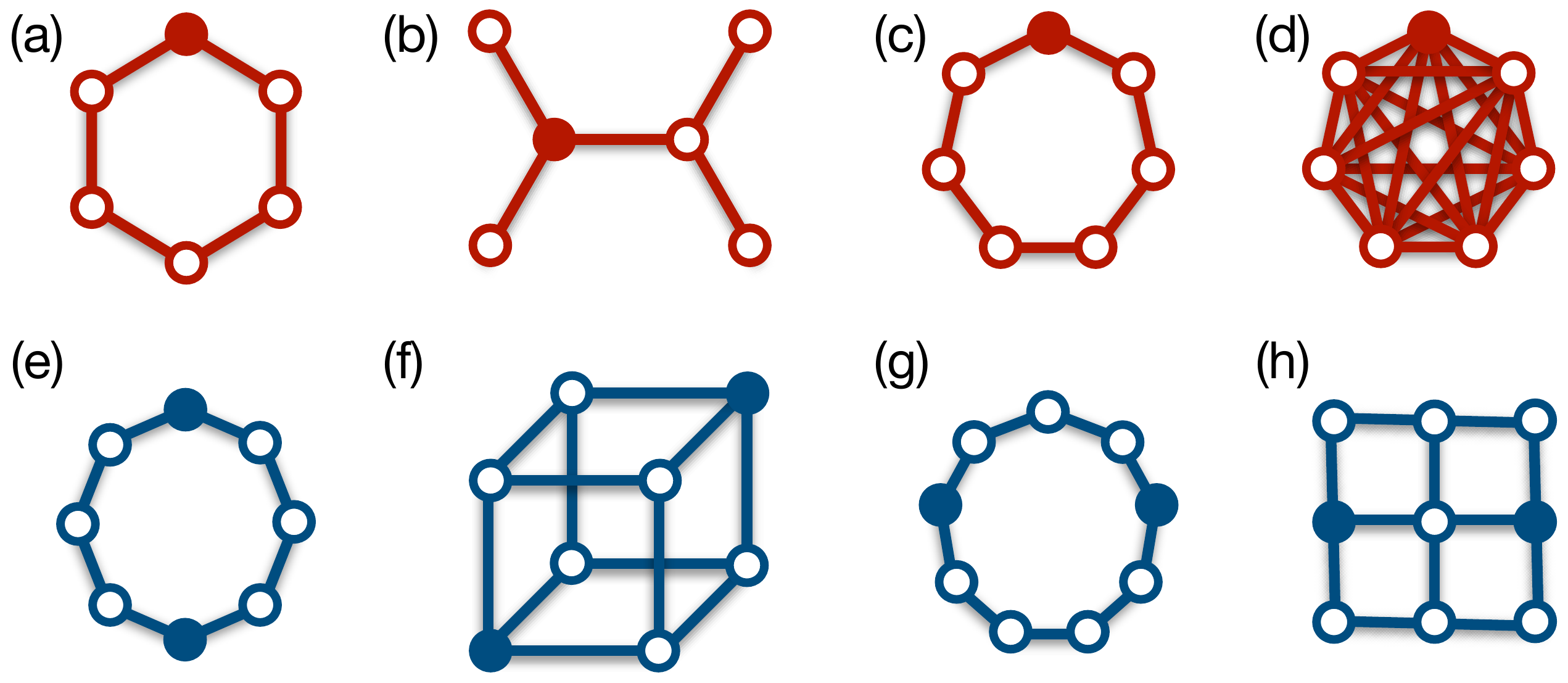}}
		\caption{Hardware connectivity graphs with $6 \sim 9$ vertices (physical qubits). Filled circles represent the initial $k$ qubits to prepare the logical data. For these graphs, there exist channel-adaptive codes to protect $k$ qubit(s) of information from general one-qubit errors and nearest-neighbor collective phase-flips.  The code length $n$, the code dimension $K$, and the number of VQC layers $L$ are (a) $n=6, K = 2, L = 5$; (b) $n=6,K = 2, L = 6$; (c) $n=7,K = 2, L = 2$; (d) $n=7,K = 2, L = 2$; (e) $n=8,K = 4, L = 4$; (f) $n=8,K = 4, L = 4$; (g) $n=9,K = 4, L = 3$; (h) $n=9,K = 4, L = 2$.}
		\label{fig:code_str}
	\end{figure*}

	Suppose for $\mathcal{N}'$, we find an $\varepsilon$-correctable approximate code with  $\varepsilon \ll 1$. Namely, with appropriate recovery $\mathcal{R}$, the entanglement fidelity is
	\begin{equation}
		F_e(\mathcal{R}\mathcal{N}') \geq 1- \varepsilon.
	\end{equation}
	Then for the original noise channel $\mathcal{N}$, the entanglement fidelity naturally has the form
	\begin{equation}
		\begin{aligned}
			F_e(\mathcal{R}\mathcal{N}) = &1- \mathcal{O}(\varepsilon p) -\mathcal{O}(\varepsilon p_{\mathrm{zz}}) \\ &-\mathcal{O}(p^2)-\mathcal{O}(p_{\mathrm{zz}}^2).
		\end{aligned}
	\end{equation}
	The QECC can push the first-order errors down to an extremely small level. To find quantum codes that correct multiple errors, we can choose a higher-order approximation and similarly implement VarQEC.

	Given a connectivity graph $G$ with edge number $|E(G)|$ and maximum vertex degree $\Delta(G)$. In the following, we set
	\begin{equation}
		p_{\mathrm{zz}}=0.99/(3n+|E(G)|),\quad p = 4p_{\mathrm{zz}}.
	\end{equation}
	For a generic input state $\rho$, the probability of receiving the same state after going through the noise channel is about $\sim 0.01$. The target error list to detect in VarQEC is 
	\begin{equation}
		\mathcal{E}=\left\{E^{\dagger}_{\alpha} E_{\beta}| E_{\alpha}, E_{\beta} \in \mathcal{E}'\right\}.
	\end{equation}
	Still, we use the VQC ansatz with alternating layers of single-qubit rotations $R_x$-$R_z$ acting on all qubits and Ising-type interactions $R_{zz}$ acting on adjacent qubits. The circuit depth of a VQC with $L$ layers is of order $\mathcal{O}(L\Delta(G))$.

	After adequate optimization, we find approximate channel-adaptive codes for $\mathcal{N} = \mathcal{N}_2 \circ \mathcal{N}_1$ with hardware connectivity graphs shown in Fig.~\ref{fig:code_str}. The codes for graphs (a,b,h) are degenerate, and the others are non-degenerate. Six physical qubits suffice to encode one logical qubit, and eight physical qubits suffice to encode two logical qubits.

	Note that up to a local unitary transformation, these codes can correct an arbitrary single-qubit error followed by an adjacent $U \otimes U$ error for any fixed $U \in \operatorname{U}(2)$ with eigenvalues $\{-1,1\}$.
	
	We investigated the codes for graphs (c,d) in more detail. Clearly, they have code parameters $((7,2,3))_2$. We calculate their \textit{quantum weight enumerators}~\cite{shor1997quantum}, which were defined by
	\begin{equation}
		A(z) = \sum_{j=0}^{n} A_j z^j, \quad B(z) = \sum_{j=0}^{n} B_j z^j
	\end{equation}
	with coefficients
	\begin{equation}
		A_j = \frac{1}{K^2}\sum_{\operatorname{wt}(O_{\alpha})=j} \operatorname{Tr}(O_{\alpha}P_c)\operatorname{Tr}(O_{\alpha}^{\dagger}P_c),
	\end{equation}
	\begin{equation}
		B_j = \frac{1}{K}\sum_{\operatorname{wt}(O_{\alpha})=j} \operatorname{Tr}(O_{\alpha}P_cO_{\alpha}^{\dagger}P_c).\mkern44mu
	\end{equation}
	These two codes are locally equivalent and therefore have the same weight enumerators, i.e.,
	
	\begin{equation}
		A(z) = 1 + 2z^3 + 9z^4 + 24z^5+22z^6+6z^7,\mkern26mu
	\end{equation}
	\begin{equation}
		B(z) = 1 + 17z^3 + 45z^4 + 78z^5+82z^6+33z^7.
	\end{equation}
	
	Further, we verified that they are locally equivalent to a non-degenerate additive code stabilized by
	\begin{equation}
		\begin{array}{llllllllllllll}
			g_{1} & = &X&I&Z&X&X&I&X\\
			g_{2} & = &Z&I&I&X&X&X&Z\\
			g_{3} & = &I&X&Z&X&Z&Z&Z\\
			g_{4} & = &I&Z&Z&I&Z&Y&Z\\
			g_{5} & = &I&I&Y&X&Z&I&X\\
			g_{6} & = &I&I&I&Z&Y&Y&X
		\end{array}
	\end{equation}
	up to permutation of qubits. This additive code can correct arbitrary single-qubit errors and 2-qubit collective phase-flips occurring on any qubit pairs, i.e., the error set
	\begin{equation}
		\begin{aligned}\label{eq:E_add_ZZ}
			\mathcal{E} = \Big\{I, X_j, Y_j, Z_j, Z_iZ_j\Big\}.
		\end{aligned}
	\end{equation}
	with indices $i,j \in [1,n]$. 
	
	According to the quantum Hamming bound, for one logical qubit, no non-degenerate quantum code with code length $n < 7$ can correct arbitrary single-qubit errors as well as 2-qubit collective phase-flips since
	\begin{equation}
		2^n \geq K(3n + \frac{n(n-1)}{2})
	\end{equation}
	with $K=2$ only holds when $n \geq 7$.
	
	Two $((7,2,3))_2$ stabilizer codes were investigated in detail. One is the famous Steane code~\cite{steane1996multiple} based on the Calderbank-Shor-Steane (CSS) construction. The other is a non-CSS code found by numerical greedy search, called the bare code~\cite{li2017fault}. Their weight enumerators are as follows,
	\begin{alignat}{5}
		A^{\{\mathrm{Steane}\}}(z) &= 1 + 21z^4 + 42z^6,\nonumber\\
		B^{\{\mathrm{Steane}\}}(z) &= 1 + 21z^3 + 21z^4 \nonumber\\&\quad+ 126z^5+ 42z^6+ 45z^7,
  	\end{alignat}
        and
	\begin{alignat}{5}
        A^{\{\mathrm{bare}\}}(z) = &1 + 5z^2+ 11z^4+47z^6,\nonumber\\
		B^{\{\mathrm{bare}\}}(z) = &1 + 5z^2+ 36z^3\nonumber\\&\quad+ 11z^4+ 96z^5+ 47z^6.
	\end{alignat}
	QECCs with different weight enumerators are not locally and translationally equivalent. Our code is different from the Steane and the bare $((7,2,3))_2$ codes. See Appendix~\ref{ap:weight enumerators} more weight enumerators. 
	
	The Steane and the bare codes cannot correct nearest-neighbor collective phase-flips. For the combined channel of nearest-neighbor collective phase-flips with noise rate $p_{\text{zz}}$ and single-qubit errors with noise rate $p$, the entanglement fidelity of our code is of the form 
	\begin{equation}
		\begin{aligned}
			F_e(\mathcal{R}\mathcal{N}) = &1 -\mathcal{O}(p^2)-\mathcal{O}(p_{\mathrm{zz}}^2)
		\end{aligned}
	\end{equation}
	whereas the entanglement fidelity of the Steane and the bare codes is of the form
	\begin{equation}
		\begin{aligned}
			F_e(\mathcal{R}\mathcal{N}) = &1 -\mathcal{O}(p_{\mathrm{zz}})  -\mathcal{O}(p^2) -\mathcal{O}(p_{\mathrm{zz}}^2).
		\end{aligned}
	\end{equation}
	
	Although our code was written in a quantum simulator that has not been open-sourced yet, we rewrote some example implementations with Qiskit, an open-source software development kit. They are available on GitHub~\cite{VarQECcode}.

	\section{Noise Resilience}\label{Sec:noise-resilience}
	Although the previously introduced results are obtained by numerical simulation, VarQEC is a hybrid quantum-classical algorithm meant to be run on NISQ devices where quantum gates are inevitably noisy. In this section, we demonstrate that VarQEC is pretty resilient to random gate errors. As long as the error rate $p_{\mathrm{gate}}$ is below a reasonable threshold, VarQEC can find an efficient encoding circuit that prepares the correct code. This resilience is essentially analogous to the noise resilience in variational quantum compiling~\cite{sharma2020noise}.
	
	We start from the simplest noise model, global depolarizing, and introduce the following theorem.
	\begin{theorem}\label{noise_resilience}
		Suppose the variational quantum circuit in VarQEC is accompanied by global depolarizing noise acting continuously throughout the circuit. If the ideal circuit is capable of finding an eligible quantum code, after adequate optimization with the noisy circuit, the output parameters $\boldsymbol{\theta}'_\mathrm{opt}$ are still correct.
	\end{theorem}
	\begin{proof}
		Consider the cost functions in Eqs.~\eqref{l1norm},~\eqref{l2norm}.
		Due to the global depolarizing noise, when we run the VQC $U(\boldsymbol{\theta})$ to prepare a basis state $|\psi_j\rangle$, we instead obtain
		\begin{equation}
			\rho_j = (1-\epsilon_1)|\psi_j\rangle\langle\psi_j| + \epsilon_1\frac{I}{2^n};
		\end{equation}
		when we apply $U^{\dagger}(\boldsymbol{\theta})E_{\mu} U(\boldsymbol{\theta})$ to an initial binary string $|\mathbf{j-1}\rangle|\mathbf{0}\rangle$ to prepare the output state $|\psi_{j,\mu}\rangle$, we instead obtain
		\begin{equation}
			\rho_{j,\mu} = (1-\epsilon_2)|\psi_{j,\mu}\rangle\langle\psi_{j,\mu}| +  \epsilon_2\frac{I}{2^n}.
		\end{equation}
		Noise rates $\epsilon_1$ and $\epsilon_2$ are unknown constants determined by the circuit depth. After adequate optimization with the noisy VQC, we obtain the pseudo-optimal parameters
		\begin{equation}\label{pseuso_params}
			\begin{aligned}
				\boldsymbol{\theta}'_\mathrm{opt}  = &\arg \min _{\boldsymbol{\theta}} \\& \sum_{E_{\mu} \in \mathcal{E}}  \Big(\sum_{1 \leq i<j \leq K}\sqrt{\langle\mathbf{i-1}|\langle\mathbf{0}|\rho_{j,\mu}|\mathbf{i-1}\rangle|\mathbf{0}\rangle} \\
				&+\sum_{j=1}^K  \big|\operatorname{Tr}(\rho_jE_{\mu})-  \sum_{i=1}^K\operatorname{Tr}(\rho_iE_{\mu})/K  \big|/2 \Big)\\
				=&\arg \min_{\boldsymbol{\theta}} \\& \sum_{E_{\mu} \in \mathcal{E}}   \Big(\sum_{1 \leq i<j \leq K}\sqrt{(1-\epsilon_2)\big|\langle\psi_i|E_{\mu} |\psi_j\rangle\big|^2 + \frac{\epsilon_2}{2^n}}\\
				&+\sum_{j=1}^K \frac{1-\epsilon_1}{2} \big|\langle\psi_j|E_{\mu}|\psi_j\rangle-\overline{\langle E_{\mu} \rangle}\big| \Big).
			\end{aligned}
		\end{equation}
		Since the ideal variational quantum circuit is capable of finding an eligible quantum code, each term in the cost function Eq.~\eqref{l1norm} can be minimized to 0 (i.e., $\big|\langle\psi_i|E_{\mu} |\psi_j\rangle\big|=0$, $\big|\langle\psi_j|E_{\mu}|\psi_j\rangle-\overline{\langle E_{\mu} \rangle}\big|=0$). Comparing Eq.~\eqref{l1norm} and Eq.~\eqref{pseuso_params}, we conclude that
		\begin{equation}
			\boldsymbol{\theta}'_\mathrm{opt} = \boldsymbol{\theta}_\mathrm{opt}.
		\end{equation}
	\end{proof}
	
	To sum up, VarQEC is perfectly resilient to global depolarizing noise, i.e., it can find the correct encoding circuit in the presence of global depolarizing.
	
	\begin{figure}
		\centering
		\includegraphics[width=7cm]{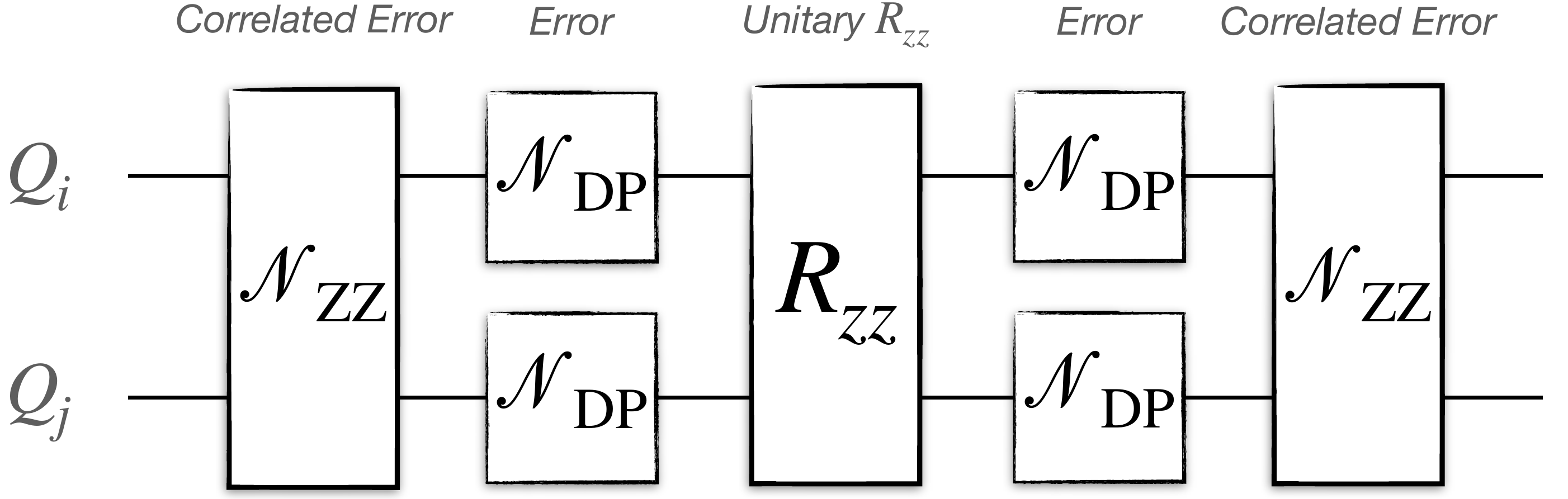}
		\caption{Gate error model. $R_{zz}$ is accompanied by local depolarizing noise $\mathcal{N}_{\mathrm{DP}}$ and collective phase-flip error $\mathcal{N}_{\mathrm{ZZ}}$ both before and after.}\label{fig:noise_gate}
	\end{figure}
	
	In practical scenarios, circuit noise is more complicated and single-qubit errors dominate. Now we consider a more realistic model. Suppose each 2-qubit $R_{zz}$ gate in the VQC is accompanied by local depolarizing noise and collective phase flips, as illustrated in Fig.~\ref{fig:noise_gate}. Before the ideal unitary $R_{zz}$, the two qubits goes through $\mathcal{N}_{\mathrm{DP}}^{\otimes2} \circ \mathcal{N}_{\mathrm{ZZ}}$, after the ideal $R_{zz}$, the system goes through $\mathcal{N}_{\mathrm{ZZ}} \circ \mathcal{N}_{\mathrm{DP}}^{\otimes2}$. In the following, for gate error rate $p_{\mathrm{gate}}$, we set the error rate of each $\mathcal{N}_{\mathrm{DP}}$ as $p_{\mathrm{gate}}/2$, the error rate of each $\mathcal{N}_{\mathrm{ZZ}}$ as $p_{\mathrm{gate}}/8$.

	Still, we use VarQEC to find channel-adaptive codes for noise channel $\mathcal{N}$ (Eq.~\eqref{overall_N}) with hardware connectivity graphs (c,d) shown in Fig.~\ref{fig:code_str}. The difference is that this time the VQC is noisy. After optimization, we obtain the pseudo-optimal parameters $\boldsymbol{\theta}'_\mathrm{opt}$. It is interesting to note that if we transfer $\boldsymbol{\theta}'_\mathrm{opt}$ to an ideal VQC, the corresponding cost function $C^{\ell_{1}}_{n, K,\mathcal{N}}(\boldsymbol{\theta}'_{\mathrm{opt}})$ can be much smaller than the one we estimated with the noisy VQC. Namely, we find a roughly correct encoder even if we use a noisy VQC in our algorithm. The comparison of cost functions for different gate error rates is given in Fig.~\ref{fig:noisy_enc}(a). The cost reduction for both graphs is obvious. Two-qubit gate error rates on state-of-the-art NISQ computers are about $\sim 10^{-2}$~\cite{ai2021exponential}. One can run our algorithm on current hardware directly.

	\begin{figure}[!htb]
		\centering
		\includegraphics[width=7.8cm]{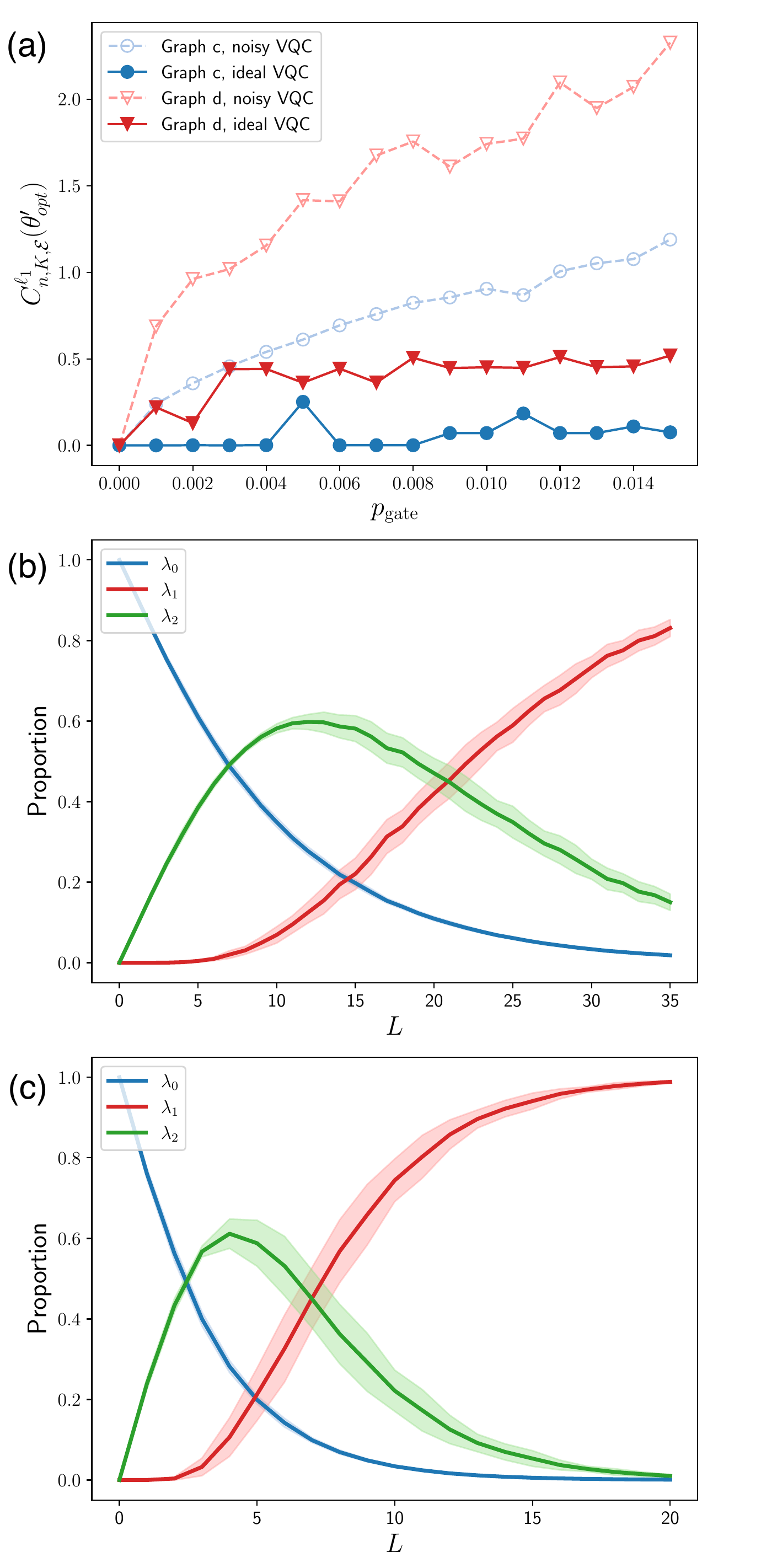}
		\caption{Noise resilience of VarQEC using VQCs with connectivity graphs (c,d) in Fig.~\ref{fig:code_str}. (a) The noisy- and ideal-VQC cost functions versus gate error rate. $\boldsymbol{\theta}'_\mathrm{opt}$ were obtained by optimizing a noisy VQC. (b) Average $\lambda$s (see Eq.~\eqref{noisy_output}) versus the number of VQC layers for graph (c). (c) Average $\lambda$s versus the number of VQC layers for graph (d). The shaded areas represent the standard deviation of 100 samples.}
		\label{fig:noisy_enc}
	\end{figure}

	Suppose the input state of a quantum circuit is $|\psi_{\mathrm{in}}\rangle$, the target unitary evolution is $U_{\mathrm{ideal}}$. The ideal output state is
	\begin{equation}
		|\psi_{\mathrm{ideal}}\rangle = U_{\mathrm{ideal}}|\psi_{\mathrm{in}}\rangle.
	\end{equation}
	However, due to quantum gate errors, the output state $\rho_{\mathrm{out}}$ is a mixed state. We express $\rho_{\mathrm{out}}$ as a summation of three terms,
	\begin{equation}\label{noisy_output}
		\begin{aligned}
			\rho_{\mathrm{out}} &= \mathcal{N}_{\mathrm{circuit}}(|\psi_{\mathrm{in}}\rangle)\\&= \lambda_0 |\psi_{\mathrm{ideal}}\rangle\langle\psi_{\mathrm{ideal}}| + \lambda_1\frac{I}{2^n} + \lambda_2\rho_2,
		\end{aligned}
	\end{equation}
	where $\mathcal{N}_{\mathrm{circuit}}$ denotes that channel of the noisy quantum circuit, $\lambda_1$ is the smallest eigenvalue of $\rho_{\mathrm{out}}$ multiplied by $2^n$, $I/2^n$ is the maximally mixed state, $\rho_2$ is a density operator orthogonal to $|\psi_{\mathrm{ideal}}\rangle$, i.e.,
	\begin{equation}
		\operatorname{Tr}(\rho_1|\psi_{\mathrm{ideal}}\rangle\langle\psi_{\mathrm{ideal}}|) = 0.
	\end{equation}
	The latter two terms of Eq.~\eqref{noisy_output} are both induced by gate errors, but they have different effects on the noise resilience of our algorithm. The second term is a global white noise, as we analyzed in Theorem~\ref{noise_resilience}, it does not affect the optimal parameters. However, the third term $\lambda_2\rho_2$ non-trivially alters the optimization landscape and introduces some local minima. Usually, both the second term and the third term are not negligible. Nevertheless, we are certain about the trend: with the increase of circuit depth, the second term will dominate the third term eventually~\cite{dalzell2021random, deshpande2021tight}. 
	
	For VQCs corresponding to graphs (c,d), we fix gate error rate $0.01$, and try different numbers of layers with randomly sampled $\boldsymbol{\theta}$. The average value of $\lambda'$s are shown in Fig.~\ref{fig:noisy_enc}(b,c). Each point is averaged over 100 samples. Compared with the one-dimensional ring (graph (c)), vertices in the complete graph (graph (d)) are more tightly connected, local errors can be transformed into global white noise more rapidly. For both graphs, $\lambda_1 \ll \lambda_2$ when $L$ is relatively small and $\lambda_1 \gg \lambda_2$ when $L$ is relatively large. With the decrease of gate error rate and the increase of circuit depth, VarQEC will become more resilient to noise. Additionally, one might also consider estimating cost functions in VarQEC more precisely with error mitigation techniques like virtual distillation~\cite{huggins2021virtual, koczor2021exponential}.

	\section{Barren Plateaus}\label{Sec:BP}
	The \textit{barren plateau} (BP)~\cite{mcclean2018barren, cerezo2021cost} and the \textit{noise-induced barren plateau} (NIBP)~\cite{wang2021noise} are two daunting challenges in variational quantum optimization. In this section, we numerically investigate their effects in the VarQEC algorithm. 
	
	The barren plateau is a phenomenon where the gradients vanish exponentially with the increasing number of qubits~\cite{mcclean2018barren}. It occurs when the VQCs form a unitary 2-design, regardless of whether the VQC is noisy or noiseless. Ref.~\cite{cerezo2021cost} connected the locality of the cost function and the trainability of the corresponding VQC. If the cost function is local and the circuit depth is of order $\mathcal{O}(\log(n))$, the BP does not occur (i.e., the VQC is trainable). However, if the cost function is global or the circuit depth is of order $\mathcal{O}(\text{poly}(n))$, a BP occurs in the optimization landscape, and the VQC is untrainable.
	
	In near-term quantum computation, dominating errors always only act on several local qubits. Accordingly, the cost functions $C^{\ell_{1}}_{n, K,\mathcal{E}}$ (Eq.~(\ref{l1norm})) and $C^{\ell_{2}}_{n, K,\mathcal{E}}$ (Eq.~(\ref{l2norm})) are merely influenced by local errors $E_\mu$. Therefore, we expect the same conclusion to hold for VarQEC: the BP does not occur when the circuit depth is of order $\mathcal{O}(\log(n))$ and occurs when the circuit depth is of order $\mathcal{O}(\text{poly}(n))$. 
	
	Without loss of generality, here we use the star connectivity graph $S_{n-1}$ ($S_{7}$ is illustrated in Fig.~\ref{fig:bp}(b) inset) and focus on $C^{\ell_{2}}_{n, 2,\mathcal{E}}$ with $\mathcal{E}=\left\{O_\alpha \mid \mathrm{wt}\left(O_\alpha\right)<3 \right\}$, i.e., searching for QECCs that encode one logical qubit information and correct an arbitrary single-qubit error. Fig.~\ref{fig:bp}(a) plots the partial derivative of the off-diagonal cost
	\begin{equation}
		\begin{aligned}
			\sum_{E_{\mu} \in \mathcal{E}}  \sum_{1 \leq i<j \leq K} \big|\langle\mathbf{i}|\langle\mathbf{0}|U^\dagger(\boldsymbol{\theta}) E_{\mu} U(\boldsymbol{\theta})|\mathbf{j}\rangle|\mathbf{0}\rangle \big|^2
		\end{aligned}
	\end{equation}
	and the diagonal cost
	\begin{equation}
		\begin{aligned}
			\sum_{E_{\mu} \in \mathcal{E}}\sum_{j=1}^K \big|\langle\mathbf{j}|\langle\mathbf{0}|U^\dagger(\boldsymbol{\theta}) E_{\mu} U(\boldsymbol{\theta})|\mathbf{j}\rangle|\mathbf{0}\rangle-\overline{\langle E_{\mu} \rangle}\big|^2/4. 
		\end{aligned}
	\end{equation}
	with respect to a randomly selected circuit parameter $\theta_j$. When the number of VQC layers is $L = 3$ or $L = \lceil \log(n) \rceil$, the circuit is trainable. However, when $L=n$, both off-diagonal and diagonal gradients decay exponentially with the increasing number of qubits. 
	
	\begin{figure}[tb]
		\centering
		\includegraphics[width=7.8cm]{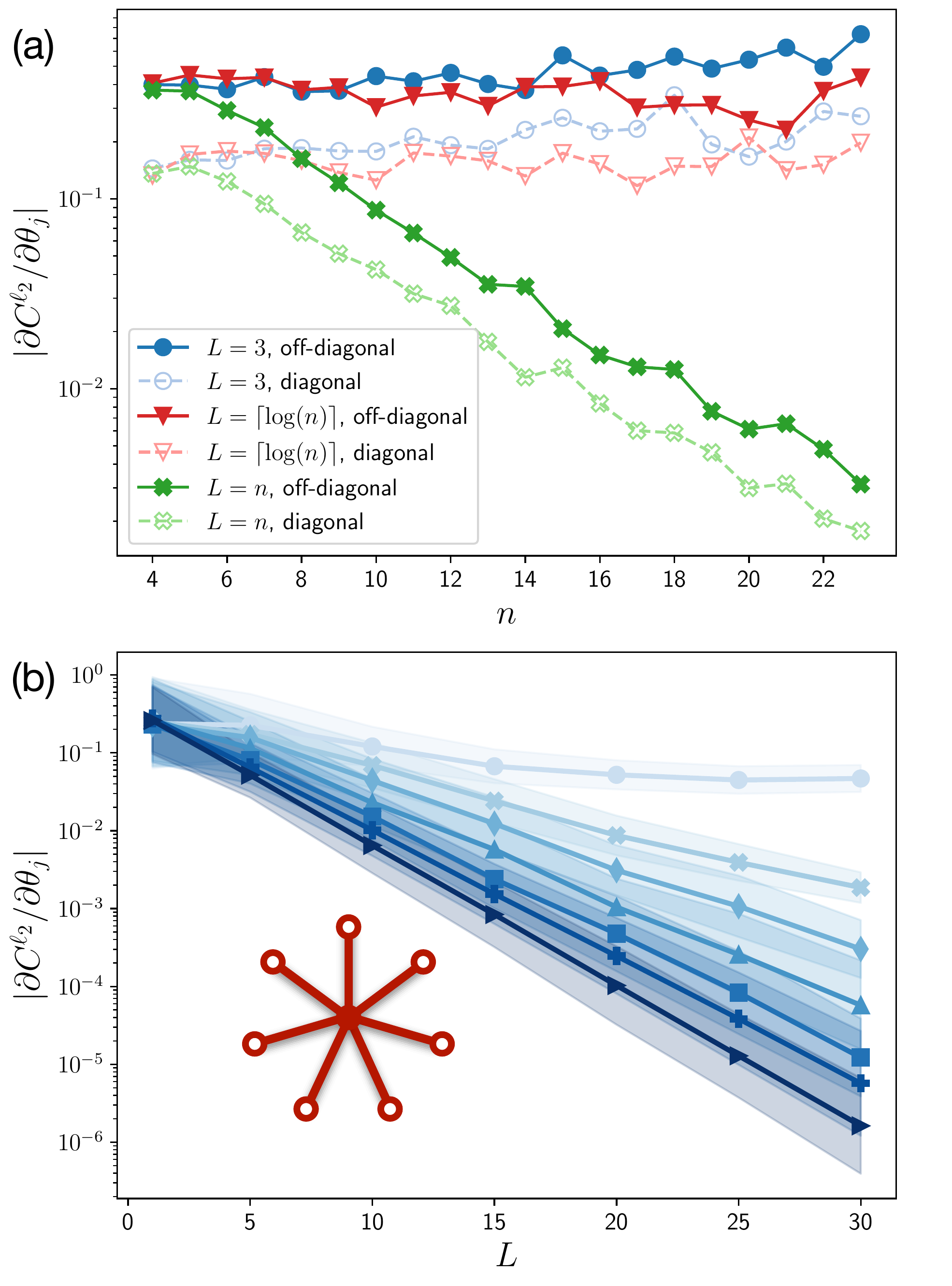}
		\caption{Barren plateaus and noise-induced barren plateaus in VarQEC. (a) Partial derivatives of the off-diagonal and diagonal parts of $C^{\ell_{2}}_{n, K,\mathcal{E}}$ with respect to a random circuit parameter for different system sizes and circuit depths. Each point is averaged over 1000 samples. (b) Partial derivatives of $C^{\ell_{2}}_{n, K,\mathcal{E}}$ with respect to a random parameter for gate noise rates (from top to bottom) $p = 0, 5 \times 10^{-3}, 0.01, 0.02, 0.03, 0.04, 0.05$. Each point is averaged over 1000 samples, and the shaded areas represent the standard deviations. Inset: connectivity graph of $S_{7}$.}
		\label{fig:bp}
	\end{figure}

	The noise-induced barren plateau refers to a conceptually different phenomenon where cost gradients vanish exponentially with $L$ due to hardware noise accumulation~\cite{wang2021noise}. Consequently, the gradients vanish exponentially with $n$ if $L$ grows linearly with $n$. Unlike the noise-free BP, NIPB only occurs when the VQC is noisy, regardless of whether the circuits form a unitary 2-design. Still, we consider the local noise model illustrated in Fig.~\ref{fig:noise_gate}, system size $n=8$, number of layers $L=1,5,10,15,20,25,30$, noise rate $p = 0, 5 \times 10^{-3}, 0.01, 0.02, 0.03, 0.04, 0.05$. The numerical results for the gradients are shown in Fig.~\ref{fig:bp}(b). With the increase of $L$, the partial derivatives of $C^{\ell_{2}}_{n, K,\mathcal{E}}$ with respect to a random parameter decay exponentially, and the decay factor is determined by the noise rate. This illustrates that although VarQEC can find a roughly correct encoder after adequate training with a noisy VQC (noise resilience), the required training time grows exponentially with the number of circuit layers.
	
	BPs and NIBPs manifest themselves in VarQEC when the circuit depth gets large. Nevertheless, we do not need to worry too much about them. From a practical standpoint, we are more interested in QECCs with a shallow (even constant depth) encoding circuit. The gradients of cost functions tend to be large when searching for these codes. In addition, there are more and more effective strategies to mitigate BPs, e.g., cost function partitioning and meta-learning~\cite{PhysRevResearch.3.033090} as well as optimization guided by classical shadows~\cite{PRXQuantum.3.020365}. These protocols can be applied to VarQEC reasonably.

	\section{Experiment on an IBM machine}\label{Sec:Experiment}
	Now we experimentally demonstrate VarQEC with a real superconducting quantum machine, $ibm\_quito$~\cite{ibmq_quito}. The connectivity graph of $ibm\_quito$ is shown in Fig.~\ref{fig:ibm}(a). Our goal is to find a 4-qubit approximate QECC to correct one amplitude damping error~\cite{leung1997approximate} using physical qubits $Q_0, Q_1, Q_2, Q_3$. 
	
	The Kraus operators of the amplitude damping channel are
	\begin{equation}
		\begin{aligned}
			A_{0}&=\left(\begin{array}{cc}
				1 & 0 \\
				0 & \sqrt{1-\gamma}
			\end{array}\right)\\&= I-\frac{\gamma}{4}(I-Z)+\mathcal{O}\left(\gamma^{2}\right),\\
			A_{1}&=\left(\begin{array}{cc}
				0 & \sqrt{\gamma} \\
				0 & 0
			\end{array}\right)=\frac{\sqrt{\gamma}}{2}(X+i Y).\\
		\end{aligned}
	\end{equation}
	Each $(I-Z)$ term contributes a factor of $\gamma$ and each $(X+iY)$ term contributes a factor of $\sqrt{\gamma}$. To correct a single amplitude damping error, we only need to consider error products with total order less than $\gamma^{3/2}$:
	\begin{equation}
		\begin{aligned}
			\mathcal{E} = \{&I, X_j+iY_j, X_j-iY_j, (X_i-iY_i)(X_j+iY_j), \\ &I_j-Z_j\}.
		\end{aligned}
	\end{equation}
	The variational quantum circuit we use is illustrated in Fig.~\ref{fig:ibm}(b). When the rotation angle $\theta = \pm \pi/2$, the VQC serves as an exact encoder. Since $Q_2$ and $Q_3$ are not directly connected, the IBM compiler adds 2 additional SWAP gates (each realized by 3 CNOT gates) to implement CNOT between $Q_2$ and $Q_3$). The hardware-efficient VQC after compiling is shown in Fig.~\ref{fig:ibm}(c).

	Due to hardware constraints, we slightly modify the VarQEC algorithm and enhance it with quantum error mitigation (EM) as follows. Suppose the initial parameter $\theta=0.1$, we iteratively apply the VQC to input states $|0000\rangle$ and $|0010\rangle$, do quantum state tomography on the output mixed states and record their density matrices $\rho_1$ and $\rho_2$. Then we classically extract their dominating eigenstates 
	\begin{equation}
		|\tilde{\psi}_1\rangle = \lim_{M \rightarrow \infty}\frac{\rho_1^M}{\operatorname{Tr}(\rho_1^M)}, |\tilde{\psi}_2\rangle = \lim_{M \rightarrow \infty}\frac{\rho_2^M}{\operatorname{Tr}(\rho_2^M)},
	\end{equation}
	estimate the cost functions $C^{\ell_{2}}_{n, K,\mathcal{E}}$ and $C^{\ell_{1}}_{n, K,\mathcal{E}}$ of logical basis states $\{|\tilde{\psi}_1\rangle, (|\tilde{\psi}_2\rangle - \langle\tilde{\psi}_1|\tilde{\psi}_2\rangle|\tilde{\psi}_1\rangle)/c\}$, where $c$ is a normalization factor.

	The cost gradients are estimated by finite differencing:
	\begin{equation}
		\frac{\partial C_{4, 2,\mathcal{E}}(\theta)}{\partial \theta}  \approx  \frac{C_{4, 2,\mathcal{E}}(\theta + \delta\theta) - C_{4, 2,\mathcal{E}}(\theta - \delta\theta)}{2\delta\theta}
	\end{equation}
	with $\delta\theta = 0.05$. In the first stage (first 15 iterations), we minimize $C^{\ell_{2}}_{n, K,\mathcal{E}}$ with learning rate $\eta = 1$ until $C^{\ell_{2}}_{n, K,\mathcal{E}} < 0.01$. Then we switch to $C^{\ell_{1}}_{n, K,\mathcal{E}}$ and minimize it with a smaller learning rate $\eta = 0.05$. The training curves of the estimated $C^{\ell_{2}}_{n, K,\mathcal{E}}$ with/without error mitigation and its real value are shown in Fig.~\ref{fig:ibm}(d). After adequate training (25 iterations), the parameter $\theta$ converges to about $1.63$, slightly greater than the ideal angle $\pi/2$ (indicated by the dashed line in the inset). Nevertheless, this difference is acceptable, the VQC still encodes an approximate amplitude damping code. 
	
	We implement a total of $152$ quantum circuits for this experiment: $100$ for estimating the gradients and $52$ for estimating the cost functions.

	\begin{figure}[bt]
		\centering
		\includegraphics[width=8.2cm]{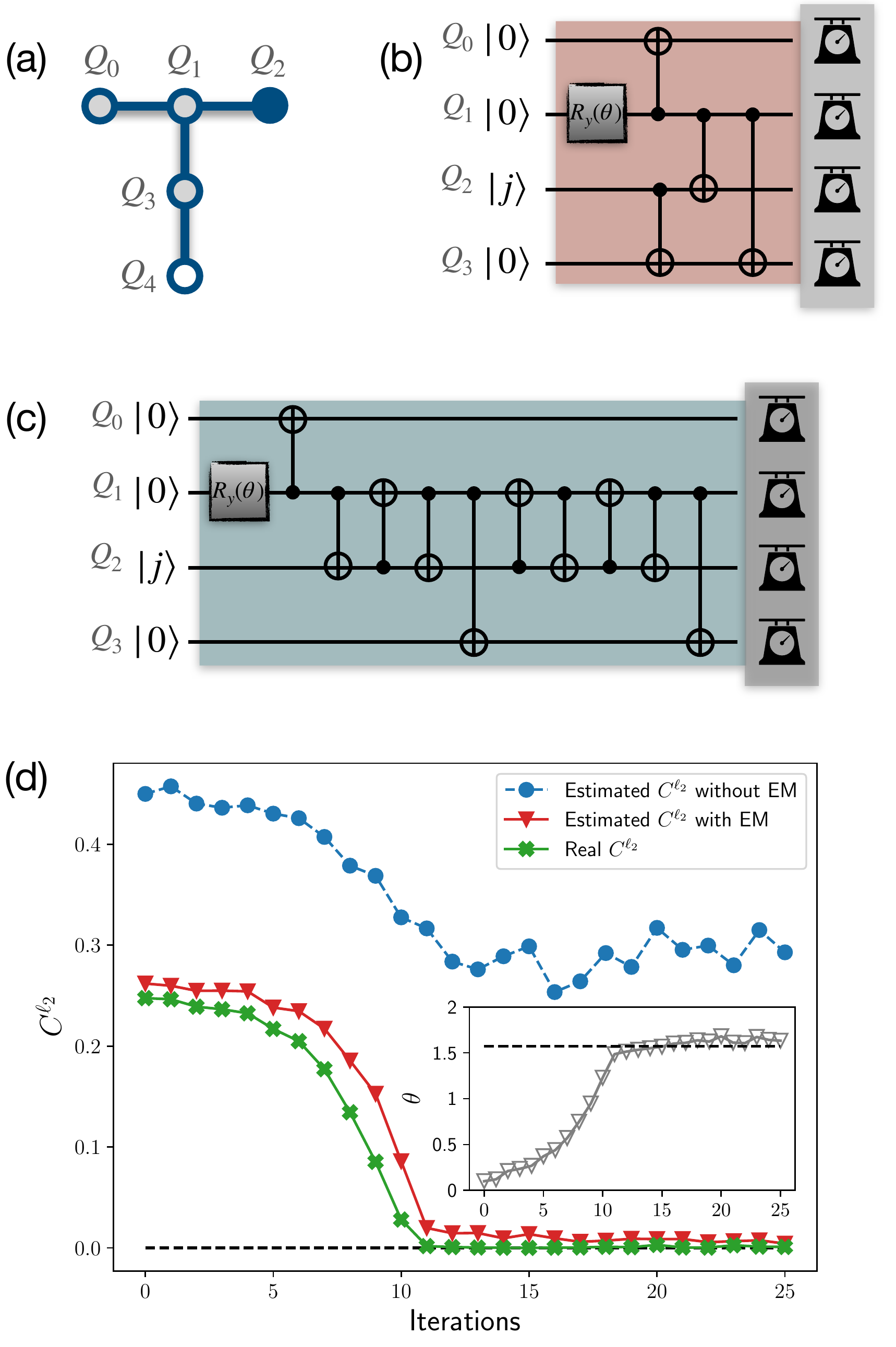}
		\caption{Experimental results of VarQEC for finding an approximate amplitude damping code. (a) The hardware connectivity of $ibm\_quito$. $Q_0, Q_1, Q_2, Q_3$ are used in the experiment. $Q_2$ encodes the initial logical information. (b-c) The original and the compiled variational quantum circuits. (d) The estimated $\ell_2$-norm cost functions with/without error mitigation (EM) and its real value during variational learning. Inset: the rotation angle $\theta$ during learning.}
		\label{fig:ibm}
	\end{figure}

	\section{Conclusions and Outlooks}\label{Conclusion}
	In this work, we proposed VarQEC, an effective variational quantum algorithm for finding various quantum error-correcting codes. VarQEC is capable of finding arbitrary quantum codes since the cost functions therein are based on the most general requirement of a QECC, the Knill-Laflamme conditions. We demonstrated its efficacy by discovering/rediscovering some symmetric, asymmetric, and channel-adaptive codes, e.g., $((5,2,3))_2$, $((5,6,2))_2$, $((6,2,3))_2$, $((7,2,3))_2$, $((12,2^6,3))_2$, $((14,2^8,3))_2$, $((10,4,4))_2$, $((6,2,d_e(2)=4))_2$, $((8,3,d_e(2)=4))_2$. Some discovered codes are equivalent to stabilizer ones and some are not. We investigated them in detail. In particular, VarQEC provided numerical evidence that a quantum code with parameters $((7,3,3))_2$ does not exist. It is worth mentioning that the channel-adaptive codes with optimized encoding circuits found by our method can then be used as inner codes on the physical level in a concatenation scheme. Stabilizer QECCs for qudits can be used as outer codes.

	VarQEC is robust to hardware noise; therefore, it is particularly promising in the NISQ era. A problem worth studying further is how to choose the most resource-efficient variational quantum circuit in VarQEC. There is reason to believe that the optimal VQC ansatz is code-dependent. For example, when the target quantum code is translational-invariant, one may use a VQC with a certain amount of symmetry, where different gates can share the same parameter. If we slightly modify the cost functions, VarQEC can be used for finding some QECC variants like the hybrid quantum-classical codes~\cite{8006823}, estimating the zero-error capacity of noisy quantum channels~\cite{duan2009super}, and solving quantum marginal problems~\cite{yu2021complete}.
	
	VarQEC can also be directly revised to a classical algorithm. When a NISQ processor is not accessible, one can replace the VQCs with classical variational ansatzes like tensor networks~\cite{orus2019tensor, cirac2021matrix, cheng2021simulating} or neural network quantum states~\cite{carleo2017solving}, and then similarly implement optimization and search for eligible quantum codes merely with a classical computer. However, the encoding circuits can not be naturally obtained.

	\acknowledgments
	We thank Song Cheng and Sirui Lu for helpful discussions. MG acknowledges support by the Foundation for Polish Science (IRAP project, ICTQT, contract no. MAB/2018/5, co-financed by EU within Smart Growth Operational Programme). BZ is supported by General Research Fund (no. GRF/16305121).

	\bibliography{ref}

\begin{thebibliography}{101}%
\makeatletter
\providecommand \@ifxundefined [1]{%
 \@ifx{#1\undefined}
}%
\providecommand \@ifnum [1]{%
 \ifnum #1\expandafter \@firstoftwo
 \else \expandafter \@secondoftwo
 \fi
}%
\providecommand \@ifx [1]{%
 \ifx #1\expandafter \@firstoftwo
 \else \expandafter \@secondoftwo
 \fi
}%
\providecommand \natexlab [1]{#1}%
\providecommand \enquote  [1]{``#1''}%
\providecommand \bibnamefont  [1]{#1}%
\providecommand \bibfnamefont [1]{#1}%
\providecommand \citenamefont [1]{#1}%
\providecommand \href@noop [0]{\@secondoftwo}%
\providecommand \href [0]{\begingroup \@sanitize@url \@href}%
\providecommand \@href[1]{\@@startlink{#1}\@@href}%
\providecommand \@@href[1]{\endgroup#1\@@endlink}%
\providecommand \@sanitize@url [0]{\catcode `\\12\catcode `\$12\catcode
  `\&12\catcode `\#12\catcode `\^12\catcode `\_12\catcode `\%12\relax}%
\providecommand \@@startlink[1]{}%
\providecommand \@@endlink[0]{}%
\providecommand \url  [0]{\begingroup\@sanitize@url \@url }%
\providecommand \@url [1]{\endgroup\@href {#1}{\urlprefix }}%
\providecommand \urlprefix  [0]{URL }%
\providecommand \Eprint [0]{\href }%
\providecommand \doibase [0]{https://doi.org/}%
\providecommand \selectlanguage [0]{\@gobble}%
\providecommand \bibinfo  [0]{\@secondoftwo}%
\providecommand \bibfield  [0]{\@secondoftwo}%
\providecommand \translation [1]{[#1]}%
\providecommand \BibitemOpen [0]{}%
\providecommand \bibitemStop [0]{}%
\providecommand \bibitemNoStop [0]{.\EOS\space}%
\providecommand \EOS [0]{\spacefactor3000\relax}%
\providecommand \BibitemShut  [1]{\csname bibitem#1\endcsname}%
\let\auto@bib@innerbib\@empty
\bibitem [{\citenamefont {Jones}\ \emph {et~al.}(2012)\citenamefont {Jones},
  \citenamefont {Whitfield}, \citenamefont {McMahon}, \citenamefont {Yung},
  \citenamefont {Meter}, \citenamefont {Aspuru-Guzik},\ and\ \citenamefont
  {Yamamoto}}]{jones2012faster}%
  \BibitemOpen
  \bibfield  {author} {\bibinfo {author} {\bibfnamefont {N.~C.}\ \bibnamefont
  {Jones}}, \bibinfo {author} {\bibfnamefont {J.~D.}\ \bibnamefont
  {Whitfield}}, \bibinfo {author} {\bibfnamefont {P.~L.}\ \bibnamefont
  {McMahon}}, \bibinfo {author} {\bibfnamefont {M.-H.}\ \bibnamefont {Yung}},
  \bibinfo {author} {\bibfnamefont {R.~V.}\ \bibnamefont {Meter}}, \bibinfo
  {author} {\bibfnamefont {A.}~\bibnamefont {Aspuru-Guzik}},\ and\ \bibinfo
  {author} {\bibfnamefont {Y.}~\bibnamefont {Yamamoto}},\ }\bibfield  {title}
  {\bibinfo {title} {Faster quantum chemistry simulation on fault-tolerant
  quantum computers},\ }\href {https://doi.org/10.1088/1367-2630/14/11/115023}
  {\bibfield  {journal} {\bibinfo  {journal} {New Journal of Physics}\ }\textbf
  {\bibinfo {volume} {14}},\ \bibinfo {pages} {115023} (\bibinfo {year}
  {2012})}\BibitemShut {NoStop}%
\bibitem [{\citenamefont {Shor}(1997)}]{shor1999polynomial}%
  \BibitemOpen
  \bibfield  {author} {\bibinfo {author} {\bibfnamefont {P.~W.}\ \bibnamefont
  {Shor}},\ }\bibfield  {title} {\bibinfo {title} {Polynomial-time algorithms
  for prime factorization and discrete logarithms on a quantum computer},\
  }\href {https://doi.org/10.1137/S0097539795293172} {\bibfield  {journal}
  {\bibinfo  {journal} {SIAM J. Comput.}\ }\textbf {\bibinfo {volume} {26}},\
  \bibinfo {pages} {1484–1509} (\bibinfo {year} {1997})}\BibitemShut
  {NoStop}%
\bibitem [{\citenamefont {Harrow}\ \emph {et~al.}(2009)\citenamefont {Harrow},
  \citenamefont {Hassidim},\ and\ \citenamefont {Lloyd}}]{harrow2009quantum}%
  \BibitemOpen
  \bibfield  {author} {\bibinfo {author} {\bibfnamefont {A.~W.}\ \bibnamefont
  {Harrow}}, \bibinfo {author} {\bibfnamefont {A.}~\bibnamefont {Hassidim}},\
  and\ \bibinfo {author} {\bibfnamefont {S.}~\bibnamefont {Lloyd}},\ }\bibfield
   {title} {\bibinfo {title} {Quantum algorithm for linear systems of
  equations},\ }\href {https://doi.org/10.1103/PhysRevLett.103.150502}
  {\bibfield  {journal} {\bibinfo  {journal} {Phys. Rev. Lett.}\ }\textbf
  {\bibinfo {volume} {103}},\ \bibinfo {pages} {150502} (\bibinfo {year}
  {2009})}\BibitemShut {NoStop}%
\bibitem [{\citenamefont {Shor}(1995)}]{shor1995scheme}%
  \BibitemOpen
  \bibfield  {author} {\bibinfo {author} {\bibfnamefont {P.~W.}\ \bibnamefont
  {Shor}},\ }\bibfield  {title} {\bibinfo {title} {Scheme for reducing
  decoherence in quantum computer memory},\ }\href
  {https://doi.org/10.1103/PhysRevA.52.R2493} {\bibfield  {journal} {\bibinfo
  {journal} {Phys. Rev. A}\ }\textbf {\bibinfo {volume} {52}},\ \bibinfo
  {pages} {R2493} (\bibinfo {year} {1995})}\BibitemShut {NoStop}%
\bibitem [{\citenamefont {Gottesman}(1997)}]{gottesman1997stabilizer}%
  \BibitemOpen
  \bibfield  {author} {\bibinfo {author} {\bibfnamefont {D.}~\bibnamefont
  {Gottesman}},\ }\href@noop {} {\emph {\bibinfo {title} {Stabilizer codes and
  quantum error correction}}}\ (\bibinfo  {publisher} {California Institute of
  Technology},\ \bibinfo {year} {1997})\BibitemShut {NoStop}%
\bibitem [{\citenamefont {Lidar}\ and\ \citenamefont
  {Brun}(2013)}]{lidar2013quantum}%
  \BibitemOpen
  \bibfield  {author} {\bibinfo {author} {\bibfnamefont {D.~A.}\ \bibnamefont
  {Lidar}}\ and\ \bibinfo {author} {\bibfnamefont {T.~A.}\ \bibnamefont
  {Brun}},\ }\href@noop {} {\emph {\bibinfo {title} {Quantum error
  correction}}}\ (\bibinfo  {publisher} {Cambridge University Press},\ \bibinfo
  {year} {2013})\BibitemShut {NoStop}%
\bibitem [{\citenamefont {Zeng}\ \emph {et~al.}(2019)\citenamefont {Zeng},
  \citenamefont {Chen}, \citenamefont {Zhou},\ and\ \citenamefont
  {Wen}}]{zeng2019quantum}%
  \BibitemOpen
  \bibfield  {author} {\bibinfo {author} {\bibfnamefont {B.}~\bibnamefont
  {Zeng}}, \bibinfo {author} {\bibfnamefont {X.}~\bibnamefont {Chen}}, \bibinfo
  {author} {\bibfnamefont {D.-L.}\ \bibnamefont {Zhou}},\ and\ \bibinfo
  {author} {\bibfnamefont {X.-G.}\ \bibnamefont {Wen}},\ }\href@noop {} {\emph
  {\bibinfo {title} {Quantum information meets quantum matter: From quantum
  entanglement to topological phases of many-body systems}}}\ (\bibinfo
  {publisher} {Springer},\ \bibinfo {year} {2019})\BibitemShut {NoStop}%
\bibitem [{\citenamefont {Girvin}(2021)}]{girvin2021introduction}%
  \BibitemOpen
  \bibfield  {author} {\bibinfo {author} {\bibfnamefont {S.~M.}\ \bibnamefont
  {Girvin}},\ }\href@noop {} {\bibinfo {title} {Introduction to quantum error
  correction and fault tolerance}} (\bibinfo {year} {2021}),\ \Eprint
  {https://arxiv.org/abs/2111.08894} {arXiv:2111.08894} \BibitemShut {NoStop}%
\bibitem [{\citenamefont {Pastawski}\ \emph {et~al.}(2015)\citenamefont
  {Pastawski}, \citenamefont {Yoshida}, \citenamefont {Harlow},\ and\
  \citenamefont {Preskill}}]{pastawski2015holographic}%
  \BibitemOpen
  \bibfield  {author} {\bibinfo {author} {\bibfnamefont {F.}~\bibnamefont
  {Pastawski}}, \bibinfo {author} {\bibfnamefont {B.}~\bibnamefont {Yoshida}},
  \bibinfo {author} {\bibfnamefont {D.}~\bibnamefont {Harlow}},\ and\ \bibinfo
  {author} {\bibfnamefont {J.}~\bibnamefont {Preskill}},\ }\bibfield  {title}
  {\bibinfo {title} {Holographic quantum error-correcting codes: toy models for
  the bulk/boundary correspondence},\ }\href
  {https://doi.org/10.1007/JHEP06(2015)149} {\bibfield  {journal} {\bibinfo
  {journal} {Journal of High Energy Physics}\ }\textbf {\bibinfo {volume}
  {2015}},\ \bibinfo {pages} {149} (\bibinfo {year} {2015})}\BibitemShut
  {NoStop}%
\bibitem [{\citenamefont {Knill}\ and\ \citenamefont
  {Laflamme}(1997)}]{knill1997theory}%
  \BibitemOpen
  \bibfield  {author} {\bibinfo {author} {\bibfnamefont {E.}~\bibnamefont
  {Knill}}\ and\ \bibinfo {author} {\bibfnamefont {R.}~\bibnamefont
  {Laflamme}},\ }\bibfield  {title} {\bibinfo {title} {Theory of quantum
  error-correcting codes},\ }\href {https://doi.org/10.1103/PhysRevA.55.900}
  {\bibfield  {journal} {\bibinfo  {journal} {Phys. Rev. A}\ }\textbf {\bibinfo
  {volume} {55}},\ \bibinfo {pages} {900} (\bibinfo {year} {1997})}\BibitemShut
  {NoStop}%
\bibitem [{\citenamefont {Kitaev}(1997)}]{kitaev1997quantum}%
  \BibitemOpen
  \bibfield  {author} {\bibinfo {author} {\bibfnamefont {A.~Y.}\ \bibnamefont
  {Kitaev}},\ }\bibfield  {title} {\bibinfo {title} {Quantum computations:
  algorithms and error correction},\ }\href@noop {} {\bibfield  {journal}
  {\bibinfo  {journal} {Uspekhi Matematicheskikh Nauk}\ }\textbf {\bibinfo
  {volume} {52}},\ \bibinfo {pages} {53} (\bibinfo {year} {1997})}\BibitemShut
  {NoStop}%
\bibitem [{\citenamefont {Fowler}\ \emph {et~al.}(2012)\citenamefont {Fowler},
  \citenamefont {Mariantoni}, \citenamefont {Martinis},\ and\ \citenamefont
  {Cleland}}]{fowler2012surface}%
  \BibitemOpen
  \bibfield  {author} {\bibinfo {author} {\bibfnamefont {A.~G.}\ \bibnamefont
  {Fowler}}, \bibinfo {author} {\bibfnamefont {M.}~\bibnamefont {Mariantoni}},
  \bibinfo {author} {\bibfnamefont {J.~M.}\ \bibnamefont {Martinis}},\ and\
  \bibinfo {author} {\bibfnamefont {A.~N.}\ \bibnamefont {Cleland}},\
  }\bibfield  {title} {\bibinfo {title} {Surface codes: Towards practical
  large-scale quantum computation},\ }\href
  {https://doi.org/10.1103/PhysRevA.86.032324} {\bibfield  {journal} {\bibinfo
  {journal} {Phys. Rev. A}\ }\textbf {\bibinfo {volume} {86}},\ \bibinfo
  {pages} {032324} (\bibinfo {year} {2012})}\BibitemShut {NoStop}%
\bibitem [{\citenamefont {Calderbank}\ and\ \citenamefont
  {Shor}(1996)}]{calderbank1996good}%
  \BibitemOpen
  \bibfield  {author} {\bibinfo {author} {\bibfnamefont {A.~R.}\ \bibnamefont
  {Calderbank}}\ and\ \bibinfo {author} {\bibfnamefont {P.~W.}\ \bibnamefont
  {Shor}},\ }\bibfield  {title} {\bibinfo {title} {Good quantum
  error-correcting codes exist},\ }\href
  {https://doi.org/10.1103/PhysRevA.54.1098} {\bibfield  {journal} {\bibinfo
  {journal} {Phys. Rev. A}\ }\textbf {\bibinfo {volume} {54}},\ \bibinfo
  {pages} {1098} (\bibinfo {year} {1996})}\BibitemShut {NoStop}%
\bibitem [{\citenamefont {Steane}(1996{\natexlab{a}})}]{steane1996multiple}%
  \BibitemOpen
  \bibfield  {author} {\bibinfo {author} {\bibfnamefont {A.}~\bibnamefont
  {Steane}},\ }\bibfield  {title} {\bibinfo {title} {Multiple-particle
  interference and quantum error correction},\ }\href
  {https://doi.org/http://doi.org/10.1098/rspa.1996.0136} {\bibfield  {journal}
  {\bibinfo  {journal} {Proceedings of the Royal Society of London. Series A:
  Mathematical, Physical and Engineering Sciences}\ }\textbf {\bibinfo {volume}
  {452}},\ \bibinfo {pages} {2551} (\bibinfo {year}
  {1996}{\natexlab{a}})}\BibitemShut {NoStop}%
\bibitem [{\citenamefont {Cross}\ \emph {et~al.}(2008)\citenamefont {Cross},
  \citenamefont {Smith}, \citenamefont {Smolin},\ and\ \citenamefont
  {Zeng}}]{cross2008codeword}%
  \BibitemOpen
  \bibfield  {author} {\bibinfo {author} {\bibfnamefont {A.}~\bibnamefont
  {Cross}}, \bibinfo {author} {\bibfnamefont {G.}~\bibnamefont {Smith}},
  \bibinfo {author} {\bibfnamefont {J.~A.}\ \bibnamefont {Smolin}},\ and\
  \bibinfo {author} {\bibfnamefont {B.}~\bibnamefont {Zeng}},\ }\bibfield
  {title} {\bibinfo {title} {Codeword stabilized quantum codes},\ }in\ \href
  {https://doi.org/10.1109/ISIT.2008.4595009} {\emph {\bibinfo {booktitle}
  {2008 IEEE International Symposium on Information Theory}}}\ (\bibinfo {year}
  {2008})\ pp.\ \bibinfo {pages} {364--368}\BibitemShut {NoStop}%
\bibitem [{\citenamefont {Chuang}\ \emph {et~al.}(2009)\citenamefont {Chuang},
  \citenamefont {Cross}, \citenamefont {Smith}, \citenamefont {Smolin},\ and\
  \citenamefont {Zeng}}]{chuang2009codeword}%
  \BibitemOpen
  \bibfield  {author} {\bibinfo {author} {\bibfnamefont {I.}~\bibnamefont
  {Chuang}}, \bibinfo {author} {\bibfnamefont {A.}~\bibnamefont {Cross}},
  \bibinfo {author} {\bibfnamefont {G.}~\bibnamefont {Smith}}, \bibinfo
  {author} {\bibfnamefont {J.}~\bibnamefont {Smolin}},\ and\ \bibinfo {author}
  {\bibfnamefont {B.}~\bibnamefont {Zeng}},\ }\bibfield  {title} {\bibinfo
  {title} {Codeword stabilized quantum codes: Algorithm and structure},\ }\href
  {https://doi.org/10.1063/1.3086833} {\bibfield  {journal} {\bibinfo
  {journal} {Journal of Mathematical Physics}\ }\textbf {\bibinfo {volume}
  {50}},\ \bibinfo {pages} {042109} (\bibinfo {year} {2009})}\BibitemShut
  {NoStop}%
\bibitem [{\citenamefont {Breuckmann}\ and\ \citenamefont
  {Eberhardt}(2021)}]{breuckmann2021quantum}%
  \BibitemOpen
  \bibfield  {author} {\bibinfo {author} {\bibfnamefont {N.~P.}\ \bibnamefont
  {Breuckmann}}\ and\ \bibinfo {author} {\bibfnamefont {J.~N.}\ \bibnamefont
  {Eberhardt}},\ }\bibfield  {title} {\bibinfo {title} {Quantum low-density
  parity-check codes},\ }\href {https://doi.org/10.1103/PRXQuantum.2.040101}
  {\bibfield  {journal} {\bibinfo  {journal} {PRX Quantum}\ }\textbf {\bibinfo
  {volume} {2}},\ \bibinfo {pages} {040101} (\bibinfo {year}
  {2021})}\BibitemShut {NoStop}%
\bibitem [{\citenamefont {Panteleev}\ and\ \citenamefont
  {Kalachev}(2021)}]{panteleev2021asymptotically}%
  \BibitemOpen
  \bibfield  {author} {\bibinfo {author} {\bibfnamefont {P.}~\bibnamefont
  {Panteleev}}\ and\ \bibinfo {author} {\bibfnamefont {G.}~\bibnamefont
  {Kalachev}},\ }\href@noop {} {\bibinfo {title} {Asymptotically good quantum
  and locally testable classical {LDPC} codes}} (\bibinfo {year} {2021}),\
  \Eprint {https://arxiv.org/abs/2111.03654} {arXiv:2111.03654} \BibitemShut
  {NoStop}%
\bibitem [{\citenamefont {Egan}\ \emph {et~al.}(2021)\citenamefont {Egan},
  \citenamefont {Debroy}, \citenamefont {Noel}, \citenamefont {Risinger},
  \citenamefont {Zhu}, \citenamefont {Biswas}, \citenamefont {Newman},
  \citenamefont {Li}, \citenamefont {Brown}, \citenamefont {Cetina},\ and\
  \citenamefont {Monroe}}]{egan2021fault}%
  \BibitemOpen
  \bibfield  {author} {\bibinfo {author} {\bibfnamefont {L.}~\bibnamefont
  {Egan}}, \bibinfo {author} {\bibfnamefont {D.~M.}\ \bibnamefont {Debroy}},
  \bibinfo {author} {\bibfnamefont {C.}~\bibnamefont {Noel}}, \bibinfo {author}
  {\bibfnamefont {A.}~\bibnamefont {Risinger}}, \bibinfo {author}
  {\bibfnamefont {D.}~\bibnamefont {Zhu}}, \bibinfo {author} {\bibfnamefont
  {D.}~\bibnamefont {Biswas}}, \bibinfo {author} {\bibfnamefont
  {M.}~\bibnamefont {Newman}}, \bibinfo {author} {\bibfnamefont
  {M.}~\bibnamefont {Li}}, \bibinfo {author} {\bibfnamefont {K.~R.}\
  \bibnamefont {Brown}}, \bibinfo {author} {\bibfnamefont {M.}~\bibnamefont
  {Cetina}},\ and\ \bibinfo {author} {\bibfnamefont {C.}~\bibnamefont
  {Monroe}},\ }\bibfield  {title} {\bibinfo {title} {Fault-tolerant control of
  an error-corrected qubit},\ }\href
  {https://doi.org/10.1038/s41586-021-03928-y} {\bibfield  {journal} {\bibinfo
  {journal} {Nature}\ }\textbf {\bibinfo {volume} {598}},\ \bibinfo {pages}
  {281} (\bibinfo {year} {2021})}\BibitemShut {NoStop}%
\bibitem [{\citenamefont {Postler}\ \emph {et~al.}(2021)\citenamefont
  {Postler}, \citenamefont {Heußen}, \citenamefont {Pogorelov}, \citenamefont
  {Rispler}, \citenamefont {Feldker}, \citenamefont {Meth}, \citenamefont
  {Marciniak}, \citenamefont {Stricker}, \citenamefont {Ringbauer},
  \citenamefont {Blatt}, \citenamefont {Schindler}, \citenamefont {Müller},\
  and\ \citenamefont {Monz}}]{postler2021demonstration}%
  \BibitemOpen
  \bibfield  {author} {\bibinfo {author} {\bibfnamefont {L.}~\bibnamefont
  {Postler}}, \bibinfo {author} {\bibfnamefont {S.}~\bibnamefont {Heußen}},
  \bibinfo {author} {\bibfnamefont {I.}~\bibnamefont {Pogorelov}}, \bibinfo
  {author} {\bibfnamefont {M.}~\bibnamefont {Rispler}}, \bibinfo {author}
  {\bibfnamefont {T.}~\bibnamefont {Feldker}}, \bibinfo {author} {\bibfnamefont
  {M.}~\bibnamefont {Meth}}, \bibinfo {author} {\bibfnamefont {C.~D.}\
  \bibnamefont {Marciniak}}, \bibinfo {author} {\bibfnamefont {R.}~\bibnamefont
  {Stricker}}, \bibinfo {author} {\bibfnamefont {M.}~\bibnamefont {Ringbauer}},
  \bibinfo {author} {\bibfnamefont {R.}~\bibnamefont {Blatt}}, \bibinfo
  {author} {\bibfnamefont {P.}~\bibnamefont {Schindler}}, \bibinfo {author}
  {\bibfnamefont {M.}~\bibnamefont {Müller}},\ and\ \bibinfo {author}
  {\bibfnamefont {T.}~\bibnamefont {Monz}},\ }\href@noop {} {\bibinfo {title}
  {Demonstration of fault-tolerant universal quantum gate operations}}
  (\bibinfo {year} {2021}),\ \Eprint {https://arxiv.org/abs/2111.12654}
  {arXiv:2111.12654} \BibitemShut {NoStop}%
\bibitem [{\citenamefont {Dawson}\ \emph {et~al.}(2006)\citenamefont {Dawson},
  \citenamefont {Haselgrove},\ and\ \citenamefont {Nielsen}}]{dawson2006noise}%
  \BibitemOpen
  \bibfield  {author} {\bibinfo {author} {\bibfnamefont {C.~M.}\ \bibnamefont
  {Dawson}}, \bibinfo {author} {\bibfnamefont {H.~L.}\ \bibnamefont
  {Haselgrove}},\ and\ \bibinfo {author} {\bibfnamefont {M.~A.}\ \bibnamefont
  {Nielsen}},\ }\bibfield  {title} {\bibinfo {title} {Noise thresholds for
  optical quantum computers},\ }\href
  {https://doi.org/10.1103/PhysRevLett.96.020501} {\bibfield  {journal}
  {\bibinfo  {journal} {Phys. Rev. Lett.}\ }\textbf {\bibinfo {volume} {96}},\
  \bibinfo {pages} {020501} (\bibinfo {year} {2006})}\BibitemShut {NoStop}%
\bibitem [{\citenamefont {Wilen}\ \emph {et~al.}(2021)\citenamefont {Wilen},
  \citenamefont {Abdullah}, \citenamefont {Kurinsky}, \citenamefont {Stanford},
  \citenamefont {Cardani}, \citenamefont {D'Imperio}, \citenamefont {Tomei},
  \citenamefont {Faoro}, \citenamefont {Ioffe}, \citenamefont {Liu},
  \citenamefont {Opremcak}, \citenamefont {Christensen}, \citenamefont
  {DuBois},\ and\ \citenamefont {McDermott}}]{wilen2021correlated}%
  \BibitemOpen
  \bibfield  {author} {\bibinfo {author} {\bibfnamefont {C.~D.}\ \bibnamefont
  {Wilen}}, \bibinfo {author} {\bibfnamefont {S.}~\bibnamefont {Abdullah}},
  \bibinfo {author} {\bibfnamefont {N.~A.}\ \bibnamefont {Kurinsky}}, \bibinfo
  {author} {\bibfnamefont {C.}~\bibnamefont {Stanford}}, \bibinfo {author}
  {\bibfnamefont {L.}~\bibnamefont {Cardani}}, \bibinfo {author} {\bibfnamefont
  {G.}~\bibnamefont {D'Imperio}}, \bibinfo {author} {\bibfnamefont
  {C.}~\bibnamefont {Tomei}}, \bibinfo {author} {\bibfnamefont
  {L.}~\bibnamefont {Faoro}}, \bibinfo {author} {\bibfnamefont {L.~B.}\
  \bibnamefont {Ioffe}}, \bibinfo {author} {\bibfnamefont {C.~H.}\ \bibnamefont
  {Liu}}, \bibinfo {author} {\bibfnamefont {A.}~\bibnamefont {Opremcak}},
  \bibinfo {author} {\bibfnamefont {B.~G.}\ \bibnamefont {Christensen}},
  \bibinfo {author} {\bibfnamefont {J.~L.}\ \bibnamefont {DuBois}},\ and\
  \bibinfo {author} {\bibfnamefont {R.}~\bibnamefont {McDermott}},\ }\bibfield
  {title} {\bibinfo {title} {Correlated charge noise and relaxation errors in
  superconducting qubits},\ }\href {https://doi.org/10.1038/s41586-021-03557-5}
  {\bibfield  {journal} {\bibinfo  {journal} {Nature}\ }\textbf {\bibinfo
  {volume} {594}},\ \bibinfo {pages} {369} (\bibinfo {year}
  {2021})}\BibitemShut {NoStop}%
\bibitem [{\citenamefont {Guo}\ \emph {et~al.}(2021)\citenamefont {Guo},
  \citenamefont {Zhao}, \citenamefont {Grassl}, \citenamefont {Nie},
  \citenamefont {Xiang}, \citenamefont {Xin}, \citenamefont {Yin},\ and\
  \citenamefont {Zeng}}]{guo2021testing}%
  \BibitemOpen
  \bibfield  {author} {\bibinfo {author} {\bibfnamefont {Q.}~\bibnamefont
  {Guo}}, \bibinfo {author} {\bibfnamefont {Y.-Y.}\ \bibnamefont {Zhao}},
  \bibinfo {author} {\bibfnamefont {M.}~\bibnamefont {Grassl}}, \bibinfo
  {author} {\bibfnamefont {X.}~\bibnamefont {Nie}}, \bibinfo {author}
  {\bibfnamefont {G.-Y.}\ \bibnamefont {Xiang}}, \bibinfo {author}
  {\bibfnamefont {T.}~\bibnamefont {Xin}}, \bibinfo {author} {\bibfnamefont
  {Z.-Q.}\ \bibnamefont {Yin}},\ and\ \bibinfo {author} {\bibfnamefont
  {B.}~\bibnamefont {Zeng}},\ }\bibfield  {title} {\bibinfo {title} {Testing a
  quantum error-correcting code on various platforms},\ }\href
  {https://doi.org/https://doi.org/10.1016/j.scib.2020.07.033} {\bibfield
  {journal} {\bibinfo  {journal} {Science Bulletin}\ }\textbf {\bibinfo
  {volume} {66}},\ \bibinfo {pages} {29} (\bibinfo {year} {2021})}\BibitemShut
  {NoStop}%
\bibitem [{\citenamefont {Yu}\ \emph {et~al.}(2007)\citenamefont {Yu},
  \citenamefont {Chen},\ and\ \citenamefont {Oh}}]{yu2007graphical}%
  \BibitemOpen
  \bibfield  {author} {\bibinfo {author} {\bibfnamefont {S.}~\bibnamefont
  {Yu}}, \bibinfo {author} {\bibfnamefont {Q.}~\bibnamefont {Chen}},\ and\
  \bibinfo {author} {\bibfnamefont {C.~H.}\ \bibnamefont {Oh}},\ }\href@noop {}
  {\bibinfo {title} {Graphical quantum error-correcting codes}} (\bibinfo
  {year} {2007}),\ \Eprint {https://arxiv.org/abs/0709.1780} {arXiv:0709.1780}
  \BibitemShut {NoStop}%
\bibitem [{\citenamefont {Hu}\ \emph {et~al.}(2008)\citenamefont {Hu},
  \citenamefont {Tang}, \citenamefont {Zhao}, \citenamefont {Chen},
  \citenamefont {Yu},\ and\ \citenamefont {Oh}}]{hu2008graphical}%
  \BibitemOpen
  \bibfield  {author} {\bibinfo {author} {\bibfnamefont {D.}~\bibnamefont
  {Hu}}, \bibinfo {author} {\bibfnamefont {W.}~\bibnamefont {Tang}}, \bibinfo
  {author} {\bibfnamefont {M.}~\bibnamefont {Zhao}}, \bibinfo {author}
  {\bibfnamefont {Q.}~\bibnamefont {Chen}}, \bibinfo {author} {\bibfnamefont
  {S.}~\bibnamefont {Yu}},\ and\ \bibinfo {author} {\bibfnamefont {C.~H.}\
  \bibnamefont {Oh}},\ }\bibfield  {title} {\bibinfo {title} {Graphical
  nonbinary quantum error-correcting codes},\ }\href
  {https://doi.org/10.1103/PhysRevA.78.012306} {\bibfield  {journal} {\bibinfo
  {journal} {Phys. Rev. A}\ }\textbf {\bibinfo {volume} {78}},\ \bibinfo
  {pages} {012306} (\bibinfo {year} {2008})}\BibitemShut {NoStop}%
\bibitem [{\citenamefont {Jayashankar}\ \emph {et~al.}(2020)\citenamefont
  {Jayashankar}, \citenamefont {Babu}, \citenamefont {Ng},\ and\ \citenamefont
  {Mandayam}}]{PhysRevA.101.042307}%
  \BibitemOpen
  \bibfield  {author} {\bibinfo {author} {\bibfnamefont {A.}~\bibnamefont
  {Jayashankar}}, \bibinfo {author} {\bibfnamefont {A.~M.}\ \bibnamefont
  {Babu}}, \bibinfo {author} {\bibfnamefont {H.~K.}\ \bibnamefont {Ng}},\ and\
  \bibinfo {author} {\bibfnamefont {P.}~\bibnamefont {Mandayam}},\ }\bibfield
  {title} {\bibinfo {title} {Finding good quantum codes using the cartan
  form},\ }\href {https://doi.org/10.1103/PhysRevA.101.042307} {\bibfield
  {journal} {\bibinfo  {journal} {Phys. Rev. A}\ }\textbf {\bibinfo {volume}
  {101}},\ \bibinfo {pages} {042307} (\bibinfo {year} {2020})}\BibitemShut
  {NoStop}%
\bibitem [{\citenamefont {Li}\ \emph {et~al.}(2017)\citenamefont {Li},
  \citenamefont {Guti\'errez}, \citenamefont {David}, \citenamefont
  {Hernandez},\ and\ \citenamefont {Brown}}]{li2017fault}%
  \BibitemOpen
  \bibfield  {author} {\bibinfo {author} {\bibfnamefont {M.}~\bibnamefont
  {Li}}, \bibinfo {author} {\bibfnamefont {M.}~\bibnamefont {Guti\'errez}},
  \bibinfo {author} {\bibfnamefont {S.~E.}\ \bibnamefont {David}}, \bibinfo
  {author} {\bibfnamefont {A.}~\bibnamefont {Hernandez}},\ and\ \bibinfo
  {author} {\bibfnamefont {K.~R.}\ \bibnamefont {Brown}},\ }\bibfield  {title}
  {\bibinfo {title} {Fault tolerance with bare ancillary qubits for a [[7,1,3]]
  code},\ }\href {https://doi.org/10.1103/PhysRevA.96.032341} {\bibfield
  {journal} {\bibinfo  {journal} {Phys. Rev. A}\ }\textbf {\bibinfo {volume}
  {96}},\ \bibinfo {pages} {032341} (\bibinfo {year} {2017})}\BibitemShut
  {NoStop}%
\bibitem [{\citenamefont {F\"osel}\ \emph {et~al.}(2018)\citenamefont
  {F\"osel}, \citenamefont {Tighineanu}, \citenamefont {Weiss},\ and\
  \citenamefont {Marquardt}}]{fosel2018reinforcement}%
  \BibitemOpen
  \bibfield  {author} {\bibinfo {author} {\bibfnamefont {T.}~\bibnamefont
  {F\"osel}}, \bibinfo {author} {\bibfnamefont {P.}~\bibnamefont {Tighineanu}},
  \bibinfo {author} {\bibfnamefont {T.}~\bibnamefont {Weiss}},\ and\ \bibinfo
  {author} {\bibfnamefont {F.}~\bibnamefont {Marquardt}},\ }\bibfield  {title}
  {\bibinfo {title} {Reinforcement learning with neural networks for quantum
  feedback},\ }\href {https://doi.org/10.1103/PhysRevX.8.031084} {\bibfield
  {journal} {\bibinfo  {journal} {Phys. Rev. X}\ }\textbf {\bibinfo {volume}
  {8}},\ \bibinfo {pages} {031084} (\bibinfo {year} {2018})}\BibitemShut
  {NoStop}%
\bibitem [{\citenamefont {Baireuther}\ \emph {et~al.}(2018)\citenamefont
  {Baireuther}, \citenamefont {O'Brien}, \citenamefont {Tarasinski},\ and\
  \citenamefont {Beenakker}}]{baireuther2018machine}%
  \BibitemOpen
  \bibfield  {author} {\bibinfo {author} {\bibfnamefont {P.}~\bibnamefont
  {Baireuther}}, \bibinfo {author} {\bibfnamefont {T.~E.}\ \bibnamefont
  {O'Brien}}, \bibinfo {author} {\bibfnamefont {B.}~\bibnamefont
  {Tarasinski}},\ and\ \bibinfo {author} {\bibfnamefont {C.~W.~J.}\
  \bibnamefont {Beenakker}},\ }\bibfield  {title} {\bibinfo {title}
  {Machine-learning-assisted correction of correlated qubit errors in a
  topological code},\ }\href {https://doi.org/10.22331/q-2018-01-29-48}
  {\bibfield  {journal} {\bibinfo  {journal} {{Quantum}}\ }\textbf {\bibinfo
  {volume} {2}},\ \bibinfo {pages} {48} (\bibinfo {year} {2018})}\BibitemShut
  {NoStop}%
\bibitem [{\citenamefont {Andreasson}\ \emph {et~al.}(2019)\citenamefont
  {Andreasson}, \citenamefont {Johansson}, \citenamefont {Liljestrand},\ and\
  \citenamefont {Granath}}]{andreasson2019quantum}%
  \BibitemOpen
  \bibfield  {author} {\bibinfo {author} {\bibfnamefont {P.}~\bibnamefont
  {Andreasson}}, \bibinfo {author} {\bibfnamefont {J.}~\bibnamefont
  {Johansson}}, \bibinfo {author} {\bibfnamefont {S.}~\bibnamefont
  {Liljestrand}},\ and\ \bibinfo {author} {\bibfnamefont {M.}~\bibnamefont
  {Granath}},\ }\bibfield  {title} {\bibinfo {title} {Quantum error correction
  for the toric code using deep reinforcement learning},\ }\href
  {https://doi.org/10.22331/q-2019-09-02-183} {\bibfield  {journal} {\bibinfo
  {journal} {{Quantum}}\ }\textbf {\bibinfo {volume} {3}},\ \bibinfo {pages}
  {183} (\bibinfo {year} {2019})}\BibitemShut {NoStop}%
\bibitem [{\citenamefont {Nautrup}\ \emph {et~al.}(2019)\citenamefont
  {Nautrup}, \citenamefont {Delfosse}, \citenamefont {Dunjko}, \citenamefont
  {Briegel},\ and\ \citenamefont {Friis}}]{nautrup2019optimizing}%
  \BibitemOpen
  \bibfield  {author} {\bibinfo {author} {\bibfnamefont {H.~P.}\ \bibnamefont
  {Nautrup}}, \bibinfo {author} {\bibfnamefont {N.}~\bibnamefont {Delfosse}},
  \bibinfo {author} {\bibfnamefont {V.}~\bibnamefont {Dunjko}}, \bibinfo
  {author} {\bibfnamefont {H.~J.}\ \bibnamefont {Briegel}},\ and\ \bibinfo
  {author} {\bibfnamefont {N.}~\bibnamefont {Friis}},\ }\bibfield  {title}
  {\bibinfo {title} {Optimizing quantum error correction codes with
  reinforcement learning},\ }\href {https://doi.org/10.22331/q-2019-12-16-215}
  {\bibfield  {journal} {\bibinfo  {journal} {{Quantum}}\ }\textbf {\bibinfo
  {volume} {3}},\ \bibinfo {pages} {215} (\bibinfo {year} {2019})}\BibitemShut
  {NoStop}%
\bibitem [{\citenamefont {Reimpell}\ and\ \citenamefont
  {Werner}(2005)}]{reimpell2005iterative}%
  \BibitemOpen
  \bibfield  {author} {\bibinfo {author} {\bibfnamefont {M.}~\bibnamefont
  {Reimpell}}\ and\ \bibinfo {author} {\bibfnamefont {R.~F.}\ \bibnamefont
  {Werner}},\ }\bibfield  {title} {\bibinfo {title} {Iterative optimization of
  quantum error correcting codes},\ }\href
  {https://doi.org/10.1103/PhysRevLett.94.080501} {\bibfield  {journal}
  {\bibinfo  {journal} {Phys. Rev. Lett.}\ }\textbf {\bibinfo {volume} {94}},\
  \bibinfo {pages} {080501} (\bibinfo {year} {2005})}\BibitemShut {NoStop}%
\bibitem [{\citenamefont {Fletcher}\ \emph {et~al.}(2007)\citenamefont
  {Fletcher}, \citenamefont {Shor},\ and\ \citenamefont
  {Win}}]{fletcher2007optimum}%
  \BibitemOpen
  \bibfield  {author} {\bibinfo {author} {\bibfnamefont {A.~S.}\ \bibnamefont
  {Fletcher}}, \bibinfo {author} {\bibfnamefont {P.~W.}\ \bibnamefont {Shor}},\
  and\ \bibinfo {author} {\bibfnamefont {M.~Z.}\ \bibnamefont {Win}},\
  }\bibfield  {title} {\bibinfo {title} {Optimum quantum error recovery using
  semidefinite programming},\ }\href
  {https://doi.org/10.1103/PhysRevA.75.012338} {\bibfield  {journal} {\bibinfo
  {journal} {Phys. Rev. A}\ }\textbf {\bibinfo {volume} {75}},\ \bibinfo
  {pages} {012338} (\bibinfo {year} {2007})}\BibitemShut {NoStop}%
\bibitem [{\citenamefont {Fletcher}(2007)}]{fletcher2007channel}%
  \BibitemOpen
  \bibfield  {author} {\bibinfo {author} {\bibfnamefont {A.~S.}\ \bibnamefont
  {Fletcher}},\ }\href@noop {} {\bibinfo {title} {Channel-adapted quantum error
  correction}} (\bibinfo {year} {2007}),\ \Eprint
  {https://arxiv.org/abs/0706.3400} {arXiv:0706.3400} \BibitemShut {NoStop}%
\bibitem [{\citenamefont {Sweke}\ \emph {et~al.}(2020)\citenamefont {Sweke},
  \citenamefont {Kesselring}, \citenamefont {van Nieuwenburg},\ and\
  \citenamefont {Eisert}}]{sweke2020reinforcement}%
  \BibitemOpen
  \bibfield  {author} {\bibinfo {author} {\bibfnamefont {R.}~\bibnamefont
  {Sweke}}, \bibinfo {author} {\bibfnamefont {M.~S.}\ \bibnamefont
  {Kesselring}}, \bibinfo {author} {\bibfnamefont {E.~P.~L.}\ \bibnamefont {van
  Nieuwenburg}},\ and\ \bibinfo {author} {\bibfnamefont {J.}~\bibnamefont
  {Eisert}},\ }\bibfield  {title} {\bibinfo {title} {Reinforcement learning
  decoders for fault-tolerant quantum computation},\ }\href
  {https://doi.org/10.1088/2632-2153/abc609} {\bibfield  {journal} {\bibinfo
  {journal} {Machine Learning: Science and Technology}\ }\textbf {\bibinfo
  {volume} {2}},\ \bibinfo {pages} {025005} (\bibinfo {year}
  {2020})}\BibitemShut {NoStop}%
\bibitem [{\citenamefont {Liu}\ and\ \citenamefont
  {Poulin}(2019)}]{liu2019neural}%
  \BibitemOpen
  \bibfield  {author} {\bibinfo {author} {\bibfnamefont {Y.-H.}\ \bibnamefont
  {Liu}}\ and\ \bibinfo {author} {\bibfnamefont {D.}~\bibnamefont {Poulin}},\
  }\bibfield  {title} {\bibinfo {title} {Neural belief-propagation decoders for
  quantum error-correcting codes},\ }\href
  {https://doi.org/10.1103/PhysRevLett.122.200501} {\bibfield  {journal}
  {\bibinfo  {journal} {Phys. Rev. Lett.}\ }\textbf {\bibinfo {volume} {122}},\
  \bibinfo {pages} {200501} (\bibinfo {year} {2019})}\BibitemShut {NoStop}%
\bibitem [{\citenamefont {Locher}\ \emph {et~al.}(2022)\citenamefont {Locher},
  \citenamefont {Cardarelli},\ and\ \citenamefont
  {Müller}}]{locher2022quantum}%
  \BibitemOpen
  \bibfield  {author} {\bibinfo {author} {\bibfnamefont {D.~F.}\ \bibnamefont
  {Locher}}, \bibinfo {author} {\bibfnamefont {L.}~\bibnamefont {Cardarelli}},\
  and\ \bibinfo {author} {\bibfnamefont {M.}~\bibnamefont {Müller}},\
  }\href@noop {} {\bibinfo {title} {Quantum error correction with quantum
  autoencoders}} (\bibinfo {year} {2022}),\ \Eprint
  {https://arxiv.org/abs/2202.00555} {arXiv:2202.00555} \BibitemShut {NoStop}%
\bibitem [{\citenamefont {Knill}\ and\ \citenamefont
  {Laflamme}(1996)}]{knill1996concatenated}%
  \BibitemOpen
  \bibfield  {author} {\bibinfo {author} {\bibfnamefont {E.}~\bibnamefont
  {Knill}}\ and\ \bibinfo {author} {\bibfnamefont {R.}~\bibnamefont
  {Laflamme}},\ }\href@noop {} {\bibinfo {title} {Concatenated quantum codes}}
  (\bibinfo {year} {1996}),\ \Eprint {https://arxiv.org/abs/quant-ph/9608012}
  {arXiv:quant-ph/9608012} \BibitemShut {NoStop}%
\bibitem [{\citenamefont {Grassl}\ \emph {et~al.}(2009)\citenamefont {Grassl},
  \citenamefont {Shor}, \citenamefont {Smith}, \citenamefont {Smolin},\ and\
  \citenamefont {Zeng}}]{grassl2009generalized}%
  \BibitemOpen
  \bibfield  {author} {\bibinfo {author} {\bibfnamefont {M.}~\bibnamefont
  {Grassl}}, \bibinfo {author} {\bibfnamefont {P.}~\bibnamefont {Shor}},
  \bibinfo {author} {\bibfnamefont {G.}~\bibnamefont {Smith}}, \bibinfo
  {author} {\bibfnamefont {J.}~\bibnamefont {Smolin}},\ and\ \bibinfo {author}
  {\bibfnamefont {B.}~\bibnamefont {Zeng}},\ }\bibfield  {title} {\bibinfo
  {title} {Generalized concatenated quantum codes},\ }\href
  {https://doi.org/10.1103/PhysRevA.79.050306} {\bibfield  {journal} {\bibinfo
  {journal} {Phys. Rev. A}\ }\textbf {\bibinfo {volume} {79}},\ \bibinfo
  {pages} {050306} (\bibinfo {year} {2009})}\BibitemShut {NoStop}%
\bibitem [{\citenamefont {Gottesman}(2002)}]{gottesman2002introduction}%
  \BibitemOpen
  \bibfield  {author} {\bibinfo {author} {\bibfnamefont {D.}~\bibnamefont
  {Gottesman}},\ }\bibfield  {title} {\bibinfo {title} {An introduction to
  quantum error correction},\ }in\ \href@noop {} {\emph {\bibinfo {booktitle}
  {Proceedings of Symposia in Applied Mathematics}}},\ Vol.~\bibinfo {volume}
  {58}\ (\bibinfo {year} {2002})\ pp.\ \bibinfo {pages} {221--236}\BibitemShut
  {NoStop}%
\bibitem [{\citenamefont {Aliferis}\ \emph {et~al.}(2009)\citenamefont
  {Aliferis}, \citenamefont {Brito}, \citenamefont {DiVincenzo}, \citenamefont
  {Preskill}, \citenamefont {Steffen},\ and\ \citenamefont
  {Terhal}}]{aliferis2009fault}%
  \BibitemOpen
  \bibfield  {author} {\bibinfo {author} {\bibfnamefont {P.}~\bibnamefont
  {Aliferis}}, \bibinfo {author} {\bibfnamefont {F.}~\bibnamefont {Brito}},
  \bibinfo {author} {\bibfnamefont {D.~P.}\ \bibnamefont {DiVincenzo}},
  \bibinfo {author} {\bibfnamefont {J.}~\bibnamefont {Preskill}}, \bibinfo
  {author} {\bibfnamefont {M.}~\bibnamefont {Steffen}},\ and\ \bibinfo {author}
  {\bibfnamefont {B.~M.}\ \bibnamefont {Terhal}},\ }\bibfield  {title}
  {\bibinfo {title} {Fault-tolerant computing with biased-noise superconducting
  qubits: a case study},\ }\href
  {https://doi.org/10.1088/1367-2630/11/1/013061} {\bibfield  {journal}
  {\bibinfo  {journal} {New Journal of Physics}\ }\textbf {\bibinfo {volume}
  {11}},\ \bibinfo {pages} {013061} (\bibinfo {year} {2009})}\BibitemShut
  {NoStop}%
\bibitem [{\citenamefont {Jackson}\ \emph
  {et~al.}(2016{\natexlab{a}})\citenamefont {Jackson}, \citenamefont {Grassl},\
  and\ \citenamefont {Zeng}}]{jackson2016concatenated}%
  \BibitemOpen
  \bibfield  {author} {\bibinfo {author} {\bibfnamefont {T.}~\bibnamefont
  {Jackson}}, \bibinfo {author} {\bibfnamefont {M.}~\bibnamefont {Grassl}},\
  and\ \bibinfo {author} {\bibfnamefont {B.}~\bibnamefont {Zeng}},\ }\bibfield
  {title} {\bibinfo {title} {Concatenated codes for amplitude damping},\ }in\
  \href {https://doi.org/10.1109/ISIT.2016.7541703} {\emph {\bibinfo
  {booktitle} {2016 IEEE International Symposium on Information Theory
  (ISIT)}}}\ (\bibinfo {year} {2016})\ pp.\ \bibinfo {pages}
  {2269--2273}\BibitemShut {NoStop}%
\bibitem [{\citenamefont {Leung}\ \emph {et~al.}(1997)\citenamefont {Leung},
  \citenamefont {Nielsen}, \citenamefont {Chuang},\ and\ \citenamefont
  {Yamamoto}}]{leung1997approximate}%
  \BibitemOpen
  \bibfield  {author} {\bibinfo {author} {\bibfnamefont {D.~W.}\ \bibnamefont
  {Leung}}, \bibinfo {author} {\bibfnamefont {M.~A.}\ \bibnamefont {Nielsen}},
  \bibinfo {author} {\bibfnamefont {I.~L.}\ \bibnamefont {Chuang}},\ and\
  \bibinfo {author} {\bibfnamefont {Y.}~\bibnamefont {Yamamoto}},\ }\bibfield
  {title} {\bibinfo {title} {Approximate quantum error correction can lead to
  better codes},\ }\href {https://doi.org/10.1103/PhysRevA.56.2567} {\bibfield
  {journal} {\bibinfo  {journal} {Phys. Rev. A}\ }\textbf {\bibinfo {volume}
  {56}},\ \bibinfo {pages} {2567} (\bibinfo {year} {1997})}\BibitemShut
  {NoStop}%
\bibitem [{\citenamefont {Schumacher}\ and\ \citenamefont
  {Westmoreland}(2002)}]{schumacher2002approximate}%
  \BibitemOpen
  \bibfield  {author} {\bibinfo {author} {\bibfnamefont {B.}~\bibnamefont
  {Schumacher}}\ and\ \bibinfo {author} {\bibfnamefont {M.~D.}\ \bibnamefont
  {Westmoreland}},\ }\bibfield  {title} {\bibinfo {title} {Approximate quantum
  error correction},\ }\href {https://doi.org/10.1023/A:1019653202562}
  {\bibfield  {journal} {\bibinfo  {journal} {Quantum Information Processing}\
  }\textbf {\bibinfo {volume} {1}},\ \bibinfo {pages} {5} (\bibinfo {year}
  {2002})}\BibitemShut {NoStop}%
\bibitem [{\citenamefont {Brand\~ao}\ \emph {et~al.}(2019)\citenamefont
  {Brand\~ao}, \citenamefont {Crosson}, \citenamefont {\ifmmode
  \mbox{\c{S}}\else \c{S}\fi{}ahino\ifmmode~\breve{g}\else \u{g}\fi{}lu},\ and\
  \citenamefont {Bowen}}]{brandao2019quantum}%
  \BibitemOpen
  \bibfield  {author} {\bibinfo {author} {\bibfnamefont {F.~G. S.~L.}\
  \bibnamefont {Brand\~ao}}, \bibinfo {author} {\bibfnamefont {E.}~\bibnamefont
  {Crosson}}, \bibinfo {author} {\bibfnamefont {M.~B.}\ \bibnamefont {\ifmmode
  \mbox{\c{S}}\else \c{S}\fi{}ahino\ifmmode~\breve{g}\else \u{g}\fi{}lu}},\
  and\ \bibinfo {author} {\bibfnamefont {J.}~\bibnamefont {Bowen}},\ }\bibfield
   {title} {\bibinfo {title} {Quantum error correcting codes in eigenstates of
  translation-invariant spin chains},\ }\href
  {https://doi.org/10.1103/PhysRevLett.123.110502} {\bibfield  {journal}
  {\bibinfo  {journal} {Phys. Rev. Lett.}\ }\textbf {\bibinfo {volume} {123}},\
  \bibinfo {pages} {110502} (\bibinfo {year} {2019})}\BibitemShut {NoStop}%
\bibitem [{\citenamefont {B\'eny}\ and\ \citenamefont
  {Oreshkov}(2010)}]{beny2010general}%
  \BibitemOpen
  \bibfield  {author} {\bibinfo {author} {\bibfnamefont {C.}~\bibnamefont
  {B\'eny}}\ and\ \bibinfo {author} {\bibfnamefont {O.}~\bibnamefont
  {Oreshkov}},\ }\bibfield  {title} {\bibinfo {title} {General conditions for
  approximate quantum error correction and near-optimal recovery channels},\
  }\href {https://doi.org/10.1103/PhysRevLett.104.120501} {\bibfield  {journal}
  {\bibinfo  {journal} {Phys. Rev. Lett.}\ }\textbf {\bibinfo {volume} {104}},\
  \bibinfo {pages} {120501} (\bibinfo {year} {2010})}\BibitemShut {NoStop}%
\bibitem [{\citenamefont {Bures}(1969)}]{bures1969extension}%
  \BibitemOpen
  \bibfield  {author} {\bibinfo {author} {\bibfnamefont {D.}~\bibnamefont
  {Bures}},\ }\bibfield  {title} {\bibinfo {title} {An extension of
  {K}akutani's theorem on infinite product measures to the tensor product of
  semifinite w*-algebras},\ }\href
  {https://doi.org/https://doi.org/10.2307/1995012} {\bibfield  {journal}
  {\bibinfo  {journal} {Transactions of the American Mathematical Society}\
  }\textbf {\bibinfo {volume} {135}},\ \bibinfo {pages} {199} (\bibinfo {year}
  {1969})}\BibitemShut {NoStop}%
\bibitem [{\citenamefont {Cerezo}\ \emph
  {et~al.}(2021{\natexlab{a}})\citenamefont {Cerezo}, \citenamefont
  {Arrasmith}, \citenamefont {Babbush}, \citenamefont {Benjamin}, \citenamefont
  {Endo}, \citenamefont {Fujii}, \citenamefont {McClean}, \citenamefont
  {Mitarai}, \citenamefont {Yuan}, \citenamefont {Cincio},\ and\ \citenamefont
  {Coles}}]{cerezo2021variational}%
  \BibitemOpen
  \bibfield  {author} {\bibinfo {author} {\bibfnamefont {M.}~\bibnamefont
  {Cerezo}}, \bibinfo {author} {\bibfnamefont {A.}~\bibnamefont {Arrasmith}},
  \bibinfo {author} {\bibfnamefont {R.}~\bibnamefont {Babbush}}, \bibinfo
  {author} {\bibfnamefont {S.~C.}\ \bibnamefont {Benjamin}}, \bibinfo {author}
  {\bibfnamefont {S.}~\bibnamefont {Endo}}, \bibinfo {author} {\bibfnamefont
  {K.}~\bibnamefont {Fujii}}, \bibinfo {author} {\bibfnamefont {J.~R.}\
  \bibnamefont {McClean}}, \bibinfo {author} {\bibfnamefont {K.}~\bibnamefont
  {Mitarai}}, \bibinfo {author} {\bibfnamefont {X.}~\bibnamefont {Yuan}},
  \bibinfo {author} {\bibfnamefont {L.}~\bibnamefont {Cincio}},\ and\ \bibinfo
  {author} {\bibfnamefont {P.~J.}\ \bibnamefont {Coles}},\ }\bibfield  {title}
  {\bibinfo {title} {Variational quantum algorithms},\ }\href
  {https://doi.org/10.1038/s42254-021-00348-9} {\bibfield  {journal} {\bibinfo
  {journal} {Nature Reviews Physics}\ }\textbf {\bibinfo {volume} {3}},\
  \bibinfo {pages} {625} (\bibinfo {year} {2021}{\natexlab{a}})}\BibitemShut
  {NoStop}%
\bibitem [{\citenamefont {Bharti}\ \emph {et~al.}(2022)\citenamefont {Bharti},
  \citenamefont {Cervera-Lierta}, \citenamefont {Kyaw}, \citenamefont {Haug},
  \citenamefont {Alperin-Lea}, \citenamefont {Anand}, \citenamefont {Degroote},
  \citenamefont {Heimonen}, \citenamefont {Kottmann}, \citenamefont {Menke},
  \citenamefont {Mok}, \citenamefont {Sim}, \citenamefont {Kwek},\ and\
  \citenamefont {Aspuru-Guzik}}]{bharti2022noisy}%
  \BibitemOpen
  \bibfield  {author} {\bibinfo {author} {\bibfnamefont {K.}~\bibnamefont
  {Bharti}}, \bibinfo {author} {\bibfnamefont {A.}~\bibnamefont
  {Cervera-Lierta}}, \bibinfo {author} {\bibfnamefont {T.~H.}\ \bibnamefont
  {Kyaw}}, \bibinfo {author} {\bibfnamefont {T.}~\bibnamefont {Haug}}, \bibinfo
  {author} {\bibfnamefont {S.}~\bibnamefont {Alperin-Lea}}, \bibinfo {author}
  {\bibfnamefont {A.}~\bibnamefont {Anand}}, \bibinfo {author} {\bibfnamefont
  {M.}~\bibnamefont {Degroote}}, \bibinfo {author} {\bibfnamefont
  {H.}~\bibnamefont {Heimonen}}, \bibinfo {author} {\bibfnamefont {J.~S.}\
  \bibnamefont {Kottmann}}, \bibinfo {author} {\bibfnamefont {T.}~\bibnamefont
  {Menke}}, \bibinfo {author} {\bibfnamefont {W.-K.}\ \bibnamefont {Mok}},
  \bibinfo {author} {\bibfnamefont {S.}~\bibnamefont {Sim}}, \bibinfo {author}
  {\bibfnamefont {L.-C.}\ \bibnamefont {Kwek}},\ and\ \bibinfo {author}
  {\bibfnamefont {A.}~\bibnamefont {Aspuru-Guzik}},\ }\bibfield  {title}
  {\bibinfo {title} {Noisy intermediate-scale quantum algorithms},\ }\href
  {https://doi.org/10.1103/RevModPhys.94.015004} {\bibfield  {journal}
  {\bibinfo  {journal} {Rev. Mod. Phys.}\ }\textbf {\bibinfo {volume} {94}},\
  \bibinfo {pages} {015004} (\bibinfo {year} {2022})}\BibitemShut {NoStop}%
\bibitem [{\citenamefont {Peruzzo}\ \emph {et~al.}(2014)\citenamefont
  {Peruzzo}, \citenamefont {McClean}, \citenamefont {Shadbolt}, \citenamefont
  {Yung}, \citenamefont {Zhou}, \citenamefont {Love}, \citenamefont
  {Aspuru-Guzik},\ and\ \citenamefont {O'Brien}}]{peruzzo2014variational}%
  \BibitemOpen
  \bibfield  {author} {\bibinfo {author} {\bibfnamefont {A.}~\bibnamefont
  {Peruzzo}}, \bibinfo {author} {\bibfnamefont {J.}~\bibnamefont {McClean}},
  \bibinfo {author} {\bibfnamefont {P.}~\bibnamefont {Shadbolt}}, \bibinfo
  {author} {\bibfnamefont {M.-H.}\ \bibnamefont {Yung}}, \bibinfo {author}
  {\bibfnamefont {X.-Q.}\ \bibnamefont {Zhou}}, \bibinfo {author}
  {\bibfnamefont {P.~J.}\ \bibnamefont {Love}}, \bibinfo {author}
  {\bibfnamefont {A.}~\bibnamefont {Aspuru-Guzik}},\ and\ \bibinfo {author}
  {\bibfnamefont {J.~L.}\ \bibnamefont {O'Brien}},\ }\bibfield  {title}
  {\bibinfo {title} {A variational eigenvalue solver on a photonic quantum
  processor},\ }\href {https://doi.org/10.1038/ncomms5213} {\bibfield
  {journal} {\bibinfo  {journal} {Nature Communications}\ }\textbf {\bibinfo
  {volume} {5}},\ \bibinfo {pages} {4213} (\bibinfo {year} {2014})}\BibitemShut
  {NoStop}%
\bibitem [{\citenamefont {Kandala}\ \emph {et~al.}(2017)\citenamefont
  {Kandala}, \citenamefont {Mezzacapo}, \citenamefont {Temme}, \citenamefont
  {Takita}, \citenamefont {Brink}, \citenamefont {Chow},\ and\ \citenamefont
  {Gambetta}}]{kandala2017hardware}%
  \BibitemOpen
  \bibfield  {author} {\bibinfo {author} {\bibfnamefont {A.}~\bibnamefont
  {Kandala}}, \bibinfo {author} {\bibfnamefont {A.}~\bibnamefont {Mezzacapo}},
  \bibinfo {author} {\bibfnamefont {K.}~\bibnamefont {Temme}}, \bibinfo
  {author} {\bibfnamefont {M.}~\bibnamefont {Takita}}, \bibinfo {author}
  {\bibfnamefont {M.}~\bibnamefont {Brink}}, \bibinfo {author} {\bibfnamefont
  {J.~M.}\ \bibnamefont {Chow}},\ and\ \bibinfo {author} {\bibfnamefont
  {J.~M.}\ \bibnamefont {Gambetta}},\ }\bibfield  {title} {\bibinfo {title}
  {Hardware-efficient variational quantum eigensolver for small molecules and
  quantum magnets},\ }\href {https://doi.org/10.1038/nature23879} {\bibfield
  {journal} {\bibinfo  {journal} {Nature}\ }\textbf {\bibinfo {volume} {549}},\
  \bibinfo {pages} {242} (\bibinfo {year} {2017})}\BibitemShut {NoStop}%
\bibitem [{\citenamefont {Nam}\ \emph {et~al.}(2020)\citenamefont {Nam},
  \citenamefont {Chen}, \citenamefont {Pisenti}, \citenamefont {Wright},
  \citenamefont {Delaney}, \citenamefont {Maslov}, \citenamefont {Brown},
  \citenamefont {Allen}, \citenamefont {Amini}, \citenamefont {Apisdorf},
  \citenamefont {Beck}, \citenamefont {Blinov}, \citenamefont {Chaplin},
  \citenamefont {Chmielewski}, \citenamefont {Collins}, \citenamefont
  {Debnath}, \citenamefont {Hudek}, \citenamefont {Ducore}, \citenamefont
  {Keesan}, \citenamefont {Kreikemeier}, \citenamefont {Mizrahi}, \citenamefont
  {Solomon}, \citenamefont {Williams}, \citenamefont {Wong-Campos},
  \citenamefont {Moehring}, \citenamefont {Monroe},\ and\ \citenamefont
  {Kim}}]{nam2020ground}%
  \BibitemOpen
  \bibfield  {author} {\bibinfo {author} {\bibfnamefont {Y.}~\bibnamefont
  {Nam}}, \bibinfo {author} {\bibfnamefont {J.-S.}\ \bibnamefont {Chen}},
  \bibinfo {author} {\bibfnamefont {N.~C.}\ \bibnamefont {Pisenti}}, \bibinfo
  {author} {\bibfnamefont {K.}~\bibnamefont {Wright}}, \bibinfo {author}
  {\bibfnamefont {C.}~\bibnamefont {Delaney}}, \bibinfo {author} {\bibfnamefont
  {D.}~\bibnamefont {Maslov}}, \bibinfo {author} {\bibfnamefont {K.~R.}\
  \bibnamefont {Brown}}, \bibinfo {author} {\bibfnamefont {S.}~\bibnamefont
  {Allen}}, \bibinfo {author} {\bibfnamefont {J.~M.}\ \bibnamefont {Amini}},
  \bibinfo {author} {\bibfnamefont {J.}~\bibnamefont {Apisdorf}}, \bibinfo
  {author} {\bibfnamefont {K.~M.}\ \bibnamefont {Beck}}, \bibinfo {author}
  {\bibfnamefont {A.}~\bibnamefont {Blinov}}, \bibinfo {author} {\bibfnamefont
  {V.}~\bibnamefont {Chaplin}}, \bibinfo {author} {\bibfnamefont
  {M.}~\bibnamefont {Chmielewski}}, \bibinfo {author} {\bibfnamefont
  {C.}~\bibnamefont {Collins}}, \bibinfo {author} {\bibfnamefont
  {S.}~\bibnamefont {Debnath}}, \bibinfo {author} {\bibfnamefont {K.~M.}\
  \bibnamefont {Hudek}}, \bibinfo {author} {\bibfnamefont {A.~M.}\ \bibnamefont
  {Ducore}}, \bibinfo {author} {\bibfnamefont {M.}~\bibnamefont {Keesan}},
  \bibinfo {author} {\bibfnamefont {S.~M.}\ \bibnamefont {Kreikemeier}},
  \bibinfo {author} {\bibfnamefont {J.}~\bibnamefont {Mizrahi}}, \bibinfo
  {author} {\bibfnamefont {P.}~\bibnamefont {Solomon}}, \bibinfo {author}
  {\bibfnamefont {M.}~\bibnamefont {Williams}}, \bibinfo {author}
  {\bibfnamefont {J.~D.}\ \bibnamefont {Wong-Campos}}, \bibinfo {author}
  {\bibfnamefont {D.}~\bibnamefont {Moehring}}, \bibinfo {author}
  {\bibfnamefont {C.}~\bibnamefont {Monroe}},\ and\ \bibinfo {author}
  {\bibfnamefont {J.}~\bibnamefont {Kim}},\ }\bibfield  {title} {\bibinfo
  {title} {Ground-state energy estimation of the water molecule on a
  trapped-ion quantum computer},\ }\href
  {https://doi.org/10.1038/s41534-020-0259-3} {\bibfield  {journal} {\bibinfo
  {journal} {npj Quantum Information}\ }\textbf {\bibinfo {volume} {6}},\
  \bibinfo {pages} {33} (\bibinfo {year} {2020})}\BibitemShut {NoStop}%
\bibitem [{\citenamefont {Cao}\ \emph {et~al.}(2021)\citenamefont {Cao},
  \citenamefont {Yu}, \citenamefont {Wu}, \citenamefont {Shannon},
  \citenamefont {Zeng},\ and\ \citenamefont {Joynt}}]{cao2021energy}%
  \BibitemOpen
  \bibfield  {author} {\bibinfo {author} {\bibfnamefont {C.}~\bibnamefont
  {Cao}}, \bibinfo {author} {\bibfnamefont {Y.}~\bibnamefont {Yu}}, \bibinfo
  {author} {\bibfnamefont {Z.}~\bibnamefont {Wu}}, \bibinfo {author}
  {\bibfnamefont {N.}~\bibnamefont {Shannon}}, \bibinfo {author} {\bibfnamefont
  {B.}~\bibnamefont {Zeng}},\ and\ \bibinfo {author} {\bibfnamefont
  {R.}~\bibnamefont {Joynt}},\ }\href@noop {} {\bibinfo {title} {Mitigating
  algorithmic errors in quantum optimization through energy extrapolation}}
  (\bibinfo {year} {2021}),\ \Eprint {https://arxiv.org/abs/2109.08132}
  {arXiv:2109.08132} \BibitemShut {NoStop}%
\bibitem [{\citenamefont {Romero}\ \emph {et~al.}(2017)\citenamefont {Romero},
  \citenamefont {Olson},\ and\ \citenamefont
  {Aspuru-Guzik}}]{romero2017quantum}%
  \BibitemOpen
  \bibfield  {author} {\bibinfo {author} {\bibfnamefont {J.}~\bibnamefont
  {Romero}}, \bibinfo {author} {\bibfnamefont {J.~P.}\ \bibnamefont {Olson}},\
  and\ \bibinfo {author} {\bibfnamefont {A.}~\bibnamefont {Aspuru-Guzik}},\
  }\bibfield  {title} {\bibinfo {title} {Quantum autoencoders for efficient
  compression of quantum data},\ }\href
  {https://iopscience.iop.org/article/10.1088/2058-9565/aa8072} {\bibfield
  {journal} {\bibinfo  {journal} {Quantum Science and Technology}\ }\textbf
  {\bibinfo {volume} {2}},\ \bibinfo {pages} {045001} (\bibinfo {year}
  {2017})}\BibitemShut {NoStop}%
\bibitem [{\citenamefont {Cao}\ and\ \citenamefont
  {Wang}(2021)}]{cao2021noise}%
  \BibitemOpen
  \bibfield  {author} {\bibinfo {author} {\bibfnamefont {C.}~\bibnamefont
  {Cao}}\ and\ \bibinfo {author} {\bibfnamefont {X.}~\bibnamefont {Wang}},\
  }\bibfield  {title} {\bibinfo {title} {Noise-assisted quantum autoencoder},\
  }\href {https://doi.org/10.1103/PhysRevApplied.15.054012} {\bibfield
  {journal} {\bibinfo  {journal} {Phys. Rev. Applied}\ }\textbf {\bibinfo
  {volume} {15}},\ \bibinfo {pages} {054012} (\bibinfo {year}
  {2021})}\BibitemShut {NoStop}%
\bibitem [{\citenamefont {Sharma}\ \emph {et~al.}(2020)\citenamefont {Sharma},
  \citenamefont {Khatri}, \citenamefont {Cerezo},\ and\ \citenamefont
  {Coles}}]{sharma2020noise}%
  \BibitemOpen
  \bibfield  {author} {\bibinfo {author} {\bibfnamefont {K.}~\bibnamefont
  {Sharma}}, \bibinfo {author} {\bibfnamefont {S.}~\bibnamefont {Khatri}},
  \bibinfo {author} {\bibfnamefont {M.}~\bibnamefont {Cerezo}},\ and\ \bibinfo
  {author} {\bibfnamefont {P.~J.}\ \bibnamefont {Coles}},\ }\bibfield  {title}
  {\bibinfo {title} {Noise resilience of variational quantum compiling},\
  }\href {https://doi.org/10.1088/1367-2630/ab784c} {\bibfield  {journal}
  {\bibinfo  {journal} {New Journal of Physics}\ }\textbf {\bibinfo {volume}
  {22}},\ \bibinfo {pages} {043006} (\bibinfo {year} {2020})}\BibitemShut
  {NoStop}%
\bibitem [{\citenamefont {Xu}\ \emph {et~al.}(2021)\citenamefont {Xu},
  \citenamefont {Benjamin},\ and\ \citenamefont {Yuan}}]{xu2021variational}%
  \BibitemOpen
  \bibfield  {author} {\bibinfo {author} {\bibfnamefont {X.}~\bibnamefont
  {Xu}}, \bibinfo {author} {\bibfnamefont {S.~C.}\ \bibnamefont {Benjamin}},\
  and\ \bibinfo {author} {\bibfnamefont {X.}~\bibnamefont {Yuan}},\ }\bibfield
  {title} {\bibinfo {title} {Variational circuit compiler for quantum error
  correction},\ }\href {https://doi.org/10.1103/PhysRevApplied.15.034068}
  {\bibfield  {journal} {\bibinfo  {journal} {Phys. Rev. Applied}\ }\textbf
  {\bibinfo {volume} {15}},\ \bibinfo {pages} {034068} (\bibinfo {year}
  {2021})}\BibitemShut {NoStop}%
\bibitem [{\citenamefont {Mitarai}\ \emph {et~al.}(2018)\citenamefont
  {Mitarai}, \citenamefont {Negoro}, \citenamefont {Kitagawa},\ and\
  \citenamefont {Fujii}}]{mitarai2018quantum}%
  \BibitemOpen
  \bibfield  {author} {\bibinfo {author} {\bibfnamefont {K.}~\bibnamefont
  {Mitarai}}, \bibinfo {author} {\bibfnamefont {M.}~\bibnamefont {Negoro}},
  \bibinfo {author} {\bibfnamefont {M.}~\bibnamefont {Kitagawa}},\ and\
  \bibinfo {author} {\bibfnamefont {K.}~\bibnamefont {Fujii}},\ }\bibfield
  {title} {\bibinfo {title} {Quantum circuit learning},\ }\href
  {https://doi.org/10.1103/PhysRevA.98.032309} {\bibfield  {journal} {\bibinfo
  {journal} {Phys. Rev. A}\ }\textbf {\bibinfo {volume} {98}},\ \bibinfo
  {pages} {032309} (\bibinfo {year} {2018})}\BibitemShut {NoStop}%
\bibitem [{\citenamefont {Huang}\ \emph {et~al.}(2020)\citenamefont {Huang},
  \citenamefont {Kueng},\ and\ \citenamefont {Preskill}}]{huang2020predicting}%
  \BibitemOpen
  \bibfield  {author} {\bibinfo {author} {\bibfnamefont {H.-Y.}\ \bibnamefont
  {Huang}}, \bibinfo {author} {\bibfnamefont {R.}~\bibnamefont {Kueng}},\ and\
  \bibinfo {author} {\bibfnamefont {J.}~\bibnamefont {Preskill}},\ }\bibfield
  {title} {\bibinfo {title} {Predicting many properties of a quantum system
  from very few measurements},\ }\href
  {https://doi.org/10.1038/s41567-020-0932-7} {\bibfield  {journal} {\bibinfo
  {journal} {Nature Physics}\ }\textbf {\bibinfo {volume} {16}},\ \bibinfo
  {pages} {1050} (\bibinfo {year} {2020})}\BibitemShut {NoStop}%
\bibitem [{\citenamefont {Powell}(1964)}]{powell1964efficient}%
  \BibitemOpen
  \bibfield  {author} {\bibinfo {author} {\bibfnamefont {M.~J.~D.}\
  \bibnamefont {Powell}},\ }\bibfield  {title} {\bibinfo {title} {{An efficient
  method for finding the minimum of a function of several variables without
  calculating derivatives}},\ }\href {https://doi.org/10.1093/comjnl/7.2.155}
  {\bibfield  {journal} {\bibinfo  {journal} {The Computer Journal}\ }\textbf
  {\bibinfo {volume} {7}},\ \bibinfo {pages} {155} (\bibinfo {year} {1964})},\
  \Eprint
  {https://arxiv.org/abs/https://academic.oup.com/comjnl/article-pdf/7/2/155/959784/070155.pdf}
  {https://academic.oup.com/comjnl/article-pdf/7/2/155/959784/070155.pdf}
  \BibitemShut {NoStop}%
\bibitem [{\citenamefont {Haug}\ \emph {et~al.}(2021)\citenamefont {Haug},
  \citenamefont {Bharti},\ and\ \citenamefont {Kim}}]{QFIM_1}%
  \BibitemOpen
  \bibfield  {author} {\bibinfo {author} {\bibfnamefont {T.}~\bibnamefont
  {Haug}}, \bibinfo {author} {\bibfnamefont {K.}~\bibnamefont {Bharti}},\ and\
  \bibinfo {author} {\bibfnamefont {M.}~\bibnamefont {Kim}},\ }\bibfield
  {title} {\bibinfo {title} {Capacity and quantum geometry of parametrized
  quantum circuits},\ }\href {https://doi.org/10.1103/PRXQuantum.2.040309}
  {\bibfield  {journal} {\bibinfo  {journal} {PRX Quantum}\ }\textbf {\bibinfo
  {volume} {2}},\ \bibinfo {pages} {040309} (\bibinfo {year}
  {2021})}\BibitemShut {NoStop}%
\bibitem [{\citenamefont {Johnson}\ \emph {et~al.}(2017)\citenamefont
  {Johnson}, \citenamefont {Romero}, \citenamefont {Olson}, \citenamefont
  {Cao},\ and\ \citenamefont {Aspuru-Guzik}}]{johnson2017qvector}%
  \BibitemOpen
  \bibfield  {author} {\bibinfo {author} {\bibfnamefont {P.~D.}\ \bibnamefont
  {Johnson}}, \bibinfo {author} {\bibfnamefont {J.}~\bibnamefont {Romero}},
  \bibinfo {author} {\bibfnamefont {J.}~\bibnamefont {Olson}}, \bibinfo
  {author} {\bibfnamefont {Y.}~\bibnamefont {Cao}},\ and\ \bibinfo {author}
  {\bibfnamefont {A.}~\bibnamefont {Aspuru-Guzik}},\ }\href@noop {} {\bibinfo
  {title} {{QVECTOR}: an algorithm for device-tailored quantum error
  correction}} (\bibinfo {year} {2017}),\ \Eprint
  {https://arxiv.org/abs/1711.02249} {arXiv:1711.02249} \BibitemShut {NoStop}%
\bibitem [{\citenamefont {Laflamme}\ \emph {et~al.}(1996)\citenamefont
  {Laflamme}, \citenamefont {Miquel}, \citenamefont {Paz},\ and\ \citenamefont
  {Zurek}}]{laflamme1996perfect}%
  \BibitemOpen
  \bibfield  {author} {\bibinfo {author} {\bibfnamefont {R.}~\bibnamefont
  {Laflamme}}, \bibinfo {author} {\bibfnamefont {C.}~\bibnamefont {Miquel}},
  \bibinfo {author} {\bibfnamefont {J.~P.}\ \bibnamefont {Paz}},\ and\ \bibinfo
  {author} {\bibfnamefont {W.~H.}\ \bibnamefont {Zurek}},\ }\bibfield  {title}
  {\bibinfo {title} {Perfect quantum error correcting code},\ }\href
  {https://doi.org/10.1103/PhysRevLett.77.198} {\bibfield  {journal} {\bibinfo
  {journal} {Phys. Rev. Lett.}\ }\textbf {\bibinfo {volume} {77}},\ \bibinfo
  {pages} {198} (\bibinfo {year} {1996})}\BibitemShut {NoStop}%
\bibitem [{\citenamefont {Rains}\ \emph {et~al.}(1997)\citenamefont {Rains},
  \citenamefont {Hardin}, \citenamefont {Shor},\ and\ \citenamefont
  {Sloane}}]{rains1997nonadditive}%
  \BibitemOpen
  \bibfield  {author} {\bibinfo {author} {\bibfnamefont {E.~M.}\ \bibnamefont
  {Rains}}, \bibinfo {author} {\bibfnamefont {R.~H.}\ \bibnamefont {Hardin}},
  \bibinfo {author} {\bibfnamefont {P.~W.}\ \bibnamefont {Shor}},\ and\
  \bibinfo {author} {\bibfnamefont {N.~J.~A.}\ \bibnamefont {Sloane}},\
  }\bibfield  {title} {\bibinfo {title} {A nonadditive quantum code},\ }\href
  {https://doi.org/10.1103/PhysRevLett.79.953} {\bibfield  {journal} {\bibinfo
  {journal} {Phys. Rev. Lett.}\ }\textbf {\bibinfo {volume} {79}},\ \bibinfo
  {pages} {953} (\bibinfo {year} {1997})}\BibitemShut {NoStop}%
\bibitem [{\citenamefont {Steane}(1996{\natexlab{b}})}]{steane1996simple}%
  \BibitemOpen
  \bibfield  {author} {\bibinfo {author} {\bibfnamefont {A.~M.}\ \bibnamefont
  {Steane}},\ }\bibfield  {title} {\bibinfo {title} {Simple quantum
  error-correcting codes},\ }\href {https://doi.org/10.1103/PhysRevA.54.4741}
  {\bibfield  {journal} {\bibinfo  {journal} {Phys. Rev. A}\ }\textbf {\bibinfo
  {volume} {54}},\ \bibinfo {pages} {4741} (\bibinfo {year}
  {1996}{\natexlab{b}})}\BibitemShut {NoStop}%
\bibitem [{\citenamefont {Ioffe}\ and\ \citenamefont
  {M\'ezard}(2007)}]{ioffe2007asymmetric}%
  \BibitemOpen
  \bibfield  {author} {\bibinfo {author} {\bibfnamefont {L.}~\bibnamefont
  {Ioffe}}\ and\ \bibinfo {author} {\bibfnamefont {M.}~\bibnamefont
  {M\'ezard}},\ }\bibfield  {title} {\bibinfo {title} {Asymmetric quantum
  error-correcting codes},\ }\href {https://doi.org/10.1103/PhysRevA.75.032345}
  {\bibfield  {journal} {\bibinfo  {journal} {Phys. Rev. A}\ }\textbf {\bibinfo
  {volume} {75}},\ \bibinfo {pages} {032345} (\bibinfo {year}
  {2007})}\BibitemShut {NoStop}%
\bibitem [{\citenamefont {Sarvepalli}\ \emph {et~al.}(2008)\citenamefont
  {Sarvepalli}, \citenamefont {Klappenecker},\ and\ \citenamefont
  {Rotteler}}]{sarvepalli2008asymmetric}%
  \BibitemOpen
  \bibfield  {author} {\bibinfo {author} {\bibfnamefont {P.~K.}\ \bibnamefont
  {Sarvepalli}}, \bibinfo {author} {\bibfnamefont {A.}~\bibnamefont
  {Klappenecker}},\ and\ \bibinfo {author} {\bibfnamefont {M.}~\bibnamefont
  {Rotteler}},\ }\bibfield  {title} {\bibinfo {title} {Asymmetric quantum
  {LDPC} codes},\ }in\ \href {https://doi.org/10.1109/ISIT.2008.4594997} {\emph
  {\bibinfo {booktitle} {2008 IEEE International Symposium on Information
  Theory}}}\ (\bibinfo {year} {2008})\ pp.\ \bibinfo {pages}
  {305--309}\BibitemShut {NoStop}%
\bibitem [{\citenamefont {Sarvepalli}\ \emph {et~al.}(2009)\citenamefont
  {Sarvepalli}, \citenamefont {Klappenecker},\ and\ \citenamefont
  {Rötteler}}]{sarvepalli2009asymmetric}%
  \BibitemOpen
  \bibfield  {author} {\bibinfo {author} {\bibfnamefont {P.~K.}\ \bibnamefont
  {Sarvepalli}}, \bibinfo {author} {\bibfnamefont {A.}~\bibnamefont
  {Klappenecker}},\ and\ \bibinfo {author} {\bibfnamefont {M.}~\bibnamefont
  {Rötteler}},\ }\bibfield  {title} {\bibinfo {title} {Asymmetric quantum
  codes: constructions, bounds and performance},\ }\href
  {https://doi.org/10.1098/rspa.2008.0439} {\bibfield  {journal} {\bibinfo
  {journal} {Proceedings of the Royal Society A: Mathematical, Physical and
  Engineering Sciences}\ }\textbf {\bibinfo {volume} {465}},\ \bibinfo {pages}
  {1645} (\bibinfo {year} {2009})}\BibitemShut {NoStop}%
\bibitem [{\citenamefont {Ezerman}\ \emph {et~al.}(2011)\citenamefont
  {Ezerman}, \citenamefont {Ling},\ and\ \citenamefont
  {Sole}}]{ezerman2011additive}%
  \BibitemOpen
  \bibfield  {author} {\bibinfo {author} {\bibfnamefont {M.~F.}\ \bibnamefont
  {Ezerman}}, \bibinfo {author} {\bibfnamefont {S.}~\bibnamefont {Ling}},\ and\
  \bibinfo {author} {\bibfnamefont {P.}~\bibnamefont {Sole}},\ }\bibfield
  {title} {\bibinfo {title} {Additive asymmetric quantum codes},\ }\href
  {https://doi.org/10.1109/TIT.2011.2159040} {\bibfield  {journal} {\bibinfo
  {journal} {IEEE Transactions on Information Theory}\ }\textbf {\bibinfo
  {volume} {57}},\ \bibinfo {pages} {5536} (\bibinfo {year}
  {2011})}\BibitemShut {NoStop}%
\bibitem [{\citenamefont {Ezerman}\ \emph {et~al.}(2013)\citenamefont
  {Ezerman}, \citenamefont {Jitman}, \citenamefont {Ling},\ and\ \citenamefont
  {Pasechnik}}]{ezerman2013css}%
  \BibitemOpen
  \bibfield  {author} {\bibinfo {author} {\bibfnamefont {M.~F.}\ \bibnamefont
  {Ezerman}}, \bibinfo {author} {\bibfnamefont {S.}~\bibnamefont {Jitman}},
  \bibinfo {author} {\bibfnamefont {S.}~\bibnamefont {Ling}},\ and\ \bibinfo
  {author} {\bibfnamefont {D.~V.}\ \bibnamefont {Pasechnik}},\ }\bibfield
  {title} {\bibinfo {title} {{CSS}-like constructions of asymmetric quantum
  codes},\ }\href {https://doi.org/10.1109/TIT.2013.2272575} {\bibfield
  {journal} {\bibinfo  {journal} {IEEE Transactions on Information Theory}\
  }\textbf {\bibinfo {volume} {59}},\ \bibinfo {pages} {6732} (\bibinfo {year}
  {2013})}\BibitemShut {NoStop}%
\bibitem [{\citenamefont {Jackson}\ \emph
  {et~al.}(2016{\natexlab{b}})\citenamefont {Jackson}, \citenamefont {Grassl},\
  and\ \citenamefont {Zeng}}]{jackson2016codeword}%
  \BibitemOpen
  \bibfield  {author} {\bibinfo {author} {\bibfnamefont {T.}~\bibnamefont
  {Jackson}}, \bibinfo {author} {\bibfnamefont {M.}~\bibnamefont {Grassl}},\
  and\ \bibinfo {author} {\bibfnamefont {B.}~\bibnamefont {Zeng}},\ }\bibfield
  {title} {\bibinfo {title} {Codeword stabilized quantum codes for asymmetric
  channels},\ }in\ \href {https://doi.org/10.1109/ISIT.2016.7541702} {\emph
  {\bibinfo {booktitle} {2016 IEEE International Symposium on Information
  Theory (ISIT)}}}\ (\bibinfo {year} {2016})\ pp.\ \bibinfo {pages}
  {2264--2268}\BibitemShut {NoStop}%
\bibitem [{\citenamefont {Bonilla~Ataides}\ \emph {et~al.}(2021)\citenamefont
  {Bonilla~Ataides}, \citenamefont {Tuckett}, \citenamefont {Bartlett},
  \citenamefont {Flammia},\ and\ \citenamefont {Brown}}]{bonilla2021xzzx}%
  \BibitemOpen
  \bibfield  {author} {\bibinfo {author} {\bibfnamefont {J.~P.}\ \bibnamefont
  {Bonilla~Ataides}}, \bibinfo {author} {\bibfnamefont {D.~K.}\ \bibnamefont
  {Tuckett}}, \bibinfo {author} {\bibfnamefont {S.~D.}\ \bibnamefont
  {Bartlett}}, \bibinfo {author} {\bibfnamefont {S.~T.}\ \bibnamefont
  {Flammia}},\ and\ \bibinfo {author} {\bibfnamefont {B.~J.}\ \bibnamefont
  {Brown}},\ }\bibfield  {title} {\bibinfo {title} {The xzzx surface code},\
  }\href {https://doi.org/10.1038/s41467-021-22274-1} {\bibfield  {journal}
  {\bibinfo  {journal} {Nature Communications}\ }\textbf {\bibinfo {volume}
  {12}},\ \bibinfo {pages} {2172} (\bibinfo {year} {2021})}\BibitemShut
  {NoStop}%
\bibitem [{\citenamefont {Prabhu}\ and\ \citenamefont
  {Reichardt}(2021)}]{prabhu2021distancefour}%
  \BibitemOpen
  \bibfield  {author} {\bibinfo {author} {\bibfnamefont {P.}~\bibnamefont
  {Prabhu}}\ and\ \bibinfo {author} {\bibfnamefont {B.~W.}\ \bibnamefont
  {Reichardt}},\ }\href@noop {} {\bibinfo {title} {Distance-four quantum codes
  with combined postselection and error correction}} (\bibinfo {year} {2021}),\
  \Eprint {https://arxiv.org/abs/2112.03785} {arXiv:2112.03785} \BibitemShut
  {NoStop}%
\bibitem [{\citenamefont {Calderbank}\ \emph {et~al.}(1998)\citenamefont
  {Calderbank}, \citenamefont {Rains}, \citenamefont {Shor},\ and\
  \citenamefont {Sloane}}]{calderbank1998quantum}%
  \BibitemOpen
  \bibfield  {author} {\bibinfo {author} {\bibfnamefont {A.}~\bibnamefont
  {Calderbank}}, \bibinfo {author} {\bibfnamefont {E.}~\bibnamefont {Rains}},
  \bibinfo {author} {\bibfnamefont {P.}~\bibnamefont {Shor}},\ and\ \bibinfo
  {author} {\bibfnamefont {N.}~\bibnamefont {Sloane}},\ }\bibfield  {title}
  {\bibinfo {title} {Quantum error correction via codes over {GF}(4)},\ }\href
  {https://doi.org/10.1109/18.681315} {\bibfield  {journal} {\bibinfo
  {journal} {IEEE Transactions on Information Theory}\ }\textbf {\bibinfo
  {volume} {44}},\ \bibinfo {pages} {1369} (\bibinfo {year}
  {1998})}\BibitemShut {NoStop}%
\bibitem [{\citenamefont {Hama}(2020)}]{hama2020quantum}%
  \BibitemOpen
  \bibfield  {author} {\bibinfo {author} {\bibfnamefont {Y.}~\bibnamefont
  {Hama}},\ }\bibfield  {title} {\bibinfo {title} {Quantum circuits for
  collective amplitude damping in two-qubit systems},\ }\href@noop {} {\
  (\bibinfo {year} {2020})},\ \Eprint {https://arxiv.org/abs/2012.02410}
  {arXiv:2012.02410} \BibitemShut {NoStop}%
\bibitem [{\citenamefont {Grassl}\ \emph {et~al.}(2018)\citenamefont {Grassl},
  \citenamefont {Kong}, \citenamefont {Wei}, \citenamefont {Yin},\ and\
  \citenamefont {Zeng}}]{grassl2018quantum}%
  \BibitemOpen
  \bibfield  {author} {\bibinfo {author} {\bibfnamefont {M.}~\bibnamefont
  {Grassl}}, \bibinfo {author} {\bibfnamefont {L.}~\bibnamefont {Kong}},
  \bibinfo {author} {\bibfnamefont {Z.}~\bibnamefont {Wei}}, \bibinfo {author}
  {\bibfnamefont {Z.-Q.}\ \bibnamefont {Yin}},\ and\ \bibinfo {author}
  {\bibfnamefont {B.}~\bibnamefont {Zeng}},\ }\bibfield  {title} {\bibinfo
  {title} {Quantum error-correcting codes for qudit amplitude damping},\
  }\href@noop {} {\bibfield  {journal} {\bibinfo  {journal} {IEEE Transactions
  on Information Theory}\ }\textbf {\bibinfo {volume} {64}},\ \bibinfo {pages}
  {4674} (\bibinfo {year} {2018})}\BibitemShut {NoStop}%
\bibitem [{\citenamefont {Shor}\ and\ \citenamefont
  {Laflamme}(1997)}]{shor1997quantum}%
  \BibitemOpen
  \bibfield  {author} {\bibinfo {author} {\bibfnamefont {P.}~\bibnamefont
  {Shor}}\ and\ \bibinfo {author} {\bibfnamefont {R.}~\bibnamefont
  {Laflamme}},\ }\bibfield  {title} {\bibinfo {title} {Quantum analog of the
  macwilliams identities for classical coding theory},\ }\href
  {https://doi.org/10.1103/PhysRevLett.78.1600} {\bibfield  {journal} {\bibinfo
   {journal} {Phys. Rev. Lett.}\ }\textbf {\bibinfo {volume} {78}},\ \bibinfo
  {pages} {1600} (\bibinfo {year} {1997})}\BibitemShut {NoStop}%
\bibitem [{Var(2022)}]{VarQECcode}%
  \BibitemOpen
  \href@noop {} {\bibfield  {journal} {\bibinfo  {journal} {``VarQEC GitHub
  repository". \url{https://github.com/caochenfeng/VarQEC-public}}\ } (\bibinfo
  {year} {2022})}\BibitemShut {NoStop}%
\bibitem [{\citenamefont {Chen}\ \emph {et~al.}(2021)\citenamefont {Chen},
  \citenamefont {Satzinger}, \citenamefont {Atalaya}, \citenamefont {Korotkov},
  \citenamefont {Dunsworth}, \citenamefont {Sank}, \citenamefont {Quintana},
  \citenamefont {McEwen}, \citenamefont {Barends}, \citenamefont {Klimov},
  \citenamefont {Hong}, \citenamefont {Jones}, \citenamefont {Petukhov},
  \citenamefont {Kafri}, \citenamefont {Demura}, \citenamefont {Burkett},
  \citenamefont {Gidney}, \citenamefont {Fowler}, \citenamefont {Paler},
  \citenamefont {Putterman}, \citenamefont {Aleiner}, \citenamefont {Arute},
  \citenamefont {Arya}, \citenamefont {Babbush}, \citenamefont {Bardin},
  \citenamefont {Bengtsson}, \citenamefont {Bourassa}, \citenamefont
  {Broughton}, \citenamefont {Buckley}, \citenamefont {Buell}, \citenamefont
  {Bushnell}, \citenamefont {Chiaro}, \citenamefont {Collins}, \citenamefont
  {Courtney}, \citenamefont {Derk}, \citenamefont {Eppens}, \citenamefont
  {Erickson}, \citenamefont {Farhi}, \citenamefont {Foxen}, \citenamefont
  {Giustina}, \citenamefont {Greene}, \citenamefont {Gross}, \citenamefont
  {Harrigan}, \citenamefont {Harrington}, \citenamefont {Hilton}, \citenamefont
  {Ho}, \citenamefont {Huang}, \citenamefont {Huggins}, \citenamefont {Ioffe},
  \citenamefont {Isakov}, \citenamefont {Jeffrey}, \citenamefont {Jiang},
  \citenamefont {Kechedzhi}, \citenamefont {Kim}, \citenamefont {Kitaev},
  \citenamefont {Kostritsa}, \citenamefont {Landhuis}, \citenamefont {Laptev},
  \citenamefont {Lucero}, \citenamefont {Martin}, \citenamefont {McClean},
  \citenamefont {McCourt}, \citenamefont {Mi}, \citenamefont {Miao},
  \citenamefont {Mohseni}, \citenamefont {Montazeri}, \citenamefont
  {Mruczkiewicz}, \citenamefont {Mutus}, \citenamefont {Naaman}, \citenamefont
  {Neeley}, \citenamefont {Neill}, \citenamefont {Newman}, \citenamefont {Niu},
  \citenamefont {O'Brien}, \citenamefont {Opremcak}, \citenamefont {Ostby},
  \citenamefont {Pat{\'o}}, \citenamefont {Redd}, \citenamefont {Roushan},
  \citenamefont {Rubin}, \citenamefont {Shvarts}, \citenamefont {Strain},
  \citenamefont {Szalay}, \citenamefont {Trevithick}, \citenamefont
  {Villalonga}, \citenamefont {White}, \citenamefont {Yao}, \citenamefont
  {Yeh}, \citenamefont {Yoo}, \citenamefont {Zalcman}, \citenamefont {Neven},
  \citenamefont {Boixo}, \citenamefont {Smelyanskiy}, \citenamefont {Chen},
  \citenamefont {Megrant}, \citenamefont {Kelly},\ and\ \citenamefont {{Google
  Quantum AI}}}]{ai2021exponential}%
  \BibitemOpen
  \bibfield  {author} {\bibinfo {author} {\bibfnamefont {Z.}~\bibnamefont
  {Chen}}, \bibinfo {author} {\bibfnamefont {K.~J.}\ \bibnamefont {Satzinger}},
  \bibinfo {author} {\bibfnamefont {J.}~\bibnamefont {Atalaya}}, \bibinfo
  {author} {\bibfnamefont {A.~N.}\ \bibnamefont {Korotkov}}, \bibinfo {author}
  {\bibfnamefont {A.}~\bibnamefont {Dunsworth}}, \bibinfo {author}
  {\bibfnamefont {D.}~\bibnamefont {Sank}}, \bibinfo {author} {\bibfnamefont
  {C.}~\bibnamefont {Quintana}}, \bibinfo {author} {\bibfnamefont
  {M.}~\bibnamefont {McEwen}}, \bibinfo {author} {\bibfnamefont
  {R.}~\bibnamefont {Barends}}, \bibinfo {author} {\bibfnamefont {P.~V.}\
  \bibnamefont {Klimov}}, \bibinfo {author} {\bibfnamefont {S.}~\bibnamefont
  {Hong}}, \bibinfo {author} {\bibfnamefont {C.}~\bibnamefont {Jones}},
  \bibinfo {author} {\bibfnamefont {A.}~\bibnamefont {Petukhov}}, \bibinfo
  {author} {\bibfnamefont {D.}~\bibnamefont {Kafri}}, \bibinfo {author}
  {\bibfnamefont {S.}~\bibnamefont {Demura}}, \bibinfo {author} {\bibfnamefont
  {B.}~\bibnamefont {Burkett}}, \bibinfo {author} {\bibfnamefont
  {C.}~\bibnamefont {Gidney}}, \bibinfo {author} {\bibfnamefont {A.~G.}\
  \bibnamefont {Fowler}}, \bibinfo {author} {\bibfnamefont {A.}~\bibnamefont
  {Paler}}, \bibinfo {author} {\bibfnamefont {H.}~\bibnamefont {Putterman}},
  \bibinfo {author} {\bibfnamefont {I.}~\bibnamefont {Aleiner}}, \bibinfo
  {author} {\bibfnamefont {F.}~\bibnamefont {Arute}}, \bibinfo {author}
  {\bibfnamefont {K.}~\bibnamefont {Arya}}, \bibinfo {author} {\bibfnamefont
  {R.}~\bibnamefont {Babbush}}, \bibinfo {author} {\bibfnamefont {J.~C.}\
  \bibnamefont {Bardin}}, \bibinfo {author} {\bibfnamefont {A.}~\bibnamefont
  {Bengtsson}}, \bibinfo {author} {\bibfnamefont {A.}~\bibnamefont {Bourassa}},
  \bibinfo {author} {\bibfnamefont {M.}~\bibnamefont {Broughton}}, \bibinfo
  {author} {\bibfnamefont {B.~B.}\ \bibnamefont {Buckley}}, \bibinfo {author}
  {\bibfnamefont {D.~A.}\ \bibnamefont {Buell}}, \bibinfo {author}
  {\bibfnamefont {N.}~\bibnamefont {Bushnell}}, \bibinfo {author}
  {\bibfnamefont {B.}~\bibnamefont {Chiaro}}, \bibinfo {author} {\bibfnamefont
  {R.}~\bibnamefont {Collins}}, \bibinfo {author} {\bibfnamefont
  {W.}~\bibnamefont {Courtney}}, \bibinfo {author} {\bibfnamefont {A.~R.}\
  \bibnamefont {Derk}}, \bibinfo {author} {\bibfnamefont {D.}~\bibnamefont
  {Eppens}}, \bibinfo {author} {\bibfnamefont {C.}~\bibnamefont {Erickson}},
  \bibinfo {author} {\bibfnamefont {E.}~\bibnamefont {Farhi}}, \bibinfo
  {author} {\bibfnamefont {B.}~\bibnamefont {Foxen}}, \bibinfo {author}
  {\bibfnamefont {M.}~\bibnamefont {Giustina}}, \bibinfo {author}
  {\bibfnamefont {A.}~\bibnamefont {Greene}}, \bibinfo {author} {\bibfnamefont
  {J.~A.}\ \bibnamefont {Gross}}, \bibinfo {author} {\bibfnamefont {M.~P.}\
  \bibnamefont {Harrigan}}, \bibinfo {author} {\bibfnamefont {S.~D.}\
  \bibnamefont {Harrington}}, \bibinfo {author} {\bibfnamefont
  {J.}~\bibnamefont {Hilton}}, \bibinfo {author} {\bibfnamefont
  {A.}~\bibnamefont {Ho}}, \bibinfo {author} {\bibfnamefont {T.}~\bibnamefont
  {Huang}}, \bibinfo {author} {\bibfnamefont {W.~J.}\ \bibnamefont {Huggins}},
  \bibinfo {author} {\bibfnamefont {L.~B.}\ \bibnamefont {Ioffe}}, \bibinfo
  {author} {\bibfnamefont {S.~V.}\ \bibnamefont {Isakov}}, \bibinfo {author}
  {\bibfnamefont {E.}~\bibnamefont {Jeffrey}}, \bibinfo {author} {\bibfnamefont
  {Z.}~\bibnamefont {Jiang}}, \bibinfo {author} {\bibfnamefont
  {K.}~\bibnamefont {Kechedzhi}}, \bibinfo {author} {\bibfnamefont
  {S.}~\bibnamefont {Kim}}, \bibinfo {author} {\bibfnamefont {A.}~\bibnamefont
  {Kitaev}}, \bibinfo {author} {\bibfnamefont {F.}~\bibnamefont {Kostritsa}},
  \bibinfo {author} {\bibfnamefont {D.}~\bibnamefont {Landhuis}}, \bibinfo
  {author} {\bibfnamefont {P.}~\bibnamefont {Laptev}}, \bibinfo {author}
  {\bibfnamefont {E.}~\bibnamefont {Lucero}}, \bibinfo {author} {\bibfnamefont
  {O.}~\bibnamefont {Martin}}, \bibinfo {author} {\bibfnamefont {J.~R.}\
  \bibnamefont {McClean}}, \bibinfo {author} {\bibfnamefont {T.}~\bibnamefont
  {McCourt}}, \bibinfo {author} {\bibfnamefont {X.}~\bibnamefont {Mi}},
  \bibinfo {author} {\bibfnamefont {K.~C.}\ \bibnamefont {Miao}}, \bibinfo
  {author} {\bibfnamefont {M.}~\bibnamefont {Mohseni}}, \bibinfo {author}
  {\bibfnamefont {S.}~\bibnamefont {Montazeri}}, \bibinfo {author}
  {\bibfnamefont {W.}~\bibnamefont {Mruczkiewicz}}, \bibinfo {author}
  {\bibfnamefont {J.}~\bibnamefont {Mutus}}, \bibinfo {author} {\bibfnamefont
  {O.}~\bibnamefont {Naaman}}, \bibinfo {author} {\bibfnamefont
  {M.}~\bibnamefont {Neeley}}, \bibinfo {author} {\bibfnamefont
  {C.}~\bibnamefont {Neill}}, \bibinfo {author} {\bibfnamefont
  {M.}~\bibnamefont {Newman}}, \bibinfo {author} {\bibfnamefont {M.~Y.}\
  \bibnamefont {Niu}}, \bibinfo {author} {\bibfnamefont {T.~E.}\ \bibnamefont
  {O'Brien}}, \bibinfo {author} {\bibfnamefont {A.}~\bibnamefont {Opremcak}},
  \bibinfo {author} {\bibfnamefont {E.}~\bibnamefont {Ostby}}, \bibinfo
  {author} {\bibfnamefont {B.}~\bibnamefont {Pat{\'o}}}, \bibinfo {author}
  {\bibfnamefont {N.}~\bibnamefont {Redd}}, \bibinfo {author} {\bibfnamefont
  {P.}~\bibnamefont {Roushan}}, \bibinfo {author} {\bibfnamefont {N.~C.}\
  \bibnamefont {Rubin}}, \bibinfo {author} {\bibfnamefont {V.}~\bibnamefont
  {Shvarts}}, \bibinfo {author} {\bibfnamefont {D.}~\bibnamefont {Strain}},
  \bibinfo {author} {\bibfnamefont {M.}~\bibnamefont {Szalay}}, \bibinfo
  {author} {\bibfnamefont {M.~D.}\ \bibnamefont {Trevithick}}, \bibinfo
  {author} {\bibfnamefont {B.}~\bibnamefont {Villalonga}}, \bibinfo {author}
  {\bibfnamefont {T.}~\bibnamefont {White}}, \bibinfo {author} {\bibfnamefont
  {Z.~J.}\ \bibnamefont {Yao}}, \bibinfo {author} {\bibfnamefont
  {P.}~\bibnamefont {Yeh}}, \bibinfo {author} {\bibfnamefont {J.}~\bibnamefont
  {Yoo}}, \bibinfo {author} {\bibfnamefont {A.}~\bibnamefont {Zalcman}},
  \bibinfo {author} {\bibfnamefont {H.}~\bibnamefont {Neven}}, \bibinfo
  {author} {\bibfnamefont {S.}~\bibnamefont {Boixo}}, \bibinfo {author}
  {\bibfnamefont {V.}~\bibnamefont {Smelyanskiy}}, \bibinfo {author}
  {\bibfnamefont {Y.}~\bibnamefont {Chen}}, \bibinfo {author} {\bibfnamefont
  {A.}~\bibnamefont {Megrant}}, \bibinfo {author} {\bibfnamefont
  {J.}~\bibnamefont {Kelly}},\ and\ \bibinfo {author} {\bibnamefont {{Google
  Quantum AI}}},\ }\bibfield  {title} {\bibinfo {title} {Exponential
  suppression of bit or phase errors with cyclic error correction},\ }\href
  {https://doi.org/10.1038/s41586-021-03588-y} {\bibfield  {journal} {\bibinfo
  {journal} {Nature}\ }\textbf {\bibinfo {volume} {595}},\ \bibinfo {pages}
  {383} (\bibinfo {year} {2021})}\BibitemShut {NoStop}%
\bibitem [{\citenamefont {Dalzell}\ \emph {et~al.}(2021)\citenamefont
  {Dalzell}, \citenamefont {Hunter-Jones},\ and\ \citenamefont
  {Brandão}}]{dalzell2021random}%
  \BibitemOpen
  \bibfield  {author} {\bibinfo {author} {\bibfnamefont {A.~M.}\ \bibnamefont
  {Dalzell}}, \bibinfo {author} {\bibfnamefont {N.}~\bibnamefont
  {Hunter-Jones}},\ and\ \bibinfo {author} {\bibfnamefont {F.~G. S.~L.}\
  \bibnamefont {Brandão}},\ }\href@noop {} {\bibinfo {title} {Random quantum
  circuits transform local noise into global white noise}} (\bibinfo {year}
  {2021}),\ \Eprint {https://arxiv.org/abs/2111.14907} {arXiv:2111.14907}
  \BibitemShut {NoStop}%
\bibitem [{\citenamefont {Deshpande}\ \emph {et~al.}(2021)\citenamefont
  {Deshpande}, \citenamefont {Fefferman}, \citenamefont {Gorshkov},
  \citenamefont {Gullans}, \citenamefont {Niroula},\ and\ \citenamefont
  {Shtanko}}]{deshpande2021tight}%
  \BibitemOpen
  \bibfield  {author} {\bibinfo {author} {\bibfnamefont {A.}~\bibnamefont
  {Deshpande}}, \bibinfo {author} {\bibfnamefont {B.}~\bibnamefont
  {Fefferman}}, \bibinfo {author} {\bibfnamefont {A.~V.}\ \bibnamefont
  {Gorshkov}}, \bibinfo {author} {\bibfnamefont {M.~J.}\ \bibnamefont
  {Gullans}}, \bibinfo {author} {\bibfnamefont {P.}~\bibnamefont {Niroula}},\
  and\ \bibinfo {author} {\bibfnamefont {O.}~\bibnamefont {Shtanko}},\
  }\href@noop {} {\bibinfo {title} {Tight bounds on the convergence of noisy
  random circuits to uniform}} (\bibinfo {year} {2021}),\ \Eprint
  {https://arxiv.org/abs/2112.00716} {arXiv:2112.00716} \BibitemShut {NoStop}%
\bibitem [{\citenamefont {Huggins}\ \emph {et~al.}(2021)\citenamefont
  {Huggins}, \citenamefont {McArdle}, \citenamefont {O'Brien}, \citenamefont
  {Lee}, \citenamefont {Rubin}, \citenamefont {Boixo}, \citenamefont {Whaley},
  \citenamefont {Babbush},\ and\ \citenamefont {McClean}}]{huggins2021virtual}%
  \BibitemOpen
  \bibfield  {author} {\bibinfo {author} {\bibfnamefont {W.~J.}\ \bibnamefont
  {Huggins}}, \bibinfo {author} {\bibfnamefont {S.}~\bibnamefont {McArdle}},
  \bibinfo {author} {\bibfnamefont {T.~E.}\ \bibnamefont {O'Brien}}, \bibinfo
  {author} {\bibfnamefont {J.}~\bibnamefont {Lee}}, \bibinfo {author}
  {\bibfnamefont {N.~C.}\ \bibnamefont {Rubin}}, \bibinfo {author}
  {\bibfnamefont {S.}~\bibnamefont {Boixo}}, \bibinfo {author} {\bibfnamefont
  {K.~B.}\ \bibnamefont {Whaley}}, \bibinfo {author} {\bibfnamefont
  {R.}~\bibnamefont {Babbush}},\ and\ \bibinfo {author} {\bibfnamefont {J.~R.}\
  \bibnamefont {McClean}},\ }\bibfield  {title} {\bibinfo {title} {Virtual
  distillation for quantum error mitigation},\ }\href
  {https://doi.org/10.1103/PhysRevX.11.041036} {\bibfield  {journal} {\bibinfo
  {journal} {Phys. Rev. X}\ }\textbf {\bibinfo {volume} {11}},\ \bibinfo
  {pages} {041036} (\bibinfo {year} {2021})}\BibitemShut {NoStop}%
\bibitem [{\citenamefont {Koczor}(2021)}]{koczor2021exponential}%
  \BibitemOpen
  \bibfield  {author} {\bibinfo {author} {\bibfnamefont {B.}~\bibnamefont
  {Koczor}},\ }\bibfield  {title} {\bibinfo {title} {Exponential error
  suppression for near-term quantum devices},\ }\href
  {https://doi.org/10.1103/PhysRevX.11.031057} {\bibfield  {journal} {\bibinfo
  {journal} {Phys. Rev. X}\ }\textbf {\bibinfo {volume} {11}},\ \bibinfo
  {pages} {031057} (\bibinfo {year} {2021})}\BibitemShut {NoStop}%
\bibitem [{\citenamefont {McClean}\ \emph {et~al.}(2018)\citenamefont
  {McClean}, \citenamefont {Boixo}, \citenamefont {Smelyanskiy}, \citenamefont
  {Babbush},\ and\ \citenamefont {Neven}}]{mcclean2018barren}%
  \BibitemOpen
  \bibfield  {author} {\bibinfo {author} {\bibfnamefont {J.~R.}\ \bibnamefont
  {McClean}}, \bibinfo {author} {\bibfnamefont {S.}~\bibnamefont {Boixo}},
  \bibinfo {author} {\bibfnamefont {V.~N.}\ \bibnamefont {Smelyanskiy}},
  \bibinfo {author} {\bibfnamefont {R.}~\bibnamefont {Babbush}},\ and\ \bibinfo
  {author} {\bibfnamefont {H.}~\bibnamefont {Neven}},\ }\bibfield  {title}
  {\bibinfo {title} {Barren plateaus in quantum neural network training
  landscapes},\ }\href {https://doi.org/10.1038/s41467-018-07090-4} {\bibfield
  {journal} {\bibinfo  {journal} {Nature Communications}\ }\textbf {\bibinfo
  {volume} {9}},\ \bibinfo {pages} {4812} (\bibinfo {year} {2018})}\BibitemShut
  {NoStop}%
\bibitem [{\citenamefont {Cerezo}\ \emph
  {et~al.}(2021{\natexlab{b}})\citenamefont {Cerezo}, \citenamefont {Sone},
  \citenamefont {Volkoff}, \citenamefont {Cincio},\ and\ \citenamefont
  {Coles}}]{cerezo2021cost}%
  \BibitemOpen
  \bibfield  {author} {\bibinfo {author} {\bibfnamefont {M.}~\bibnamefont
  {Cerezo}}, \bibinfo {author} {\bibfnamefont {A.}~\bibnamefont {Sone}},
  \bibinfo {author} {\bibfnamefont {T.}~\bibnamefont {Volkoff}}, \bibinfo
  {author} {\bibfnamefont {L.}~\bibnamefont {Cincio}},\ and\ \bibinfo {author}
  {\bibfnamefont {P.~J.}\ \bibnamefont {Coles}},\ }\bibfield  {title} {\bibinfo
  {title} {Cost function dependent barren plateaus in shallow parametrized
  quantum circuits},\ }\href {https://doi.org/10.1038/s41467-021-21728-w}
  {\bibfield  {journal} {\bibinfo  {journal} {Nature Communications}\ }\textbf
  {\bibinfo {volume} {12}},\ \bibinfo {pages} {1791} (\bibinfo {year}
  {2021}{\natexlab{b}})}\BibitemShut {NoStop}%
\bibitem [{\citenamefont {Wang}\ \emph {et~al.}(2021)\citenamefont {Wang},
  \citenamefont {Fontana}, \citenamefont {Cerezo}, \citenamefont {Sharma},
  \citenamefont {Sone}, \citenamefont {Cincio},\ and\ \citenamefont
  {Coles}}]{wang2021noise}%
  \BibitemOpen
  \bibfield  {author} {\bibinfo {author} {\bibfnamefont {S.}~\bibnamefont
  {Wang}}, \bibinfo {author} {\bibfnamefont {E.}~\bibnamefont {Fontana}},
  \bibinfo {author} {\bibfnamefont {M.}~\bibnamefont {Cerezo}}, \bibinfo
  {author} {\bibfnamefont {K.}~\bibnamefont {Sharma}}, \bibinfo {author}
  {\bibfnamefont {A.}~\bibnamefont {Sone}}, \bibinfo {author} {\bibfnamefont
  {L.}~\bibnamefont {Cincio}},\ and\ \bibinfo {author} {\bibfnamefont {P.~J.}\
  \bibnamefont {Coles}},\ }\bibfield  {title} {\bibinfo {title} {Noise-induced
  barren plateaus in variational quantum algorithms},\ }\href
  {https://doi.org/10.1038/s41467-021-27045-6} {\bibfield  {journal} {\bibinfo
  {journal} {Nature Communications}\ }\textbf {\bibinfo {volume} {12}},\
  \bibinfo {pages} {6961} (\bibinfo {year} {2021})}\BibitemShut {NoStop}%
\bibitem [{\citenamefont {Patti}\ \emph {et~al.}(2021)\citenamefont {Patti},
  \citenamefont {Najafi}, \citenamefont {Gao},\ and\ \citenamefont
  {Yelin}}]{PhysRevResearch.3.033090}%
  \BibitemOpen
  \bibfield  {author} {\bibinfo {author} {\bibfnamefont {T.~L.}\ \bibnamefont
  {Patti}}, \bibinfo {author} {\bibfnamefont {K.}~\bibnamefont {Najafi}},
  \bibinfo {author} {\bibfnamefont {X.}~\bibnamefont {Gao}},\ and\ \bibinfo
  {author} {\bibfnamefont {S.~F.}\ \bibnamefont {Yelin}},\ }\bibfield  {title}
  {\bibinfo {title} {Entanglement devised barren plateau mitigation},\ }\href
  {https://doi.org/10.1103/PhysRevResearch.3.033090} {\bibfield  {journal}
  {\bibinfo  {journal} {Phys. Rev. Research}\ }\textbf {\bibinfo {volume}
  {3}},\ \bibinfo {pages} {033090} (\bibinfo {year} {2021})}\BibitemShut
  {NoStop}%
\bibitem [{\citenamefont {Sack}\ \emph {et~al.}(2022)\citenamefont {Sack},
  \citenamefont {Medina}, \citenamefont {Michailidis}, \citenamefont {Kueng},\
  and\ \citenamefont {Serbyn}}]{PRXQuantum.3.020365}%
  \BibitemOpen
  \bibfield  {author} {\bibinfo {author} {\bibfnamefont {S.~H.}\ \bibnamefont
  {Sack}}, \bibinfo {author} {\bibfnamefont {R.~A.}\ \bibnamefont {Medina}},
  \bibinfo {author} {\bibfnamefont {A.~A.}\ \bibnamefont {Michailidis}},
  \bibinfo {author} {\bibfnamefont {R.}~\bibnamefont {Kueng}},\ and\ \bibinfo
  {author} {\bibfnamefont {M.}~\bibnamefont {Serbyn}},\ }\bibfield  {title}
  {\bibinfo {title} {Avoiding barren plateaus using classical shadows},\ }\href
  {https://doi.org/10.1103/PRXQuantum.3.020365} {\bibfield  {journal} {\bibinfo
   {journal} {PRX Quantum}\ }\textbf {\bibinfo {volume} {3}},\ \bibinfo {pages}
  {020365} (\bibinfo {year} {2022})}\BibitemShut {NoStop}%
\bibitem [{ibm(2022)}]{ibmq_quito}%
  \BibitemOpen
  \href@noop {} {\bibfield  {journal} {\bibinfo  {journal} {5-qubit backend:
  IBM Q team, ``IBM Q 5 Quito backend specification V1.1.34". Retrieved from
  \url{https://quantum-computing.ibm.com}}\ } (\bibinfo {year}
  {2022})}\BibitemShut {NoStop}%
\bibitem [{\citenamefont {Grassl}\ \emph {et~al.}(2017)\citenamefont {Grassl},
  \citenamefont {Lu},\ and\ \citenamefont {Zeng}}]{8006823}%
  \BibitemOpen
  \bibfield  {author} {\bibinfo {author} {\bibfnamefont {M.}~\bibnamefont
  {Grassl}}, \bibinfo {author} {\bibfnamefont {S.}~\bibnamefont {Lu}},\ and\
  \bibinfo {author} {\bibfnamefont {B.}~\bibnamefont {Zeng}},\ }\bibfield
  {title} {\bibinfo {title} {Codes for simultaneous transmission of quantum and
  classical information},\ }in\ \href
  {https://doi.org/10.1109/ISIT.2017.8006823} {\emph {\bibinfo {booktitle}
  {2017 IEEE International Symposium on Information Theory (ISIT)}}}\ (\bibinfo
  {year} {2017})\ pp.\ \bibinfo {pages} {1718--1722}\BibitemShut {NoStop}%
\bibitem [{\citenamefont {Duan}(2009)}]{duan2009super}%
  \BibitemOpen
  \bibfield  {author} {\bibinfo {author} {\bibfnamefont {R.}~\bibnamefont
  {Duan}},\ }\href@noop {} {\bibinfo {title} {Super-activation of zero-error
  capacity of noisy quantum channels}} (\bibinfo {year} {2009}),\ \Eprint
  {https://arxiv.org/abs/0906.2527} {arXiv:0906.2527} \BibitemShut {NoStop}%
\bibitem [{\citenamefont {Yu}\ \emph {et~al.}(2021)\citenamefont {Yu},
  \citenamefont {Simnacher}, \citenamefont {Wyderka}, \citenamefont {Nguyen},\
  and\ \citenamefont {G{\"u}hne}}]{yu2021complete}%
  \BibitemOpen
  \bibfield  {author} {\bibinfo {author} {\bibfnamefont {X.-D.}\ \bibnamefont
  {Yu}}, \bibinfo {author} {\bibfnamefont {T.}~\bibnamefont {Simnacher}},
  \bibinfo {author} {\bibfnamefont {N.}~\bibnamefont {Wyderka}}, \bibinfo
  {author} {\bibfnamefont {H.~C.}\ \bibnamefont {Nguyen}},\ and\ \bibinfo
  {author} {\bibfnamefont {O.}~\bibnamefont {G{\"u}hne}},\ }\bibfield  {title}
  {\bibinfo {title} {A complete hierarchy for the pure state marginal problem
  in quantum mechanics},\ }\href {https://doi.org/10.1038/s41467-020-20799-5}
  {\bibfield  {journal} {\bibinfo  {journal} {Nature Communications}\ }\textbf
  {\bibinfo {volume} {12}},\ \bibinfo {pages} {1012} (\bibinfo {year}
  {2021})}\BibitemShut {NoStop}%
\bibitem [{\citenamefont {Or{\'u}s}(2019)}]{orus2019tensor}%
  \BibitemOpen
  \bibfield  {author} {\bibinfo {author} {\bibfnamefont {R.}~\bibnamefont
  {Or{\'u}s}},\ }\bibfield  {title} {\bibinfo {title} {Tensor networks for
  complex quantum systems},\ }\href {https://doi.org/10.1038/s42254-019-0086-7}
  {\bibfield  {journal} {\bibinfo  {journal} {Nature Reviews Physics}\ }\textbf
  {\bibinfo {volume} {1}},\ \bibinfo {pages} {538} (\bibinfo {year}
  {2019})}\BibitemShut {NoStop}%
\bibitem [{\citenamefont {Cirac}\ \emph {et~al.}(2021)\citenamefont {Cirac},
  \citenamefont {P\'erez-Garc\'{\i}a}, \citenamefont {Schuch},\ and\
  \citenamefont {Verstraete}}]{cirac2021matrix}%
  \BibitemOpen
  \bibfield  {author} {\bibinfo {author} {\bibfnamefont {J.~I.}\ \bibnamefont
  {Cirac}}, \bibinfo {author} {\bibfnamefont {D.}~\bibnamefont
  {P\'erez-Garc\'{\i}a}}, \bibinfo {author} {\bibfnamefont {N.}~\bibnamefont
  {Schuch}},\ and\ \bibinfo {author} {\bibfnamefont {F.}~\bibnamefont
  {Verstraete}},\ }\bibfield  {title} {\bibinfo {title} {Matrix product states
  and projected entangled pair states: Concepts, symmetries, theorems},\ }\href
  {https://doi.org/10.1103/RevModPhys.93.045003} {\bibfield  {journal}
  {\bibinfo  {journal} {Rev. Mod. Phys.}\ }\textbf {\bibinfo {volume} {93}},\
  \bibinfo {pages} {045003} (\bibinfo {year} {2021})}\BibitemShut {NoStop}%
\bibitem [{\citenamefont {Cheng}\ \emph {et~al.}(2021)\citenamefont {Cheng},
  \citenamefont {Cao}, \citenamefont {Zhang}, \citenamefont {Liu},
  \citenamefont {Hou}, \citenamefont {Xu},\ and\ \citenamefont
  {Zeng}}]{cheng2021simulating}%
  \BibitemOpen
  \bibfield  {author} {\bibinfo {author} {\bibfnamefont {S.}~\bibnamefont
  {Cheng}}, \bibinfo {author} {\bibfnamefont {C.}~\bibnamefont {Cao}}, \bibinfo
  {author} {\bibfnamefont {C.}~\bibnamefont {Zhang}}, \bibinfo {author}
  {\bibfnamefont {Y.}~\bibnamefont {Liu}}, \bibinfo {author} {\bibfnamefont
  {S.-Y.}\ \bibnamefont {Hou}}, \bibinfo {author} {\bibfnamefont
  {P.}~\bibnamefont {Xu}},\ and\ \bibinfo {author} {\bibfnamefont
  {B.}~\bibnamefont {Zeng}},\ }\bibfield  {title} {\bibinfo {title} {Simulating
  noisy quantum circuits with matrix product density operators},\ }\href
  {https://doi.org/10.1103/PhysRevResearch.3.023005} {\bibfield  {journal}
  {\bibinfo  {journal} {Phys. Rev. Research}\ }\textbf {\bibinfo {volume}
  {3}},\ \bibinfo {pages} {023005} (\bibinfo {year} {2021})}\BibitemShut
  {NoStop}%
\bibitem [{\citenamefont {Carleo}\ and\ \citenamefont
  {Troyer}(2017)}]{carleo2017solving}%
  \BibitemOpen
  \bibfield  {author} {\bibinfo {author} {\bibfnamefont {G.}~\bibnamefont
  {Carleo}}\ and\ \bibinfo {author} {\bibfnamefont {M.}~\bibnamefont
  {Troyer}},\ }\bibfield  {title} {\bibinfo {title} {Solving the quantum
  many-body problem with artificial neural networks},\ }\href
  {https://doi.org/10.1126/science.aag2302} {\bibfield  {journal} {\bibinfo
  {journal} {Science}\ }\textbf {\bibinfo {volume} {355}},\ \bibinfo {pages}
  {602} (\bibinfo {year} {2017})}\BibitemShut {NoStop}%
\bibitem [{\citenamefont {Helstrom}(1969)}]{helstrom1969quantum}%
  \BibitemOpen
  \bibfield  {author} {\bibinfo {author} {\bibfnamefont {C.~W.}\ \bibnamefont
  {Helstrom}},\ }\bibfield  {title} {\bibinfo {title} {Quantum detection and
  estimation theory},\ }\href {https://doi.org/10.1007/BF01007479} {\bibfield
  {journal} {\bibinfo  {journal} {Journal of Statistical Physics}\ }\textbf
  {\bibinfo {volume} {1}},\ \bibinfo {pages} {231} (\bibinfo {year}
  {1969})}\BibitemShut {NoStop}%
\bibitem [{\citenamefont {\ifmmode~\check{S}\else
  \v{S}\fi{}afr\'anek}(2018)}]{Simple_QFIM}%
  \BibitemOpen
  \bibfield  {author} {\bibinfo {author} {\bibfnamefont {D.}~\bibnamefont
  {\ifmmode~\check{S}\else \v{S}\fi{}afr\'anek}},\ }\bibfield  {title}
  {\bibinfo {title} {Simple expression for the quantum {F}isher information
  matrix},\ }\href {https://doi.org/10.1103/PhysRevA.97.042322} {\bibfield
  {journal} {\bibinfo  {journal} {Phys. Rev. A}\ }\textbf {\bibinfo {volume}
  {97}},\ \bibinfo {pages} {042322} (\bibinfo {year} {2018})}\BibitemShut
  {NoStop}%
\bibitem [{\citenamefont {Liu}\ \emph {et~al.}(2019)\citenamefont {Liu},
  \citenamefont {Yuan}, \citenamefont {Lu},\ and\ \citenamefont
  {Wang}}]{Liu_2019}%
  \BibitemOpen
  \bibfield  {author} {\bibinfo {author} {\bibfnamefont {J.}~\bibnamefont
  {Liu}}, \bibinfo {author} {\bibfnamefont {H.}~\bibnamefont {Yuan}}, \bibinfo
  {author} {\bibfnamefont {X.-M.}\ \bibnamefont {Lu}},\ and\ \bibinfo {author}
  {\bibfnamefont {X.}~\bibnamefont {Wang}},\ }\bibfield  {title} {\bibinfo
  {title} {Quantum fisher information matrix and multiparameter estimation},\
  }\href {https://doi.org/10.1088/1751-8121/ab5d4d} {\bibfield  {journal}
  {\bibinfo  {journal} {Journal of Physics A: Mathematical and Theoretical}\
  }\textbf {\bibinfo {volume} {53}},\ \bibinfo {pages} {023001} (\bibinfo
  {year} {2019})}\BibitemShut {NoStop}%
\bibitem [{\citenamefont {Meyer}(2021)}]{Meyer2021fisherinformationin}%
  \BibitemOpen
  \bibfield  {author} {\bibinfo {author} {\bibfnamefont {J.~J.}\ \bibnamefont
  {Meyer}},\ }\bibfield  {title} {\bibinfo {title} {Fisher {I}nformation in
  {N}oisy {I}ntermediate-{S}cale {Q}uantum {A}pplications},\ }\href
  {https://doi.org/10.22331/q-2021-09-09-539} {\bibfield  {journal} {\bibinfo
  {journal} {{Quantum}}\ }\textbf {\bibinfo {volume} {5}},\ \bibinfo {pages}
  {539} (\bibinfo {year} {2021})}\BibitemShut {NoStop}%
\bibitem [{\citenamefont {Milnor}\ and\ \citenamefont
  {Stasheff}(2016)}]{milnor2016characteristic}%
  \BibitemOpen
  \bibfield  {author} {\bibinfo {author} {\bibfnamefont {J.}~\bibnamefont
  {Milnor}}\ and\ \bibinfo {author} {\bibfnamefont {J.~D.}\ \bibnamefont
  {Stasheff}},\ }\href@noop {} {\emph {\bibinfo {title} {Characteristic
  Classes. Annals of Mathematics Studies, volume 76}}}\ (\bibinfo  {publisher}
  {Princeton {U}niversity {P}ress},\ \bibinfo {year} {2016})\BibitemShut
  {NoStop}%
\end{thebibliography}%
	\bibliographystyle{quantum}
	
	\onecolumn\vskip5mm\noindent\hrulefill\vskip5mm\twocolumn
	
	\appendix
	\section*{\LARGE \bf Appendices}
	\section{Proof of Proposition~\ref{prop3}}\label{ap:proof7}
	\begin{proof}
		From the completeness relation of the Kraus operators $\{E_{\alpha}\}$,
		\begin{equation}
			\begin{aligned}
				\sum_{\alpha=1}^m E_{\alpha}^{\dagger}E_{\alpha} = I,
			\end{aligned}
		\end{equation}
		we know 
		\begin{equation}
			\begin{aligned}
				\operatorname{Tr}(\sum_{\alpha=1}^m E_{\alpha}E_{\alpha}^{\dagger})=&\sum_{\alpha=1}^m\operatorname{Tr}(E_{\alpha}E_{\alpha}^{\dagger})\\=&\sum_{\alpha=1}^m\operatorname{Tr}(E_{\alpha}^{\dagger}E_{\alpha})\\=&\operatorname{Tr}(\sum_{\alpha=1}^m E_{\alpha}^{\dagger}E_{\alpha})\\=&2^n.
			\end{aligned}
		\end{equation}
		Each $E_{\alpha}E^{\dagger}_{\alpha}$ and $E^{\dagger}_{\alpha}E_{\alpha}$ are positive semidefinite. Denote the eigenvalues of $E_{\alpha}E_{\alpha}^{\dagger}$ as $\xi^{\alpha}_1 \geq  \xi^{\alpha}_2 \geq \dots \geq \xi^{\alpha}_{2^n} \geq 0$, the eigenvalues of $E_{\beta}E_{\beta}^{\dagger}$ as $\xi^{\beta}_1 \geq  \xi^{\beta}_2 \geq \dots \geq \xi^{\beta}_{2^n} \geq 0$, $\sum_{\alpha,j} \xi^{\alpha}_j = \sum_{\beta,j} \xi^{\beta}_j = 2^n$.
		
		Then we have
		\begin{equation}\label{trace_bound}
			\begin{aligned}
				&\operatorname{Tr}\Big(\sum_{\beta=1}^m \sum_{\alpha=1}^m (E^{\dagger}_{\alpha}E_{\beta})^{\dagger}E^{\dagger}_{\alpha}E_{\beta}\Big)\\
				=&\sum_{\beta=1}^m \sum_{\alpha=1}^m\operatorname{Tr} (E_{\alpha}E^{\dagger}_{\alpha}E_{\beta}E^{\dagger}_{\beta})\\
				\leq &\sum_{\beta=1}^m \sum_{\alpha=1}^m \sum_{j=1}^{2^n} \xi^{\alpha}_j \xi^{\beta}_j\\
				\leq & \sum_{\beta=1}^m \sum_{\alpha=1}^m \sum_{i=1}^{2^n}\sum_{j=1}^{2^n} \xi^{\alpha}_i \xi^{\beta}_j\\
				=& 2^{2n}.
			\end{aligned}
		\end{equation}
		The first inequality uses von Neumann's trace inequality.
		
		Each $E_{\alpha}$ non-trivially acts on no more than $\lfloor (d-1)/2 \rfloor$ qubits, therefore, each error product $E^{\dagger}_{\alpha}E_{\beta}$ non-trivially acts on no more than $(d-1)$ qubits. We expand $E^{\dagger}_{\alpha}E_{\beta}$ in the Pauli basis,
		\begin{equation}
			E^{\dagger}_{\alpha}E_{\beta} = \sum_{\gamma}\chi_{\gamma}^{\alpha\beta} O_{\gamma}^{\alpha\beta},
		\end{equation}
		where each $O_{\gamma}^{\alpha\beta}$ is a Pauli tensor product with weight less than $d$, $O_{\gamma}^{\alpha\beta\dagger}O_{\gamma}^{\alpha\beta} = I$, $|\{O_{\gamma}^{\alpha\beta}\}| \leq 4^{d-1}m^2$. Then
		\begin{equation}
			\begin{aligned}
				\operatorname{Tr}\left( (E_\alpha^\dagger E_\beta)^\dagger E_\alpha^\dagger E_\beta\right )=&\sum_\gamma \operatorname{Tr}\Big(|\chi_\gamma^{\alpha\beta}|^2 O_{\gamma}^{\alpha\beta\dagger}O_{\gamma}^{\alpha\beta}\Big) \\=& \sum_\gamma |\chi_\gamma^{\alpha\beta}|^2 \operatorname{Tr}(O_{\gamma}^{\alpha\beta\dagger}O_{\gamma}^{\alpha\beta})\\ =& 2^n \sum_\gamma |\chi_\gamma^{\alpha\beta}|^2,
			\end{aligned}
		\end{equation}
		\begin{equation}
			\operatorname{Tr}\Big(\sum_{\beta=1}^m \sum_{\alpha=1}^m (E^{\dagger}_{\alpha}E_{\beta})^{\dagger}E^{\dagger}_{\alpha}E_{\beta}\Big) = 2^n \sum_{\alpha, \beta, \gamma}|\chi_{\gamma}^{\alpha\beta}|^2.
		\end{equation}
		
		According to Eq.~\eqref{trace_bound}, we have
		\begin{equation}
			\sum_{\alpha, \beta, \gamma}|\chi_{\gamma}^{\alpha\beta}|^2 \leq 2^n.
		\end{equation}
		
		For the basis states $\{|\psi_1\rangle,|\psi_2\rangle,\dots,|\psi_K\rangle\}$, 
		\begin{equation}
			\begin{aligned}
				&\sum_{\alpha,\beta}\sum_{1 \leq i<j \leq K} \big|\langle\psi_i|E^{\dagger}_{\alpha}E_{\beta} |\psi_j\rangle\big|
				\\=  &\sum_{\alpha,\beta, \gamma}\sum_{1 \leq i<j \leq K} |\chi_{\gamma}^{\alpha\beta}| \big|\langle\psi_i|O_{\gamma}^{\alpha\beta} |\psi_j\rangle\big|
				\\\leq  &\sum_{\alpha,\beta, \gamma}|\chi_{\gamma}^{\alpha\beta}| \sum_{1 \leq i<j \leq K}\sum_{\operatorname{wt}(O_{\alpha'})<d}\big|\langle\psi_i|O_{\alpha'} |\psi_j\rangle\big|
				\\ \leq &2^{n/2 + d-1}m\sum_{1 \leq i<j \leq K}\sum_{\operatorname{wt}(O_{\alpha'})<d}\big|\langle\psi_i|O_{\alpha'} |\psi_j\rangle\big|.
			\end{aligned}
		\end{equation}
		
		Similarly, we have 
		\begin{equation}
			\begin{aligned}
				&\sum_{\alpha,\beta}\sum_{j=1}^K \big|\langle\psi_j|E^{\dagger}_{\alpha}E_{\beta}|\psi_j\rangle-\overline{\langle E^{\dagger}_{\alpha}E_{\beta} \rangle}\big|/2 
				\\\leq &2^{n/2 + d-1}m \sum_{\operatorname{wt}(O_{\alpha'})<d}\sum_{j=1}^K \big|\langle\psi_j|O_{\alpha'}|\psi_j\rangle-\overline{\langle O_{\alpha'} \rangle}\big|/2.
			\end{aligned}
		\end{equation}
		Denote \\ $\mathcal{E}'= \{E^{\dagger}_{\alpha}E_{\beta}| E_{\alpha}, E_{\beta} \text{ are Kraus operators of } \mathcal{N}\}$, we have
		\begin{equation}
			C^{\ell_{1}}_{n, K,\mathcal{E}'} \leq 2^{n/2 + d-1}m C^{\ell_{1}}_{n, K, \mathcal{E}}.
		\end{equation}
		According to Proposition~\ref{prop2}, the code is $\varepsilon$-correctable with $\varepsilon$ bounded by
		\begin{equation}
			\varepsilon \leq K\sqrt{2C^{\ell_{1}}_{n, K,\mathcal{E}'}} \leq 2^{n/4 + d/2}K\sqrt{m C^{\ell_{1}}_{n, K, \mathcal{E}}}.
		\end{equation}
	\end{proof}

	\section{Proof of Proposition~\ref{prop4}}\label{ap:proof8}
	\begin{proof}
		Since each $E_{\alpha}$ is proportional to a Pauli error, we have
		\begin{equation}
			E_{\alpha}^{\dagger}E_{\alpha} = E_{\alpha}E_{\alpha}^{\dagger}.
		\end{equation}
		Further, from the completeness relation of the Kraus operators $\{E_{\alpha}\}$,
		\begin{equation}
			\sum_{\alpha=1}^m E_{\alpha}^{\dagger}E_{\alpha} = I,
		\end{equation}
		we obtain the completeness relation of the error products, $\{E^{\dagger}_{\alpha}E_{\beta}\}$
		\begin{equation}\label{complete_ep8}
			\sum_{\beta=1}^m \sum_{\alpha=1}^m (E^{\dagger}_{\alpha}E_{\beta})^{\dagger}E^{\dagger}_{\alpha}E_{\beta} = \sum_{\beta=1}^m E_{\beta}^{\dagger}E_{\beta} = I.
		\end{equation}
		The $c_{\text{Z}}$-effective weight of each $E_{\alpha}$ smaller than $d_e(c_{\text{Z}})/2$, therefore, each error product $E^{\dagger}_{\alpha}E_{\beta}$ is proportional to a Pauli tensor product with $c_{\text{Z}}$-effective weight
		\begin{equation}
			\operatorname{wt}_e(E^{\dagger}_{\alpha}E_{\beta}, c_{\text{Z}}) < d_e(c_{\text{Z}}).
		\end{equation}
		Denote
		\begin{equation}
			E^{\dagger}_{\alpha}E_{\beta} = \chi^{\alpha\beta} O^{\alpha\beta},
		\end{equation}
		where $O^{\alpha\beta}$ is a Pauli tensor product. 
		
		According to Eq.~\eqref{complete_ep8}, we have the normalization condition
		\begin{equation}
			\sum_{\alpha, \beta}|\chi^{\alpha\beta}|^2 = 1.
		\end{equation}

		For the basis states $\{|\psi_1\rangle,|\psi_2\rangle,\dots,|\psi_K\rangle\}$, 
		\begin{equation}
			\begin{aligned}
				&\sum_{\alpha,\beta}\sum_{1 \leq i<j \leq K} \big|\langle\psi_i|E^{\dagger}_{\alpha}E_{\beta} |\psi_j\rangle\big|
				\\=  &\sum_{\alpha,\beta}\sum_{1 \leq i<j \leq K} |\chi^{\alpha\beta}| \big|\langle\psi_i|O^{\alpha\beta} |\psi_j\rangle\big|
				\\\leq  &\sum_{\alpha,\beta}|\chi^{\alpha\beta}| \sum_{1 \leq i<j \leq K}\sum_{\operatorname{wt}_e(O_{\alpha’}, c_{\text{Z}})<d_e(c_{\text{Z}})}\big|\langle\psi_i|O_{\alpha'} |\psi_j\rangle\big|
				\\ \leq &m\sum_{1 \leq i<j \leq K}\sum_{\operatorname{wt}_e(O_{\alpha’}, c_{\text{Z}})<d_e(c_{\text{Z}})}\big|\langle\psi_i|O_{\alpha'} |\psi_j\rangle\big|.
			\end{aligned}
		\end{equation}
		
		Similarly, we have 
		\begin{equation}
			\begin{aligned}
				&\sum_{\alpha,\beta}\sum_{j=1}^K \big|\langle\psi_j|E^{\dagger}_{\alpha}E_{\beta}|\psi_j\rangle-\overline{\langle E^{\dagger}_{\alpha}E_{\beta} \rangle}\big|/2 
				\\\leq &m \sum_{\operatorname{wt}_e(O_{\alpha’}, c_{\text{Z}})<d_e(c_{\text{Z}})}\sum_{j=1}^K \big|\langle\psi_j|O_{\alpha'}|\psi_j\rangle-\overline{\langle O_{\alpha'} \rangle}\big|/2.
			\end{aligned}
		\end{equation}
		Denote \\$\mathcal{E}'= \{E^{\dagger}_{\alpha}E_{\beta}| E_{\alpha}, E_{\beta} \text{ are Kraus operators of } \mathcal{N}\}$, we have
		\begin{equation}
			C^{\ell_{1}}_{n, K,\mathcal{E}'} \leq m C^{\ell_{1}}_{n, K, \mathcal{E}}.
		\end{equation}
		According to Proposition~\ref{prop2}, the code is $\varepsilon$-correctable with $\varepsilon$ bounded by
		\begin{equation}
			\varepsilon \leq K\sqrt{2C^{\ell_{1}}_{n, K,\mathcal{E}'}} \leq K\sqrt{2m C^{\ell_{1}}_{n, K, \mathcal{E}}}.
		\end{equation}
	\end{proof}

	\section{Parameter Dimension and Overparameterization}\label{ap:effect quantum dimension}
	
	The quantum Fisher information matrix (QFIM) is an essential concept in quantum metrology~\cite{helstrom1969quantum, Simple_QFIM, Liu_2019}. In recent years, its applications in NISQ algorithms and quantum machine learning have also been noticed~\cite{Meyer2021fisherinformationin, QFIM_1}. Ref.~\cite{QFIM_1} uses the QFIM to assess the expressive power of a VQC with the fixed input state $|0\rangle^{\otimes n}$. In VarQEC, however, we use $K$ orthogonal input states to find an $((n, K))_2$ quantum code. In this section, we generalize the notion of QFIM to multiple input states to quantify the expressive power of a VQC for preparing an $((n, K))_2$ quantum code. Based on that, we discuss the parameter dimension and the overparameterization of a VQC encoder.
	
	Suppose the VQC encoder has parameters
	\begin{equation}
		\boldsymbol{\theta} = (\theta_1, \theta_2, \dots, \theta_N).
	\end{equation}
	For a fixed pure input state, one relates the QFIM $\mathcal{F}(\boldsymbol{\theta})$ to the distance in the space of pure quantum states by
	\begin{equation}
		\operatorname{Dist}(|\psi(\boldsymbol{\theta})\rangle,|\psi(\boldsymbol{\theta}+d \boldsymbol{\theta})\rangle)^{2}=\sum_{l, m} \mathcal{F}_{l, m}(\boldsymbol{\theta}) d \theta_{l} d \theta_{m},
	\end{equation}
	where $\operatorname{Dist}(|\psi(\boldsymbol{\theta})\rangle, |\psi(\boldsymbol{\theta}')\rangle)=1-|\langle \psi(\boldsymbol{\theta})|\psi(\boldsymbol{\theta}')\rangle|^{2}$. The QFIM is an $N$ by $N$ matrix
	\begin{equation}
		\mathcal{F}_{lm}(\boldsymbol{\theta})=4\operatorname{Re}\left[\left\langle\partial_{l} \psi| \partial_{m} \psi\right\rangle-\left\langle\partial_{l} \psi| \psi\right\rangle\left\langle\psi | \partial_{m} \psi\right\rangle\right].
	\end{equation}
	where $|\partial_{l} \psi\rangle$ denotes $\partial|\psi(\boldsymbol{\theta})\rangle/\partial\theta_l$. In this case, the parameter dimension $D_c$ for a VQC is defined as the number of independent parameters that the VQC can express in the space of output states. Numerical evidence shows that $D_c$ is usually equivalent to the rank of QFIM for hardware-efficient VQCs with periodic and non-correlated random parameters $\boldsymbol{\theta}$~\cite{QFIM_1}.

	In VarQEC, the inputs are $K$ orthogonal pure states. Denote the projector onto the output space as $P_c$. We relate $\mathcal{F}(\boldsymbol{\theta})$ to the distance in the space of $K$-dimensional projectors,
	\begin{equation}
		\operatorname{Dist}_{K}(P_c(\boldsymbol{\theta}),P_c(\boldsymbol{\theta}+d \boldsymbol{\theta}))^{2}=\sum_{l, m} \mathcal{F}_{l, m}(\boldsymbol{\theta}) d \theta_{l} d \theta_{m},
	\end{equation}
	where the distance
	\begin{equation}
		\begin{aligned}
			\operatorname{Dist}_{K}(P_c(\boldsymbol{\theta}), &P_c(\boldsymbol{\theta}'))=\\&\Bigg(\frac{\operatorname{Tr} \sqrt{\sqrt{P_c(\boldsymbol{\theta})} P_c(\boldsymbol{\theta}') \sqrt{P_c(\boldsymbol{\theta})}}}{K}\Bigg)^{2}
		\end{aligned}
	\end{equation}
	is defined as the fidelity between the normalized mixed states of projectors $P_c(\boldsymbol{\theta})$ and $P_c(\boldsymbol{\theta}')$.
	Suppose the projector $P_c(\boldsymbol{\theta})$ has eigen decomposition 
	\begin{equation}
		P_c(\boldsymbol{\theta})=\sum_{j=1}^{K}|\psi_j\rangle\langle\psi_j|
	\end{equation}
	and denote the basis of its orthogonal complement as $\{|\psi_j\rangle\}_{j=K+1, K+2,\dots,2^n}$, the QFIM under our framework is of the form
	\begin{equation}
		\begin{aligned}
			&\mathcal{F}_{lm}(\boldsymbol{\theta})=2\frac{\partial^2}{\partial\delta_l\partial\delta_m}\operatorname{Dist}_{K}(P_c(\boldsymbol{\theta}), P_c(\boldsymbol{\theta+\delta}))|_{\boldsymbol{\delta = 0}} \\
			=&\frac{2}{K^2}\sum_{\min{\{i,j\}}\leq K} \frac{ \operatorname{Re}\left(\left\langle\psi_{i}\left|\partial_{l} P_c\right| \psi_{j}\right\rangle\left\langle\psi_{j}\left|\partial_{m} P_c\right| \psi_{i}\right\rangle\right)}{\langle\psi_i|P_c|\psi_i\rangle + \langle\psi_j|P_c|\psi_j\rangle}.
		\end{aligned}
	\end{equation}
	We remark that the derivation of QFIM for projectors is the same as for density matrices~\cite{Simple_QFIM}. Therefore, a similar formula can be used to compute the QFIM of a VQC with mixed inputs/outputs.
	
	Through sampling random parameters $\boldsymbol{\theta}$ from the interval $[0,2\pi)^{N}$ and computing the QFIM, we can estimate the parameter dimension $D_c$ by $\operatorname{rank}(\mathcal{F}(\boldsymbol{\theta}))$. The VarQEC algorithm searches a $D_c$-dimensional submanifold of the complex Grassmannian $\mathbf{G r}(K, 2^n)$.

	\begin{figure}[tb]
		\centering
		\includegraphics[width=4cm]{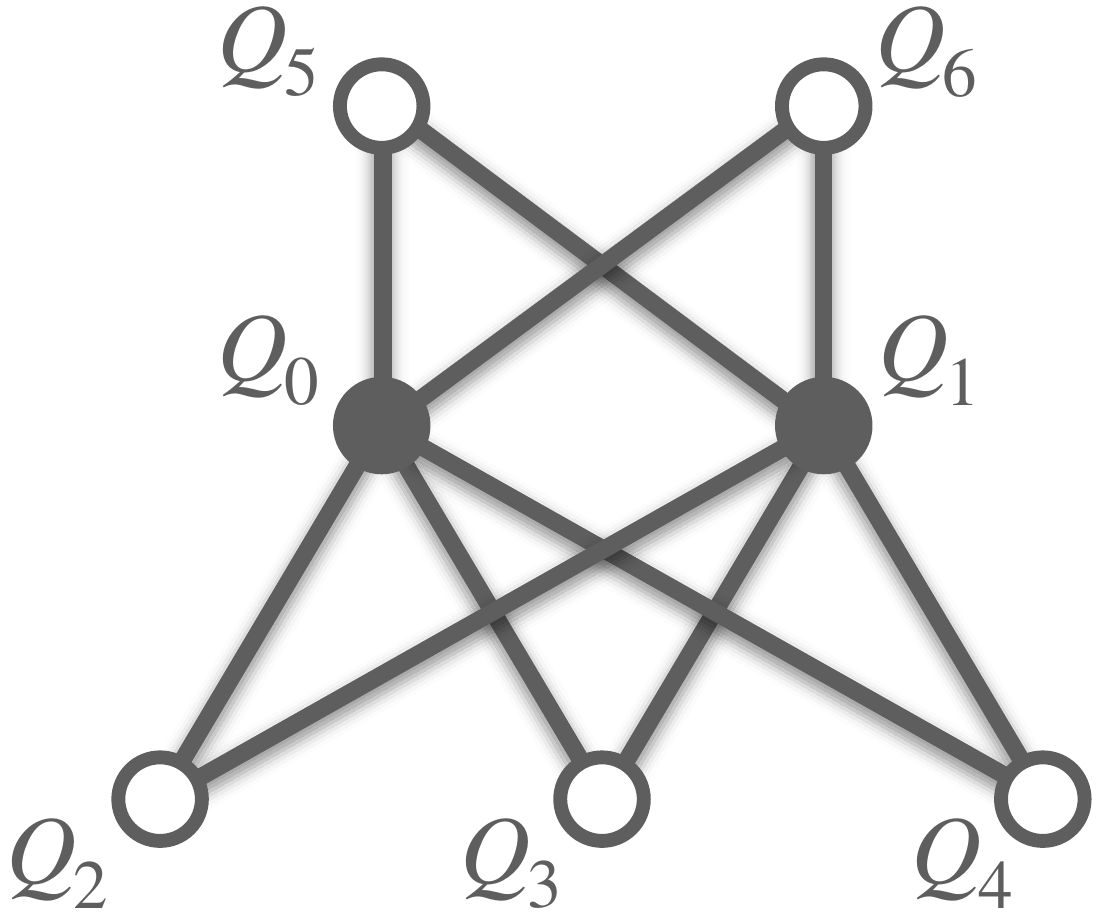}
		\caption{The bipartite connectivity graph for finding a quantum code with parameters $((7,3,3))_2$. Qubits $Q_0$ and $Q_1$ are selected to prepare the logical data.}
		\label{fig:733connectivity}
	\end{figure}
	
	\begin{figure}[t]
		\centering
		\includegraphics[width=7.5cm]{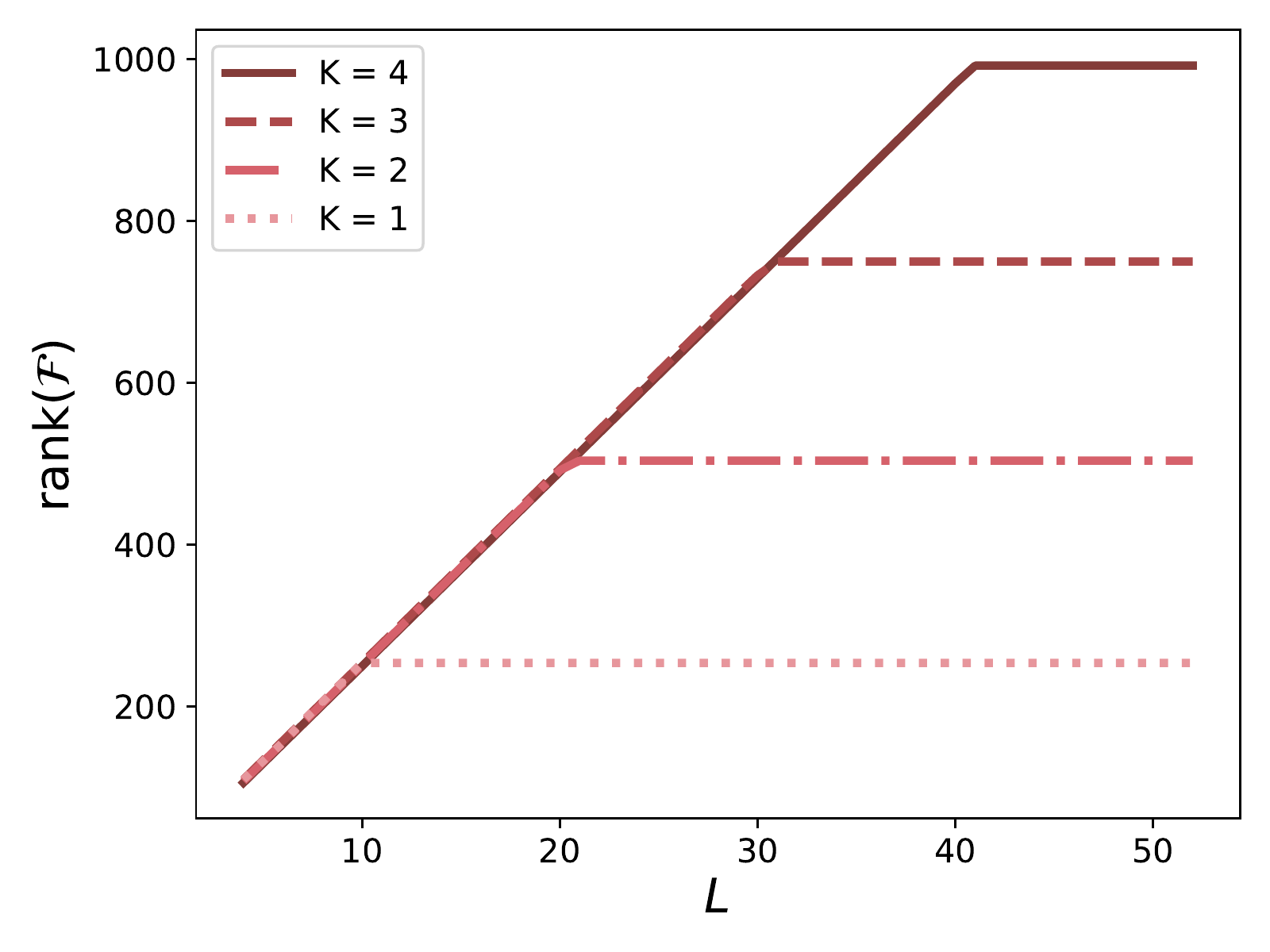}
		\caption{$\operatorname{Rank}(\mathcal{F}(\boldsymbol{\theta}))$ versus the number of VQC layers for the connectivity graph shown in Fig.~\ref{fig:733connectivity}. For code dimensions $K=1,2,3,4$, the maximum ranks are $D_c^{\mathrm{max}} = 254,504,750,992$, and the required numbers of VQC layers to achieve overparameterization are $L_{\mathrm{crit}} = 10,21,31,41$.}
		\label{fig:QFIM rank}
	\end{figure}
	
	Without loss of generality, we consider the connectivity graph shown in Fig.~\ref{fig:733connectivity}. For $K=1,2,3,4$, we randomly sample parameters $\boldsymbol{\theta}$ and plot $\operatorname{rank}(\mathcal{F}(\boldsymbol{\theta}))$ as a function of the number of VQC layers $L$ in Fig.~\ref{fig:QFIM rank}. Almost no parameterized gate is redundant when the circuit is underparameterized ($D_c/N \approx 1$). With the increase of $L$, $\operatorname{rank}(\mathcal{F}(\boldsymbol{\theta}))$ increases approximately linearly until achieving its maximum $D_c^{\mathrm{max}}$. The maximum parameter dimension for code length $n$ and code dimension $K$ is of the form
	\begin{equation}
		D_c^{\mathrm{max}} = 2K(2^n-K).
	\end{equation}
	This agrees with the fact that the dimension of the complex Grassmannian $\mathbf{G r}(K, 2^n)$ is $K(2^n-K)$~\cite{milnor2016characteristic}. When $D_c = D_c^{\mathrm{max}}$, the VQC can explore the whole $\mathbf{G r}(K, 2^n)$ manifold and prepare arbitrary $((n,K))_2$ quantum code. The required number of layers to saturate the maximum parameter dimension is approximately
	\begin{equation}
		L_{\text{crit}} = \Big\lceil \frac{2K(2^n-K) - 2n}{2n + |E(G)|} \Big\rceil,
	\end{equation}
	where $|E(G)|$ is the number of edges of the connectivity graph.

	To find a quantum code with parameters $((7,3,3))_2$, we sample $100$ different initial values of $\boldsymbol{\theta}$ and implement VarQEC with an overparameterized VQC ($L=31$). However, the cost function $C^{\ell_{1}}_{n,K,\mathcal{E}}(\boldsymbol{\theta})$ is always greater than 1. A $((7,3,3))_2$ code is improbable to exist.

	\section{A variational quantum encoder for additive codes}\label{ap:QFT-ansatz}
	The VQC with bipartite connectivity performs well in most cases. However, for some code parameters (e.g., $((10,4,4))_2$), it needs a large bunch of samples of the initial $\boldsymbol{\theta}$ to find an eligible code. Here we propose AC-VQC, another variational quantum circuit with all-to-all connectivity, to complement the bipartite ansatz. 
	
	The AC-VQC is especially resource-efficient in finding encoding circuits of additive codes. The structure of an AC-VQC is similar to the circuit of the quantum Fourier transform, as shown in Fig.~\ref{fig:QFT-VQC}. We start from two physical qubits ($Q_0,Q_1$) and apply a 2-qubit parameterized unitary operator $U_{01}$ to them. Then, we add another qubit ($Q_2$), apply 2-qubit parameterized unitary operators $U_{02}$/$U_{12}$ to the new one and each of the qubits that already exist ($Q_0$-$Q_2$, $Q_1$-$Q_2$). Repeat the steps until the system size equals $n$. In the end, we apply single-qubit rotations $R_z$ and $R_x$ to all qubits to explore the manifold of locally equivalent codes. The initial $k$ qubits prepare the logical data. The total circuit depth is of order $\mathcal{O}(\sum_{j=1}^{n-1} j) = \mathcal{O}(n^2)$.
	
	\begin{figure}[h]
		\centering
		\includegraphics[width=\columnwidth]{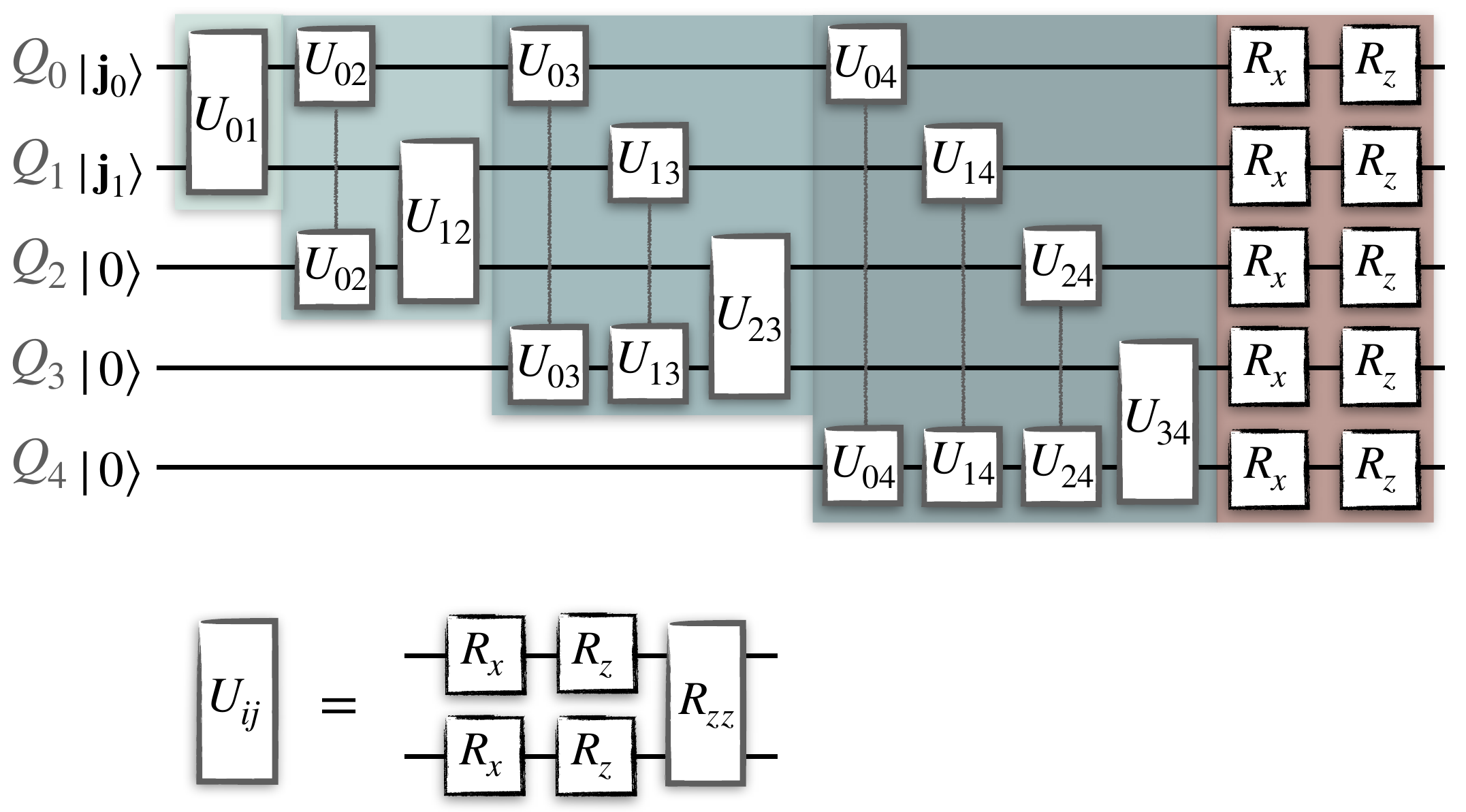}
		\caption{Schematic of AC-variational quantum circuit with $n=5$, $k=2$. Physical qubits are added layer by layer.}
		\label{fig:QFT-VQC}
	\end{figure}

	\section{Non-CWS Quantum Codes with Parameters $((6,2,3))_2$, $((7,2,3))_2$}\label{ap:623}
	The $((5,2,3))_2$ code is known to be unique. However, the classification and construction of $((6,2,3))_2$ and $((7,2,3))_2$ quantum codes are unclear. Some are said to be ``non-CWS'' since they are not locally equivalent to CWS codes. Here we present a general construction of non-CWS quantum codes based on stabilizer ones.
	
	\begin{theorem}\label{non_CWS}
		If there exists a quantum code $\mathcal{C}$ with parameters $((n, 2^k, d))_q$, then there exist degenerate $((n', 2^k, d))_q$ codes $\{\mathcal{C}'\}$ with $n'>n$ that are not locally equivalent to a CWS code.
	\end{theorem}
	\begin{proof}
		For $n'>n+1$, we can directly obtain non-CWS codes by taking the tensor product with a non-stabilizer state. 
		
		For $n' = n+1$, we take the tensor product of the code $\mathcal{C}$ with a fixed (stabilizer) state, then apply a non-Clifford entangling unitary operation to one of the original qudits and the additional qudit. This will conjugate the Pauli-stabilizers to non-local stabilizers. The resulting code is a non-CWS degenerate code with parameters $((n', 2^k, d))_q$.
	\end{proof}
	
	For example, we start from the $((5,2,3))_2$ code ($Q_0$,$Q_1$,$\dots$, $Q_4$) and initialize an addition qubit ($Q_5$) in state $|0\rangle$, then apply the transformation $U$ on $Q_5$, controlled by $Q_4$ with 
	\begin{equation} 
		U=\frac{1}{5}\left(\begin{array}{cc}
			3 & 4 \\
			-4 & 3
		\end{array}\right).
	\end{equation}
	The resulting code is a non-CWS degenerate $((6,2,3))_2$ code with basis states
	\begin{equation}
		\begin{aligned}
			|\psi_1\rangle = &\frac{\sqrt{2}}{20}(5|000000\rangle -5i|001100\rangle -3i|010010\rangle\\
			& +4i|010011\rangle + 3|011110\rangle - 4|011111\rangle\\
			& -3|100110\rangle + 4|100111\rangle -3i|101010\rangle\\
			& +4i|101011\rangle - 5i|110100\rangle - 5|111000\rangle,
		\end{aligned}
	\end{equation}
	\begin{equation}
		\begin{aligned}
			|\psi_2\rangle = &\frac{\sqrt{2}}{20}(3|00110\rangle  -4|00111\rangle -3i|01010\rangle\\
			& + 4i|01011\rangle -5i|10100\rangle + 5|11000\rangle\\
			& + 5|100000\rangle + 5i|101100\rangle + 3i|110010\rangle\\
			& -4i|110011\rangle + 3|111110\rangle - 4|111111\rangle.\mkern14mu
		\end{aligned}
	\end{equation}
	Its weight enumerators are
	\begin{equation}
		\begin{aligned}
			A(z) = &1 + \frac{9}{25}z + \frac{16}{25}z^2 + \frac{311}{25}z^4 + \frac{391}{25}z^5 +\frac{48}{25}z^6,^{}
		\end{aligned}
	\end{equation}
	\begin{equation}
		\begin{aligned}
			B(z) = &1 + \frac{9}{25}z + \frac{16}{25}z^2 + \frac{654}{25}z^3 + \frac{193}{5}z^4 \\&+ \frac{937}{25}z^5 + \frac{594}{25}z^6.
		\end{aligned}
	\end{equation}
	Likewise, non-CWS $((7,2,3))_2$ codes can be constructed based on $((6,2,3))_2$ stabilizer codes.
	
	Another class of non-CWS $((6,2,3))_2$ codes are unitarily related to the additive $((6,2,3))_2$ code stabilized by 
	\begin{equation}
		\begin{array}{llllllllllllll}
			g_{1} & = &Y&I&Z&X&X&Y\\
			g_{2} & = &Z&X&I&X&I&Z\\
			g_{3} & = &I&Z&X&X&X&X\\
			g_{4} & = &I&I&I&I&Z&Z\\
			g_{5} & = &Z&Z&Z&Z&I&I.
		\end{array}
	\end{equation}
	Here ``unitarily related’’ means they can be transformed to this code when allowing permutations of qubits and a unitary transformation of the form $U = \otimes_{j=1}^{5}U_j$ where $U_1$,$U_2$,$U_3$,$U_4$ are single-qubit unitaries and $U_5$ is a 2-qubit unitary. In our numerical experiments, all the $((6,2,3))_2$ codes discovered by VarQEC are unitarily related to this $((6,2,3))_2$ stabilizer code or the $((5,2,3))_2$ perfect code, some are related to both.
	
	Consider a quantum code $\mathcal{C}$ that is capable of correcting an error set $\mathcal{E}$. If linearly independent errors in $\mathcal{E}$ map $\mathcal{C}$ to linearly independent subspaces, we say $\mathcal{C}$ is \textit{non-degenerate} with respect to $\mathcal{E}$. If linearly independent errors in $\mathcal{E}$ map $\mathcal{C}$ to mutually orthogonal subspaces, we say $\mathcal{C}$ is \textit{pure} with respect to $\mathcal{E}$~\cite{calderbank1998quantum}. A pure code must be non-degenerate. For additive codes and CWS codes, ``non-degenerate’’ and ``pure’’ are equivalent~\cite{calderbank1998quantum}. However, we note that some of our $((7,2,3))_2$ codes are non-degenerate but impure. Here we give a ``trivial’’ construction of non-degenerate and impure $((7,2,3))_2$ codes. Still, we start with the $((5,2,3))_2$ code on qubits $Q_0,Q_1,\ldots,Q_4$ and add two additional qubits $Q_5$, $Q_6$ in a fixed state $|00\rangle$. Then we can apply either a three-qubit unitary on, e.g., $Q_4$, $Q_5$, $Q_6$, or two two-qubit unitaries on, e.g., $Q_3$, $Q_5$ and $Q_4$, $Q_6$. When these unitaries are non-Clifford, the resulting code is a non-CWS degenerate and impure code. We give a detailed example here. The following basis states span a non-degenerate but impure $((7,2,3))_2$ code:
        \begin{widetext}
	  \begin{equation}
            \begin{aligned}
			|\psi_1\rangle =&\frac{1}{20}(-4\omega^3|0000011\rangle + 3\omega^2|0000100\rangle
			-3\omega|0001110\rangle-4\omega|0001111\rangle+3\omega^3|0010110\rangle\\&
			+4\omega^3|0010111\rangle-4\omega|0011011\rangle+3|0011100\rangle
			+3\omega^3|0100110\rangle+4\omega^3|0100111\rangle-4\omega|0101011\rangle\\&
			+3\omega|0101100\rangle-4\omega^3|0110011\rangle+3\omega^2|0110100\rangle
			-3\omega|0111110\rangle-4\omega|0111111\rangle-4\omega^3|1000011\rangle\\&
			+3\omega^2|1000100\rangle+3\omega|1001110\rangle+4\omega|1001111\rangle
			+3\omega^3|1010110\rangle-4\omega^3|1010111\rangle-4\omega|1011011\rangle\\&
			-3|1011100\rangle-3\omega^3|1100110\rangle-4\omega^3|1100111\rangle\rangle
			-4\omega|1101011\rangle+3|1101100\rangle+4\omega^3|1110011\rangle\\&
			-3\omega^2|1110100\rangle-3\omega|1111110\rangle-4\omega|1111111\rangle),
	    \end{aligned}
          \end{equation}
	\begin{equation}
		\begin{aligned}
			|\psi_2\rangle =& \frac{1}{20}(4\omega|0000011\rangle - 3|0000100\rangle
			+3\omega^3|0001110\rangle+4\omega^3|0001111\rangle+3\omega|0010110\rangle
			+4\omega|0010111\rangle\\&-4\omega^3|0011011\rangle+3\omega^2|0011100\rangle
			-3\omega|0100110\rangle-4\omega|0100111\rangle+4\omega^3|0101011\rangle
			-3\omega^2|0101100\rangle\\&-4\omega|0110011\rangle+3|0110100\rangle
			-3\omega^3|0111110\rangle-4\omega^3|0111111\rangle-4\omega|1000011\rangle
			+3|1000100\rangle\\&+3\omega^3|1001110\rangle+4\omega^3|1001111\rangle
			-3\omega|1010110\rangle-4\omega|1010111\rangle-4\omega^3|1011011\rangle
			+3\omega^2|1011100\rangle\\&-3\omega|1100110\rangle-4\omega|1100111\rangle
			-4\omega^3|1101011\rangle+3\omega^2|1101100\rangle-4\omega|1110011\rangle
			+3|1110100\rangle\\&+3\omega^3|1111110\rangle-4\omega^3|1111111\rangle),
		\end{aligned}
	\end{equation}
	\end{widetext}
        where $\omega = \exp(i\pi/4)$. Its weight enumerators are
	\begin{equation}
		\begin{aligned}
			A(z) = &1 + \frac{106}{125}z + z^2 +  \frac{144}{125}z^3 +  \frac{1299}{125}z^4\\& + \frac{3318}{125}z^5 + \frac{2451}{125}z^6 + \frac{432}{125}z^7,\mkern50mu
		\end{aligned}
	\end{equation}
	\begin{equation}
		\begin{aligned}
			B(z) = &1 + \frac{106}{125}z + z^2 + + \frac{606}{25}z^3 + \frac{7071}{125}z^4\\& + \frac{9318}{125}z^5 + \frac{8679}{125}z^6 + \frac{3546}{125}z^7.
		\end{aligned}
	\end{equation}

	\section{Quantum Weight Enumerators}\label{ap:weight enumerators}
\abovedisplayskip0.2\abovedisplayskip
\belowdisplayskip0.2\belowdisplayskip
	This section lists the quantum weight enumerators of some QECCs discovered/rediscovered by VarQEC. 
	\subsection{Symmetric codes}
	\begin{flalign}&
		A^{\{5,6,2\}}(z) = 1 + 1.667z^4 + 2.667z^5,
		&\end{flalign}
	\begin{flalign}&
		B^{\{5,6,2\}}(z) = 1 + 20z^2 + 50z^3 + 75z^4 + 46z^5;
		&\end{flalign}
	\begin{flalign}&
		A^{\{5,2,3\}}(z) = 1 + 15z^4,
		&\end{flalign}
	\begin{flalign}&
		B^{\{5,2,3\}}(z) = 1 + 30z^3 + 15z^4 + 18z^5;
		&\end{flalign}
	
	\begin{flalign}&
		\begin{aligned}
			A^{\{6,2,3\}}(z) = &1 + 0.267z + 0.732z^2 + 12.070z^4 \\&+ 15.732z^5 + 2.197z^6,
		\end{aligned}
		&\end{flalign}
	\begin{flalign}&
		\begin{aligned}
			B^{\{6,2,3\}}(z) = &1 + 0.267z + 0.732z^2 + 25.605z^3 \\&+ 37.676z^4 + 38.126z^5 + 24.591z^6;
		\end{aligned}
		&\end{flalign}
	
	\begin{flalign}&
		A^{\{7,8,2\}}(z) = 1 + 5z^4 + 2z^5 + 2z^6 + 6z^7,
		&\end{flalign}
	\begin{flalign}&
		\begin{aligned}
			B^{\{7,8,2\}}(z) = &1 + 17z^2 + 40z^3 + 195z^4 + 328z^5\\&+ 299z^6 + 144z^7;
		\end{aligned}
		&\end{flalign}
	
	Non-degenerate $((7,2,3))_2$:
	\begin{flalign}&
		\begin{aligned}
			A^{\{7,2,3\}}(z) = &1 + 1.437z^2 + 18.125z^4 + 43.437z^6,
		\end{aligned}
		&\end{flalign}
	\begin{flalign}&
		\begin{aligned}
			B^{\{7,2,3\}}(z) = &1 + 1.437z^2 + 25.311z^3 + 18.125z^4 \\& + 117.377z^5 + 43.437z^6 + 49.311z^7;
		\end{aligned}
		&\end{flalign}
	
	Degenerate $((7,2,3))_2$:
	\begin{flalign}&
		\begin{aligned}
			A^{\{7,2,3\}}(z) = 1 + 5z^2 + 11z^4 + 47z^6,
		\end{aligned}
		&\end{flalign}
	\begin{flalign}&
		\begin{aligned}
			B^{\{7,2,3\}}(z) = &1 + 5z^2 + 36z^3 + 11z^4 + 96z^5\\&+47z^6+ 60z^7;
		\end{aligned}
		&\end{flalign}

	\begin{flalign}&
		\begin{aligned}
			A^{\{8,2,3\}}(z) = &1 + 0.015z^3 + 13.924z^4 + 24.091z^5\\& + 40.030z^6 + 39.893z^7 + 9.046z^8,
		\end{aligned}
		&\end{flalign}
	\begin{flalign}&
		\begin{aligned}
			B^{\{8,2,3\}}(z) = &1 + 11.970z^3 + 38.152z^4\\& + 119.817z^5 + 159.939z^6\\& + 124.213z^7 + 56.909z^8;
		\end{aligned}
		&\end{flalign}

	\begin{flalign}&
		A^{\{8,8,3\}}(z) = 1 + 28z^6 + 3z^8,
		&\end{flalign}
	\begin{flalign}&
		\begin{aligned}
			B^{\{8,8,3\}}(z) = &1 + 56z^3 + 210z^4 + 336z^5 + 728z^6\\& + 504z^7 + 213z^8;
		\end{aligned}
		&\end{flalign}
	
	\begin{flalign}&
		\begin{aligned}
			A^{\{9,8,3\}}(z) = &1 + 0.042z^4 + 5.875z^5 + 16.083z^6\\& + 24.083z^7 + 14.875z^8 + 2.042z^9,
		\end{aligned}
		&\end{flalign}
	\begin{flalign}&
		\begin{aligned}
			B^{\{9,8,3\}}(z) = &1 + 40z^3 + 162.332z^4 + 479.004z^5\\&  + 952.664z^6+ 1224.664z^7\\&  + 932.004z^8 + 304.332z^9;
		\end{aligned}
		&\end{flalign}
	
	\begin{flalign}&
		\begin{aligned}
			A^{\{10,16,3\}}(z) = &1 + 0.127z^5 + 8.525z^6 + 20.443z^7\\& + 21.253z^8 + 11.430z^9 + 1.221z^{10},
		\end{aligned}
		&\end{flalign}
	\begin{flalign}&
		\begin{aligned}
			B^{\{10,16,3\}}(z) = &1 + 55.810z^3 + 275.961z^4\\& + 954.241z^5  + 2366.014z^6\\&+ 4120.948z^7  + 4622.227z^8\\& + 3061.001z^9+926.797z^{10};
		\end{aligned}
		&\end{flalign}

	\begin{flalign}&
		A^{\{10,4,4\}}(z) = 1 + 90z^6 + 135z^8 + 30z^{10},
		&\end{flalign}
	\begin{flalign}&
		\begin{aligned}
			B^{\{10,4,4\}}(z) = &1 + 90z^4 + 216z^5 + 720z^6 + 720z^7\\& + 1485z^8 + 600z^9 + 264z^{10}.
		\end{aligned}
		&\end{flalign}

	\begin{flalign}&
		\begin{aligned}
			A^{\{11,32,3\}}(z) = &1 + 0.160z^6 + 12.522z^7 + 23.316z^8\\& + 18.319z^9 + 7.524z^{10} + 1.158z^{11}
		\end{aligned}
		&\end{flalign}
	\begin{flalign}&
		\begin{aligned}
			B^{\{11,32,3\}}(z) = &1 + 75.003z^3 + 443.260z^4\\& + 1729.654z^5  + 5219.456z^6\\&+ 11343.613z^7  + 16918.654z^8\\& + 16859.456z^9+10185.630z^{10}\\& + 2760.274z^{11}
		\end{aligned}
		&\end{flalign}

	\begin{flalign}&
		\begin{aligned}
			A^{\{12,64,3\}}(z) = &1 + 2z^7 + 15z^8 + 24z^9 + 16z^{10}\\& + 6z^{11},
		\end{aligned}
		&\end{flalign}
	\begin{flalign}&
		\begin{aligned}
			B^{\{12,64,3\}}(z) = &1 + 104z^3 + 649z^4 + 2976z^5\\& + 10472z^6+ 27184z^7  + 50691z^8\\& + 67616z^9+60952z^{10} + 33192z^{11}\\& + 8307z^{12}.
		\end{aligned}
		&\end{flalign}

	\begin{flalign}&
		\begin{aligned}
			A^{\{13,128,3\}}(z) = &1 + 5z^8 + 16z^9 + 24z^{10} + 16z^{11}\\& + 2z^{12},
		\end{aligned}
		&\end{flalign}
	\begin{flalign}&
		\begin{aligned}
			B^{\{13,128,3\}}(z) = &1 + 138z^3 + 929z^4 + 4814z^5\\& + 19592z^6+ 58628z^7  + 131987z^8\\& + 219836z^9+263864z^{10}\\& + 215954z^{11} + 107925z^{12}\\&+24918z^{13} .
		\end{aligned}
		&\end{flalign}

	\begin{flalign}&
		\begin{aligned}
			A^{\{14,256,3\}}(z) = &1 + 2z^8 + 4z^9 + 18z^{10} + 28z^{11}\\& + 11z^{12},
		\end{aligned}
		&\end{flalign}
	\begin{flalign}&
		\begin{aligned}
			B^{\{14,256,3\}}(z) = &1 + 180z^3 + 1295z^4 + 7436z^5\\& + 34418z^6+ 117320z^7\\&  + 307391z^8 + 616280z^9\\&+923372z^{10} + 1007300z^{11}\\& + 755921z^{12}+348636z^{13}\\& + 74754z^{14} .
		\end{aligned}
		&\end{flalign}

	\subsection{Asymmetric codes}
	\begin{flalign}&
		\begin{aligned}
			A^{\{6,2,d_e(\frac{1}{2})=2\}}(z) = 1 + 6z^2 + 9z^4 + 16z^6,
		\end{aligned}
		&\end{flalign}
	\begin{flalign}&
		\begin{aligned}
			B^{\{6,2,d_e(\frac{1}{2})=2\}}(z) = &1 + 15z^2 + 8z^3 + 39z^4\\& + 24z^5 + 41z^6;
		\end{aligned}
		&\end{flalign}
	
	\begin{flalign}&
		\begin{aligned}
			A^{\{7,3,d_e(\frac{1}{2})=2\}}(z) = &1 + 1.111z^2 + 2.667z^3\\& + 4.778z^4 + 13.333z^5\\& + 15.444z^6 + 53.333z^7,
		\end{aligned}
		&\end{flalign}
	\begin{flalign}&
		\begin{aligned}
			B^{\{7,3,d_e(\frac{1}{2})=2\}}(z) = &1 + 3.667z^2 + 24z^3 + 61.667z^4 \\&+ 120z^5 + 125.667z^6 + 48z^7;
		\end{aligned}
		&\end{flalign}
	
	\begin{flalign}&
		\begin{aligned}
			A^{\{5,2,d_e(2)=3\}}(z) = 1 + 4z^2 + 3z^4 + 8z^5,
		\end{aligned}
		&\end{flalign}
	\begin{flalign}&
		\begin{aligned}
			B^{\{5,2,d_e(2)=3\}}(z) = &1 + 12z^2 + 10z^3 + 19z^4\\& + 22z^5;
		\end{aligned}
		&\end{flalign}
	
	\begin{flalign}&
		\begin{aligned}
			A^{\{6,4,d_e(2)=3\}}(z) = 1 + z^3 + 4z^4 + 7z^5 + 3z^6,
		\end{aligned}
		&\end{flalign}
	\begin{flalign}&
		\begin{aligned}
			B^{\{6,4,d_e(2)=3\}}(z) = &1 + 7z^2 +36z^3+75z^4+92z^5\\& + 45z^6;
		\end{aligned}
		&\end{flalign}
	
	\begin{flalign}&
		\begin{aligned}
			A^{\{7,8,d_e(2)=3\}}(z) = 1 + z^4 +6z^5+ 6z^6 + 2z^7,
		\end{aligned}
		&\end{flalign}
	\begin{flalign}&
		\begin{aligned}
			B^{\{7,8,d_e(2)=3\}}(z) = &1 + 9z^2 + 64z^3 + 179z^4\\& + 312z^5+ 323z^6 + 136z^7;
		\end{aligned}
		&\end{flalign}

	\begin{flalign}&
		\begin{aligned}
			A^{\{6,2,d_e(2)=4\}}(z) = 1 + z^2 + 11z^4 + 16z^5 + 3z^6,
		\end{aligned}
		&\end{flalign}
	\begin{flalign}&
		\begin{aligned}
			B^{\{6,2,d_e(2)=4\}}(z) = &1 + z^2 + 24z^3 + 35z^4 + 40z^5\\& + 27z^6;
		\end{aligned}
		&\end{flalign}

	\begin{flalign}&
		\begin{aligned}
			A^{\{8,3,d_e(2)=4\}}(z) = &1 + 0.224z^2 + 1.770z^3\\& + 4.919z^4 + 17.794z^5\\&  + 28.155z^6 + 23.103z^7\\& + 8.368z^8,
		\end{aligned}
		&\end{flalign}
	\begin{flalign}&
		\begin{aligned}
			B^{\{8,3,d_e(2)=4\}}(z) = &1 + 2.021z^2 + 18.612z^3\\& + 68.016z^4+ 154.775z^5\\& + 237.904z^6+ 210.612z^7\\& + 75.058z^8.
		\end{aligned}
		&\end{flalign}

	\subsection{Channel-adaptive codes}
	\subsubsection{Nearest-neighbor collective amplitude damping}
	
	This part lists the quantum weight enumerators of the channel-adaptive codes for the one-dimensional nearest-neighbor collective amplitude damping errors discussed in Sec.~\ref{Sec:collective-ad}.
	
	\begin{flalign}&
		\begin{aligned}
			A^{\{4,3\}}(z) = &1 + 0.111z + 1.222z^2 + 0.778z^3\\& + 2.222z^4,
		\end{aligned}
		&\end{flalign}
	\begin{flalign}&
		\begin{aligned}
			B^{\{4,3\}}(z) = &1 + 1.667z + 11.667z^2 + 17z^3 \\&+ 16.667z^4;
		\end{aligned}
		&\end{flalign}
	
	\begin{flalign}&
		\begin{aligned}
			A^{\{5,2\}}(z) = 1 + z + 2z^2 + 2z^3 + 5z^4 + 5z^5,
		\end{aligned}
		&\end{flalign}
	\begin{flalign}&
		\begin{aligned}
			B^{\{5,2\}}(z) = 1 + z + 10z^2 + 18z^3 + 21z^4 + 13z^5;
		\end{aligned}
		&\end{flalign}

	\begin{flalign}&
		\begin{aligned}
			A^{\{6,5\}}(z) = &1 + 1.307z^2 + 0.032z^3 + 3.890z^4 \\&+ 0.097z^5  + 6.474z^6,
		\end{aligned}
		&\end{flalign}
	\begin{flalign}&
		\begin{aligned}
			B^{\{6,5\}}(z) = &1 + 18.533z^2 + 32.162z^3 + 103.449z^4 \\&+ 96.485z^5 + 68.371z^6;
		\end{aligned}
		&\end{flalign}

	\begin{flalign}&
		\begin{aligned}
			A^{\{7,8\}}(z) = 1 + z^4 + 6z^5 + 6z^6 + 2z^7,
		\end{aligned}
		&\end{flalign}
	\begin{flalign}&
		\begin{aligned}
			B^{\{7,8\}}(z) = &1 + 9z^2 + 64z^3 + 179z^4 + 312z^5 \\& + 323z^6 + 136z^7;
		\end{aligned}
		&\end{flalign}
	
	\begin{flalign}&
		\begin{aligned}
			A^{\{8,9\}}(z) = &1 + 0.038z^2 + 0.124z^3 + 0.827z^4 \\& + 5.282z^5 + 10.035z^6 + 8.816z^7\\& + 2.322z^8,
		\end{aligned}
		&\end{flalign}
	\begin{flalign}&
		\begin{aligned}
			B^{\{8,9\}}(z) = &1 + 6.119z^2 + 57.119z^3 + 200.778z^4\\& + 475.985z^5 + 713.867z^6 + 618.896z^7 \\&+ 230.237z^8;
		\end{aligned}
		&\end{flalign}

	\begin{flalign}&
		\begin{aligned}
			A^{\{9,16\}}(z) = &1 + z^5 + 8z^6 + 14z^7 + 7z^8 + z^9,
		\end{aligned}
		&\end{flalign}
	\begin{flalign}&
		\begin{aligned}
			B^{\{9,16\}}(z) = &1 + 4z^2 + 80z^3 + 326z^4 + 936z^5 \\&+ 1924z^6 + 2464z^7 + 1841z^8\\& + 616z^9.
		\end{aligned}
		&\end{flalign}

	\subsubsection{Nearest-neighbor collective phase-flips}

	This part lists the quantum weight enumerators of the channel-adaptive codes for the combined noise channel $\mathcal{N}$ (Eq.~\eqref{overall_N}) with hardware connectivity graphs shown in Fig.~\ref{fig:code_str}. 

	\begin{flalign}
		&A^{\{a\}}(z) = 1 + z^2 + 11z^4 + 16z^5+ 3z^6, &                                  
	\end{flalign}

	\begin{flalign}&
		B^{\{a\}}(z) = 1 + z^2 + 24z^3 + 35z^4 + 40z^5 + 27z^6;
		&\end{flalign}
	\begin{flalign}&
		A^{\{b\}}(z) = 1 + z^2 + 11z^4 + 16z^5 + 3z^6,
		&\end{flalign}
	\begin{flalign}&
		B^{\{b\}}(z) = 1 + z^2 + 24z^3 + 35z^4 + 40z^5 + 27z^6;
		&\end{flalign}
	\begin{flalign}&
		A^{\{c\}}(z) = 1 + 2z^3 + 9z^4 + 24z^5+22z^6+6z^7,
		&\end{flalign}
	\begin{flalign}&
		B^{\{c\}}(z) = 1 + 17z^3 + 45z^4 + 78z^5+82z^6+33z^7;
		&\end{flalign}
	\begin{flalign}&
		A^{\{d\}}(z) = 1 + 2z^3 + 9z^4 + 24z^5+22z^6+6z^7,
		&\end{flalign}
	\begin{flalign}&
		B^{\{d\}}(z) = 1 + 17z^3 + 45z^4 + 78z^5+ 82z^6+ 33z^7;
		&\end{flalign}
	\begin{flalign}
		&\begin{aligned}
			A^{\{e\}}(z) = &1 + 3.944z^4+12.112z^5+ 24z^6\\&+ 19.888z^7 +3.056z^8,
		\end{aligned}&
	\end{flalign}
	\begin{flalign}&
		\begin{aligned}
			B^{\{e\}}(z) = &1 + 27.888z^3+ 86.336z^4+ 215.776z^5\\&+ 319.776z^6 + 268.336z^7 + 104.888z^8;
		\end{aligned}
		&\end{flalign}
	
	\begin{flalign}&
		\begin{aligned}
			A^{\{f\}}(z) = &1 + 3.983z^4+ 12.033z^5+ 24z^6\\&+ 19.967z^7+ 3.017z^8,
		\end{aligned}
		&\end{flalign}
	
	\begin{flalign}&
		\begin{aligned}
			B^{\{f\}}(z) = &1 +27.967z^3+ 86.100z^4+215.933z^5\\&+319.933z^6+268.100z^7+ 104.967z^8;
		\end{aligned}
		&\end{flalign}
	
	\begin{flalign}&
		\begin{aligned}
			A^{\{g\}}(z) = &1 + 0.886z^3 + 3.282z^4 + 11.604z^5 \\&+ 30.352z^6+ 45.220z^7 + 29.366z^8 \\&+ 6.290z^9,
		\end{aligned}
		&\end{flalign}
	\begin{flalign}&
		\begin{aligned}
			B^{\{g\}}(z) = &1 +21.648z^3 + 78.704z^4 + 232.520z^5 \\&+ 482.255z^6+ 618.352z^7+ 462.041z^8\\&+151.480z^9;
		\end{aligned}
		&\end{flalign}
	\begin{flalign}&
		\begin{aligned}
			A^{\{h\}}(z) = &1 + 0.004z^3 + 3.996z^4 + 10.009z^5\\& + 35.974z^6+44.004z^7 + 23.030z^8\\& + 9.983z^9,
		\end{aligned}
		&\end{flalign}
	\begin{flalign}&
		\begin{aligned}
			B^{\{h\}}(z) = &1 +16.026z^3+ 89.948z^4+ 224.044z^5\\&+ 487.965z^6+ 623.974z^7 + 445.087z^8  \\&+ 159.956z^9.
		\end{aligned}
		&\end{flalign}
	
\end{document}